\RequirePackage{amsmath}
\documentclass[runningheads]{llncs}

\newif\ifshowcomments
\showcommentstrue
\showcommentsfalse

\usepackage[T1]{fontenc}
\usepackage{graphicx}
\usepackage{comment}

\usepackage{bm}

\usepackage[acronym]{glossaries}
\usepackage[acronym,nomain,shortcuts=ac]{glossaries-extra}
\setabbreviationstyle[acronym]{long-short}
\makeglossaries 

\usepackage[lambda, advantage, operators, sets, adversary, landau, probability, notions, logic, ff, mm, primitives, events, complexity, oracles, asymptotics, keys]{cryptocode}
\usepackage{thmtools,thm-restate}

\usepackage{url}
\usepackage[colorlinks]{hyperref} \usepackage[capitalize]{cleveref}
\usepackage{todonotes} \usepackage{csquotes}
\usepackage{stmaryrd}
\usepackage{multirow}
\usepackage{enumerate}
\usepackage{enumitem}
\usepackage{amssymb}
\usepackage{etoolbox}
\usepackage{quantikz}
\usepackage{orcidlink}
\AtEndEnvironment{proof}{\qed}

 \newcommand{\hvzktv}{\ensuremath{\epsilon_{\Sigma}^{\hvzk{}}}}
\newcommand{\hvzkrenyi}{\ensuremath{\delta_{\Sigma}^{\alpha\pcmathhyphen{}\hvzk}}}
\newcommand{\hvzkkl}{\ensuremath{\epsilon_{\Sigma}^{KL\pcmathhyphen{}\hvzk}}}
\newcommand{\rpsfdomrenyi}[1][\alpha]{\ensuremath{\delta_{\rpsf}^{#1\pcmathhyphen{}dom}}}
\newcommand{\rpsfdomtv}{\ensuremath{\epsilon_{\rpsf}^{dom}}}
\newcommand{\rpsfpretv}{\ensuremath{\epsilon_{\rpsf}^{pre}}}
\newcommand{\rpsfprerenyi}[1][\alpha]{\ensuremath{\delta_{\rpsf}^{#1\pcmathhyphen{}pre}}}
\newcommand{\rpsfprekl}{\ensuremath{\epsilon_{\rpsf}^{KL\pcmathhyphen{}pre}}}
\newcommand{\hfinstancetv}{\ensuremath{\epsilon_{\INSTANCE}}}
\newcommand{\hfrandtv}{\ensuremath{\epsilon_{\RAND, \nrhashqueries}}}
\newcommand{\hfsigntv}{\ensuremath{\epsilon_{\SIGN, \nrsignqueries}}}
\newcommand{\hfsignrenyi}{\ensuremath{\delta_{\SIGN, \nrsignqueries, \alpha}}}
\newcommand{\hfabort}{\ensuremath{\epsilon_{abort}}}
\newcommand{\hffinish}{\ensuremath{\epsilon_{\FINISH}}}
\renewcommand{\minentropy}{\beta}
\newcommand{\saltbits}{k}
\newcommand{\msg}{m}

\newcommand{\hashcollisionqrom}{20(\nrhashqueries + N \cdot\nrsignqueries + N)^3/|\mathcal{C}|}
\newcommand{\gammacldef}{\gamma_{cl}:= \max_{c \in \mathcal{C}} |\{y \in \mathcal{Y} \mid \gamma(y) = c\}|}

\newcommand{\ntrursis}{\pcnotionstyle{\mathcal{R}\pcmathhyphen{}NTRU\pcmathhyphen{}SIS}}
\newcommand{\ntrurisis}{\pcnotionstyle{\mathcal{R}\pcmathhyphen{}NTRU\pcmathhyphen{}ISIS}}
\newcommand{\tntrurisis}[1]{\pcnotionstyle{#1\pcmathhyphen{}\mathcal{R}\pcmathhyphen{}NTRU\pcmathhyphen{}ISIS}}
\newcommand{\tntrurspisisold}[1]{\pcnotionstyle{#1\pcmathhyphen{}\mathcal{R}\pcmathhyphen{}NTRU\pcmathhyphen{}SPISIS}}
\newcommand{\tntrurspisis}[1]{\pcnotionstyle{#1\pcmathhyphen{}\mathcal{R}\pcmathhyphen{}NTRU\pcmathhyphen{}SPISIS^*}}
\newcommand{\ntrutrapgen}{\pcalgostyle{TpdGen_{NTRU}}}
\newcommand{\ntrusamplepre}{\pcalgostyle{PreSmp_{NTRU}}}

\newcommand{\rsig}{\pcalgostyle{RSig}}
\newcommand{\setup}{\pcalgostyle{Stp}}

\newcommand{\nrchalqueries}{{q_{c}}}
\newcommand{\nrsignqueries}{{q_{s}}}
\newcommand{\nrhashqueries}{q_{\hash}}
\newcommand{\salt}{\pckeystyle{salt}}

\newcommand{\START}{\pcalgostyle{START}}
\newcommand{\RAND}{\pcalgostyle{RAND}}
\newcommand{\SIGN}{\pcalgostyle{SIGN}}
\newcommand{\FINISH}{\pcalgostyle{FINISH}}
\newcommand{\INSTANCE}{\pcalgostyle{INSTANCE}}

\newcommand{\REPRO}{\pcalgostyle{REPRO}}

\newcommand{\seufcra}{\pcnotionstyle{SUF\pcmathhyphen{}CRA}}
\newcommand{\weufcra}{\pcnotionstyle{UF\pcmathhyphen{}CRA}}
\newcommand{\weufcraone}{\pcnotionstyle{UF\pcmathhyphen{}CRA1}}
\newcommand{\eufnra}{\pcnotionstyle{UF\pcmathhyphen{}NRA}}

\newcommand{\wrec}{\pcnotionstyle{W\pcmathhyphen{}REC}}
\newcommand{\cur}{\pcnotionstyle{CUR}}

\newcommand{\anon}{\pcnotionstyle{ANON}}

\newcommand{\pcrequire}{\highlightkeyword{require}}

\newcommand{\chal}{\pcoraclestyle{CHAL}}

\newcommand{\ch}{\pckeystyle{ch}}
\newcommand{\com}{\pckeystyle{com}}
\newcommand{\rsp}{\pckeystyle{rsp}}
\newcommand{\imp}{\pcnotionstyle{IMP}}
\newcommand{\instance}{\pckeystyle{inst}}
\newcommand{\gen}{\pcalgostyle{Gen}}
\newcommand{\AOS}{\pcalgostyle{AOS}}

\newcommand{\SUC}{\pcalgostyle{SUC}}
\newcommand{\CL}{\pcalgostyle{CL}}
\newcommand{\SZ}{\pcalgostyle{SZ}}
\newcommand{\PProp}{\pcalgostyle{P}}

\newcommand{\trapgen}{\pcalgostyle{TpdGen}}

\newcommand{\sampledom}{\pcalgostyle{SampleDom}}
\newcommand{\samplepre}{\pcalgostyle{SamplePre}}

\newcommand{\SK}{\pcalgostyle{SK}}

\newcommand{\rpsf}{\pcalgostyle{RPSF}}

\newcommand{\domain}{\pckeystyle{D}}
\newcommand{\range}{\pckeystyle{R}}
\newcommand{\oneway}{\pcnotionstyle{OW}}
\newcommand{\collision}{\pcnotionstyle{COL}}

\newcommand{\ID}{\pcalgostyle{ID}}
\newcommand{\hvzk}{\pcnotionstyle{HVZK}}

\newcommand{\ketbra}[2]{\ket{#1}\bra{#2}}
\newcommand{\trace}{\operatorname{Tr}}
\newcommand{\tracenorm}[1]{\norm{#1}_{tr}}
\newcommand{\operatornorm}[1]{\norm{#1}_{\infty}}
\newcommand{\euclnorm}[1]{\norm{#1}_{2}}
\DeclareMathOperator{\CNOT}{CNOT}

\newcommand{\game}[1]{\pcgameprocedurestyle{\pcgamename}_{#1}}

\newcommand{\oraclef}[1]{\ensuremath{{f_{#1}}}}
\newcommand{\Oraclef}[1]{\ensuremath{\oracle{}_{\oraclef{#1}}}}

\ifshowcomments
\newcommand{\cm}[1]{\textcolor{purple}{[\textit{CM: #1}]}}
\newcommand{\mb}[1]{\textcolor{brown}{[\textit{MB: #1}]}}
\else
\newcommand{\cm}[1]{}
\newcommand{\mb}[1]{}
\fi

\newcommand{\BC}{\operatorname{BC}}
\newcommand{\hellingerdistance}{\operatorname{dh}} 
\newboolean{shortver}
\setboolean{shortver}{false}
\newboolean{anonymous}
\setboolean{anonymous}{false}
\newboolean{withprevreviews}
\setboolean{withprevreviews}{false}

\ifthenelse{\boolean{shortver}}{}{
    \setcounter{tocdepth}{2}
}

\begin{document}

\newacronym{PQ}{PQ}{post-quantum}
\newacronym{RSS}{RSS}{ring signature scheme}
\newacronym{RO}{RO}{random oracle}

\newacronym{NIST}{NIST}{National Institute of Standards and Technology}

\newacronym{ppt}{ppt}{probabilistic polynomial time}

\newacronym{KDF}{KDF}{key-derivation function}
\newacronym{sKEM}{sKEM}{split-KEM}
\newacronym{AKEM}{AKEM}{authenticated KEM}
\newacronym{HPKE}{HPKE}{hybrid public key encryption}

\newacronym{DVS}{DVS}{designated verifier signature}
\newacronym{MAC}{MAC}{message authentication code}
\newacronym{RS}{RS}{ring signature}
\newacronym{dRS}{dRS}{deniable \ac{RS}}
\newacronym{FDH}{FDH}{full-domain-hash}
\newacronym{PFDH}{PFDH}{probabilistic full-domain-hash}

\newacronym{PSF}{PSF}{preimage sampleable (trapdoor) function}
\newacronym{RPSF}{RPSF}{ring preimage sampleable function}
\newacronym{RTDF}{RTDF}{ring trapdoor function}

\newacronym{AKE}{AKE}{authenticated key exchange}
\newacronym{DAKE}{DAKE}{deniable authenticated key exchange}
\newacronym{BAKE}{BAKE}{bundled \ac{AKE}}

\newacronym{PRF}{PRF}{pseudorandom function}

\newacronym{ROM}{ROM}{random oracle model}
\newacronym{QROM}{QROM}{quantum random oracle model}

\newacronym{KCI}{KCI}{key compromise impersonation}
\newacronym{PFS}{PFS}{perfect forward secrecy}
\newacronym{WFS}{WFS}{weak forward secrecy}
\newacronym{MEX}{MEX}{maximal exposure}
\newacronym{RR}{RR}{randomness reveal}
\newacronym{SSR}{SSR}{session state reveal}

\newacronym{DoS}{DoS}{denial of service}

\newacronym{HVZK}{HVZK}{honest verifier zero knowledge}

\newacronym{RD}{RD}{Rényi divergence}
\newacronym{SD}{SD}{statistical distance}

\newacronym{FSwA}{FSwA}{Fiat-Shamir with aborts}
\newacronym{HSwA}{HSwA}{Hash-then-Sign with aborts}
\title{
    Quantum Oracle Distribution Switching and Applications to \textsc{Falcon} and Ring Signatures
}

\titlerunning{Quantum Oracle Distribution Switching and Applications}

\ifthenelse{\boolean{anonymous}}{
    \author{}
    \institute{}
}{
    \author{Marvin Beckmann\inst{1}\orcidlink{0009-0008-0178-1423} \and
        Christian Majenz\inst{1}\orcidlink{0000-0002-1877-8385}}

    \authorrunning{Marvin Beckmann \and Christian Majenz}

    \institute{Technical University of Denmark, Kongens Lyngby 2800
        \email{\{mabeck, chmaj\}@dtu.dk}}
}

\maketitle
\begin{abstract}

    Motivated by digital signature algorithms ranging from \textsc{Falcon} to fully-anonymous ring signatures used in Signal-style key exchange, such as \textsc{Gandalf}, we revisit a fundamental problem in post-quantum security proofs: distinguishing oracle functions whose outputs are sampled independently from distributions $P$ and $Q$ that are close.
    In the context of signatures, closeness is often measured via Rényi divergences, which yield \emph{multiplicative} bounds in the classical setting.
    A counterexample shows that such multiplicative-error bounds for distinguishers with classical oracle access have no analogue for quantum access, and we provide two alternative approaches based on small-range distributions and reprogramming techniques.
    We also give a concrete, optimal bound for the case where $P$ and $Q$ are close in statistical distance.

    We apply these techniques to the motivating constructions.
    (i)~We give the first QROM security proof for \textsc{Falcon} that avoids oracle indistinguishability arguments based on statistical distance.
    This is crucial, as \textsc{Falcon}'s ROM proof relies on Rényi divergence, while the statistical distance induced by its parameters is too large to yield meaningful bounds.
    (ii)~We formalize and abstract the ring signature construction used in \textsc{Gandalf} as a modular framework by defining ring trapdoor preimage-sampleable functions (RPSFs), for which we obtain two QROM proofs.

    We also provide two QROM security proofs for AOS ring signatures, adapting existing QROM techniques.
    Together with our results on RPSF-based ring signatures, this yields QROM security proofs for a broad class of fully-anonymous linear ring signature constructions, including \textsc{Gandalf} and the AOS-based constructions \textsc{Erebor} and \textsc{MayoRS}.

\end{abstract}
\ifthenelse{\boolean{shortver}}{}{
    \tableofcontents
    \newpage
}

\section{Introduction}\label{sec: Introduction}

A digital signature scheme allows a user to authenticate messages to other parties.
With the growing threat of quantum computers, \ac{PQ} secure signatures are essential.
\ac{NIST} concluded its \ac{PQ} standardization process with three signature standards: ML-DSA, FN-DSA, and SLH-DSA.
Of these, FN-DSA (formerly \textsc{Falcon} \cite{fouque2018falcon}) has received significant attention in recent cryptographic research.

At its core, \textsc{Falcon} is built on the GPV framework \cite{STOC:GenPeiVai08}, a general paradigm for constructing hash-and-sign signatures from preimage sampleable functions.
However, \textsc{Falcon} deviates from the GPV framework in two important ways.
First, it uses rejection sampling (aborts and repeats) to ensure correctness.
Second, and more relevant to this work, it employs optimization techniques that replace statistical distance arguments with Rényi divergence properties, enabling significantly tighter parameter choices.
As shown in \cite{EPRINT:GajJanKil24b}, the statistical distance of \textsc{Falcon}'s preimage sampler is only on the order of $2^{-34}$, which is far too large for standard proof techniques.

The existing \ac{QROM} proof for the GPV framework \cite{AC:BDFLSZ11} proceeds by simulating the entire random oracle with hash values that are statistically close to uniform.
Giving an adversary quantum access to such a simulated oracle remains indistinguishable from the honest oracle, up to an error depending on the statistical distance and the number of queries.
However, this approach fundamentally requires the simulated hash values to have negligible statistical distance from uniform.
Since \textsc{Falcon}'s optimization techniques replace statistical distance with Rényi divergence, these existing \ac{QROM} techniques do not apply, and there is currently no argument that \textsc{Falcon} with its current parameters is secure in the \ac{QROM}.

Closing this gap requires understanding how quantum oracle access behaves when the underlying output distribution is switched from one distribution to another that is close in terms of Rényi divergence rather than statistical distance. Classically, the Rényi divergence yields a clean purely multiplicative bound \cite{EC:LanSteSte14,van_Erven_2014} when switching from the uniform to the simulated distribution.
Whether an analogous bound holds if provided quantum oracle access is the central technical question we address in this work.

Rivest, Shamir, and Tauman~\cite{AC:RivShaTau01} introduced \acp{RSS}, which provide the same functionality as group signatures without trusted setup, allowing a member to sign on behalf of a group without revealing their identity.
The same GPV framework that underlies \textsc{Falcon} also serves as the foundation for ring signature constructions like \textsc{Gandalf} \cite{C:GajJanKil24} and \textsc{FalconRS}, which extend the preimage sampleable function approach to the ring setting.

A prominent application of \acp{RSS} is the construction of \ac{DAKE}, e.g., for messenger applications.
The most widely used key exchange protocol for end-to-end-encrypted instant messaging is the Signal protocol (WhatsApp, Signal, and Facebook Messenger, etc.).
Its initial \ac{DAKE}, X3DH~\cite{Signal:X3DH}, is, however, based on the Diffie-Hellman key exchange, which can be broken by quantum computing attacks.
A variant of that protocol, PQXDH~\cite{Signal:PQXDH}, uses a hybrid approach with ML-KEM to ensure \ac{PQ} confidentiality.
PQXDH thus prevents ``harvest-now-decrypt-later'' attacks, but lacks \ac{PQ} authentication and might not be \ac{PQ}-anonymous either (``harvest-now-judge-later''~\cite{cryptoeprint:2025/1090}).
Most proposals for fully \ac{PQ} secure Signal-conforming DAKEs use \ac{RSS}~\cite{PKC:HKKP21,PKC:BFGJS22,USENIX2025:HKW}.
An \ac{AKEM} is also the primitive behind two modes of the \acl{HPKE} standard~\cite{rfc9180}.
Recent work~\cite{C:GajJanKil24} provides a generic construction of deniable \acp{AKEM} based on KEMs and \acp{RSS}.

Signal-conforming \ac{DAKE} constructions apply \acp{RSS} for ring size $2$, and require strong anonymity (anonymity under full key-exposure)~\cite{TCC:BenKatMor06}.
In this context, \acp{RSS} with a signature size growing linearly in the ring size are typically more efficient than logarithmic-sized \acp{RSS}.
Besides constructions like \textsc{Gandalf}, there is a second approach that uses the AOS-transform~\cite{AC:AbeOhkSuz02}, transforming $\mathsf{\Sigma}$-protocols into linear-sized \acp{RSS}, used in \textsc{Erebor} and \textsc{MayoRS}~\cite{CiC:BorLaiLer24,cryptoeprint:2025/1090}.

Both approaches and the resulting explicit constructions are currently only proven secure in the classical \ac{ROM}, and thus \emph{not} supported by a \ac{QROM} proof.
In other words, they do not enjoy provable \ac{PQ} security.
Prior to our work, there were thus two options for \ac{PQ} Signal-conforming \ac{DAKE}: Accepting the lower level of assurance provided by a ring-signature-based protocol without provable \ac{PQ} security, or using a split-KEM-based protocol and accepting its disadvantages (like its inefficiency).
\vspace{.2cm}

\noindent\textbf{Our Contributions.}
In this work, we provide \ac{QROM} proofs for \textsc{Falcon} and the two types of \ac{RSS} used in Signal-conforming \acp{DAKE}.

To facilitate the \ac{QROM} proofs that use Rényi divergence-based optimizations, we give a comprehensive set of results for using Rényi divergences for the (in)distinguishability of quantum-accessible oracles.
This includes a counterexample showing that the main approach used in the classical setting cannot work for quantum-accessible oracles, as well as two alternatives that do work: Targeted reprogramming of output values for high-entropy input based on adaptive reprogramming results \cite{AC:GHHM21}, and a technique to replace all outputs via small-range distributions \cite{FOCS:Zhandry12}.

For the ring trapdoor paradigm of constructing \acp{RSS}, we generalize and abstract the approach of \textsc{FalconRS} and \textsc{Gandalf} by formalizing a novel primitive: \acp{RPSF}.
We then give a generic \ac{RSS} construction from \acp{RPSF} that admits security proofs using either statistical distance or Rényi divergence properties independently for domain sampling and preimage sampling, yielding four concrete security bounds.
Finally, we prove \ac{QROM} security for the generic \ac{RPSF}-based \ac{RSS} construction using two different techniques.
One of them, (a generalization of) the formalism of history-free reductions~\cite{AC:BDFLSZ11}, is relatively straightforward and relies on statistical distance arguments.
The other uses Rényi divergences for tighter parameters, and requires use of our study of Rényi divergences for quantum-accessible oracles.

For AOS \acp{RSS}, we give two security bounds, one for generic $\mathsf{\Sigma}$-protocols and a tighter bound for $\mathsf{\Sigma}$-protocols with commit-and-open structure.
To obtain the generic bound, we employ the measure-and-reprogram technique~\cite{C:DonFehMaj20}.
While the resulting bound grows exponentially with the ring size, it can give meaningful guarantees for ring size 2, as needed for the application to Signal-conforming \acp{DAKE}.
For the second bound, practical for larger rings, we use straightline extraction techniques, yielding a multiplicatively tight reduction to the special soundness of the $\mathsf{\Sigma}$-protocol. The result can be extended to Merkle tree-based commit-and-open $\mathsf{\Sigma}$-protocols.

\vspace{.2cm}
\noindent\textbf{Technical Overview.}
In the following, we give a more detailed overview.

\noindent\textit{Quantum Oracle Distribution Switching.}
Consider an adversary with quantum access to an oracle for a function  $\oraclef{P}:\mathcal{X}\to \mathcal{Y}$.
The outputs of $f$ are independently sampled from a distribution $P$ over $\mathcal{Y}$ for each $x \in \mathcal{X}$.
If $P$ is the uniform distribution, then this is exactly the \ac{QROM}. We give several results on the distinguishability of such oracles for pairs of distributions $P$ and $Q$.

First, we provide an asymptotically tight \emph{explicit} bound on the statistical distance of an algorithm's output, where the algorithm interacts with an oracle for $\oraclef{P}$ or $\oraclef{Q}$, for two distributions $P$ and $Q$, as a function of the statistical distance of $P$ and $Q$.\footnote{Asymptotically, this was also shown in \cite{belovs2019quantum}, but without explicit bounds which are needed for parameter determination in practice. We need explicit bounds, so we included the theorem statement with explicit bounds in the main body, and the proof in \ifthenelse{\boolean{shortver}}{the full version of this work \cite[B.1]{cryptoeprint:2026/293}}{Appendix \ref{subsec: Proof Oracle Distribution Switching using Statistical Distance}}.}
For this explicit bound, we use a \emph{compressed oracle}~\cite{C:Zhandry19} view, and analyze the trace norm of the final states produced by an algorithm interacting with the two oracles.
By adding null-terms, we bound the trace norm by a sum of the operator norms of two ``compression'' operators that depend on the underlying classical distributions.
The norm of these operators can be bounded by the statistical distance of the classical distributions.

Now consider the Rényi divergence $R_{\alpha}(P\|Q)$ of order $\alpha$ between probability distributions $P$ and $Q$.
For an algorithm making $q$ classical queries to $\oraclef{P}$ or
$\oraclef{Q}$, the probability of any output-dependent event under $P$ admits
a multiplicative bound in terms of the Rényi divergence as $R_\alpha(P\|Q)^q$ with the event occuring under $Q$, up to an additional exponent $\frac{\alpha-1}{\alpha}$.

We show that no analogous bound can hold in the quantum setting.
Our construction proceeds as follows. From $f_R$, where $R \in \{P,Q\}$, we define a Boolean function $h[f_R]$ on a domain of size $2^{1+n}$ such that:
(i) a single query to $h[f_R]$ can be implemented using $2$ queries to $f_R$; and
(ii) $h[f_Q]$ is \emph{balanced in expectation} over the choice of $f_Q$, whereas $h[f_P]$ is biased in expectation.

We then apply the Deutsch-Jozsa algorithm~\cite{deutsch1992rapid} to $h[f_R]$ using a single query.
This yields a distinguisher which, given oracle access to $f_R$, outputs $1$ whenever the algorithm reports ``constant''.
Since $h[f_Q]$ is balanced in expectation, this occurs with probability $\bigO{2^{-n}}$ for $R = Q$, while it occurs with constant probability for $R = P$.
Hence, their ratio grows as $\Theta(2^{n})$, which outpaces any loss depending only on $P$ and $Q$, and continues to do so even after weakening the comparison to $\prob{1 \gets \adv^{\oraclef{Q}}}^{\gamma}$ for any fixed $\gamma > 0$.

As a positive result, we show that by accepting an additive error, we can use the \emph{small-range distribution} toolkit from \cite{FOCS:Zhandry12}.
Here, $r$ values are sampled according to the underlying distribution, and for each value in $\mathcal{X}$, one of these $r$ samples is assigned as the value in $\mathcal{Y}$.
For these $r$ samples, we can use the classical Rényi divergence properties and bound the overall capabilities of any algorithm.
However, the error term depends on $r$, and for practical applications, $r$ must be chosen super-polynomially large, increasing the power of the Rényi divergence factor.
To have practical application, it turns out that the statistical distance between the two distributions would also need to be negligible, so the Rényi divergence would not need to be used in the first place.
This concludes the Rényi divergence study with the insight that replacing the entire underlying oracle is impractical when using the Rényi divergence.

Taking a closer look at explicit constructions, we observed that the hash input usually includes a salt to achieve strong unforgeability.
In this case, we can use results on \emph{adaptive reprogramming}~\cite{AC:GHHM21}, also known as \emph{resampling} \cite{EC:ABKM22}, to \emph{replace only the outputs used by the signing oracle}.
In this technique, $R$ positions are reprogrammed by sampling them again from the same underlying distribution before they are programmed into the oracle.
Replacing these samples with samples from a different distribution, we can use the classical properties of the Rényi divergence (or the statistical distance).

\noindent\textit{\textsc{Falcon}.}
\textsc{Falcon} has enough entropy in the message hash input to enable our results for adaptive reprogramming in combination with the Rényi divergence.
The reduction follows ideas from \cite{EPRINT:GajJanKil24b} closely.
First, we reprogram the hashes from each preimage sampling try using one of the $C_s$ targets from the $\tntrurspisis{C_s}$ assumption, where $C_s$ bounds the total number of preimage sampling attempts across all signing queries.
This assumption provides $C_s$ syndromes together with Gaussian preimages and asks the adversary to find a \emph{new} short preimage for any of them.
We apply our results on the Rényi divergence here.
In the next step these signatures are replaced by the preimages provided with the hash values in the previous game.
Finally, we can either include challenges from a multi-target assumption $\tntrurisis{t}$ using small-range distributions or apply measure-and-reprogram to include a challenge from a single-target assumption $\ntrurisis$.
The assumptions differ in that they have either $t$ or $1$ uniform targets for which a valid preimage needs to be found.

\noindent\textit{Ring Trapdoor Ring Signatures.}
We introduce \ac{RTDF} and describe preimage sampleable properties for them, generalizing \acp{PSF}~\cite{STOC:GenPeiVai08}, to give a general \ac{RSS} construction.
Brakerski and Kalal~\cite{EPRINT:BraKal10} introduce a less general framework from \acp{PSF} to construct ring signatures in the standard model.
The first digital signature scheme for which \ac{QROM} security was proven uses \acp{PSF} \cite{AC:BDFLSZ11}. Unsurprisingly, similar strategies can be employed when proving the security of our \ac{RSS}:
We observe that a \emph{history-free} \ac{ROM} proof can be given.
For this, we adapt the original notion of history-free proofs for digital signatures to the setting of \acp{RSS}.
In essence, we can simulate the  \ac{RO} by composing a \emph{private} \ac{RO} with the domain sampler and the evaluation function of the \ac{RTDF}.
Combining the domain sampler and the evaluation function to simulate the \ac{RO} requires statistical distance arguments, and under those, this mapping produces a distribution in the range close to uniform, replacing the \ac{RO}.
Similar to the ordinary signature case, the classical history-free reduction implies a \ac{QROM} reduction.
We present a second separate result that works with properties based on the Rényi divergence that uses similar strategies as the proof for \textsc{Falcon}.

\noindent\textit{AOS Ring Signatures.}
Here, we first describe how to simulate the signing oracle.
This again requires adaptive reprogramming in the \ac{QROM}. To simulate, replace the single instance of the honest $\Sigma$-protocol prover by the \ac{HVZK} simulator, and reprogram the \ac{RO} accordingly.
Given sufficient min-entropy in the simulated commitments and thus the \ac{RO} inputs, we can apply the adaptive reprogramming lemma from \cite{AC:GHHM21}.
It is from this point on that we give two reductions, one for generic sigma protocols and a tight one for commit-and-open protocols.

The strategy for the generic proof is to construct an impersonation adversary against the underlying $\mathsf{\Sigma}$-protocol using a successful forgery adversary. This adversary must send a commitment $\com$, receive a challenge $\ch$, and send a response $\rsp$ in order.
Like in ordinary Fiat-Shamir, this requires programming the random oracle, but here, challenges are computed from the previous commitments in the ring.
In the impersonation game, the challenge is received after the commitment, so to reprogram, the query producing the challenge has to happen \emph{after} the commitment was queried, requiring a definite time order.
To get such an order, we use the measure-and-reprogram technique~\cite{C:DonFehMaj20}. This technique incurs a multiplicative loss in $(2q+1)^{2n}$, where $n$ is the number of queries for which the order is desired.
As we need an order of all queries for the ring, and the ring is of size $N$, this introduces an exponential loss in $N$. Here, one of the reprogrammings actually programs the $\mathsf{\Sigma}$-protocol challenge, and the remaining ones  \emph{commit} the adversary to a \emph{classical ordering} of the queries made for producing the signature.

For commit-and-open protocols, we use a different approach.
Commit-and-open protocols in the (Q)ROM allow \ac{RO}-based extraction of commitments. If an adversary succeeds for sufficiently many different challenges, one can use (generalized) special soundness to extract a witness.
Generalizing the strategy from \cite{TCC:ChiManSpo19,C:DFMS22}, we define a compressed oracle database property for fooling the extractor, i.e., the database contains a valid forgery, but the special soundness extractor fails.
Using the quantum transition capacity framework  \cite{EC:CFHL21} we bound the probability that after $\nrhashqueries$ hash queries, this database property holds.
If the database property does not hold, and the prover succeeds, we can measure the compressed oracle database to recover a witness. An additional Merkle tree optimization only requires slight modifications in the database properties.

\vspace{.2cm}

\noindent\textbf{Additional Related Work.} To the best of our knowledge, the AOS framework has not been analyzed in the \ac{QROM}, and ring trapdoor constructions also have not been proven secure in the \ac{QROM}.
There is one work~\cite{PQCRYPTO:DerRamSla18} that considers \ac{QROM} security of a logarithmic-size ring signature based on symmetric primitives to build an accumulator and simulation-sound extractability of an additional proof system (not considering inefficient plain model \acp{RSS} like \cite{EPRINT:BraKal10}).
The specific ring signatures, though logarithmic in asymptotics, are in the order of several hundred KBs even for small rings, and therefore practically infeasible.
Recent candidates for fully anonymous \acp{RSS}~\cite{C:LyuNguSei21,AC:BeuKatPin20,CiC:BorLaiLer24,ACNS:LuAuZha19,C:YELAD21} are also not proven secure in the \ac{QROM}.
We identified some works that consider the \ac{QROM}, but they are either not fully anonymous~\cite{PKC:HaqSca20}, or consider \ac{QROM} quantum-access-secure \ac{RSS}~\cite{PKC:CCLM22,C:CGHKLM21}, but are inefficient.
The framework of~\cite{EPRINT:CHHHLY21} for group signatures and accountable \acp{RSS} covers neither the AOS transform nor the ring-trapdoor paradigm, but discarding the $\mathrm{msk}$ in accountable \acp{RSS} does yield anonymity under full key exposure.
This does rely on assuming the $\mathrm{msk}$ is irrecoverable, and the repetition parameter ($2\lambda N$) in their work puts signature sizes out of practical range for Signal-conforming \acp{DAKE}.
Katsumata et al.~\cite{cryptoeprint:2025/1090} present a new weaker anonymity notion for \acp{RSS} that they argue to be sufficient.
Full anonymity implies this new deniability notion.
Signal-conforming \acp{DAKE} can also be constructed from split-KEMs like \cite{USENIX:CHNRV24,cryptoeprint:2025/853} instead of \ac{RSS} using the \textsc{K-Waay}~\cite{USENIX:CHNRV24}
protocol.
\vspace{.2cm}

\noindent\textbf{Concurrent independent work.}
After the research for this paper was concluded, we noticed a concurrent independent work~\cite{mittalring} that also generalizes the \ac{RTDF} approach from \cite{EPRINT:BraKal10} to capture the work of \textsc{Gandalf}.
Their work does not consider \ac{QROM} security, and they consider weaker anonymity notions.

Subsequent to our analysis of the AOS \ac{RSS}, there has been some further work on ring signatures, specifically an efficient \ac{QROM} secure ring signature \textsc{PegaRing}~\cite{zhang2025pegasus} was introduced whose security is also proven using commit-and-open $\Sigma$-protocols in the DualRing~\cite{C:YELAD21} framework using techniques similar to those we applied for proving the security of the AOS framework.
\section{Preliminaries}\label{sec: Preliminaries}
By $a \gets \mathcal{D}$, we assign a value to $a$ given $\mathcal{D}$,
where  $\mathcal{D}$ is a distribution or an algorithm.
If $D$ is a finite set, $a \gets D$ denotes uniform sampling from the set.
Let $\mathcal{U}(D)$ denote the uniform distribution over $D$.
We denote $[n]:= \{1,2,\ldots,n\}$.
For standard definitions of the statistical distance $\Delta(P, Q)$, the Rényi divergence $R_{\alpha}(P \| Q) $ and the Kullback-Leibler divergence $D_{KL}(P \| Q)$, see \ifthenelse{\boolean{shortver}}{the full version of this work \cite[A.1]{cryptoeprint:2026/293}}{Appendix~\ref{app:prob}}.

\begin{lemma}[{\cite[Lemma 4.1]{EC:LanSteSte14} and \cite[Theorem 1 and 9]{van_Erven_2014}}]\label{lemma: renyi divergence}
    Let $P$ and $Q$ be two discrete probability distributions and $E$ an event such that $E\subseteq\supp{P}\subseteq \supp{Q}$.
    Let $f: \supp{Q} \to \mathcal{X}$ be a function (or stochastic map).
    For any $\alpha \in (1, \infty]$, we have the probability preservation property and the data processing inequality,
    \begin{align*}
        \prob{P \in E}^{\frac{\alpha}{\alpha-1}} \leq & R_\alpha(P\|Q) \cdot \prob{Q \in E},\ \  \text{and}\ \
        R_\alpha(f(P) \| f(Q)) \leq R_\alpha(P \| Q).
    \end{align*}
\end{lemma}

\begin{lemma}[{\cite[Lemma 2.9]{AC:BLLSS15}}]\label{lemma: renyi multiplication}
    Let $\alpha \in (1, \infty]$, and let $P_1, \ldots, P_q$ and $Q_1, \ldots, Q_q$ be distributions with $\supp{P_i}\subseteq \supp{Q_i}$.
    For independent coordinates it holds that
    \[R_{\alpha}\left(\bigotimes_{i \in [q]} P_i \,\middle\|\, \bigotimes_{i \in [q]} Q_i\right) = \prod_{i \in [q]} R_{\alpha}(P_i\|Q_i) \leq \left(\max_{i \in [q]} R_{\alpha}(P_i\|Q_i)\right)^{q}.\]
\end{lemma}

\begin{lemma}[{\cite[Lemma 1]{CHES:PopDucGun14}}]\label{lemma: KL bound}
    Let $\adv^{P}, \adv^{Q}$ be algorithms making at most $q$ queries to an oracle sampling from distribution $P$ and $Q$
    respectively, and returning a bit.
    Let $0 \leq D_{KL}(P \| Q) \leq \epsilon$.
    Then, it holds that
    \[\left|\prob{1 \gets \adv^{P}} - \prob{1 \gets \adv^{Q}}\right| \leq \sqrt{{q\epsilon/2}}.\]
\end{lemma}

\subsection{Signatures}\label{subsec: Signatures}

The standard definitions of digital signatures are recalled in \ifthenelse{\boolean{shortver}}{the full version of this work \cite[A.2]{cryptoeprint:2026/293}}{\cref{subsec: additional prelims Signatures}}.

A \Acf{RSS} is a quadruple of \ac{ppt} algorithms $(\setup, \kgen, \sign, \verify)$.
The key generation is performed by each participant individually.
A ring $\rho$ consists of the individual public keys of at most $\kappa$ participants.
For this ring, each member of $\rho$ can generate a signature $\sigma$ on behalf of the entire ring using their own secret key.
A signature can be verified using only the public information $\rho$.
The formal syntax and security definitions can be found in \ifthenelse{\boolean{shortver}}{the full version of this work \cite[A.3]{cryptoeprint:2026/293}}{Appendix~\ref{subsec: additional prelims Ring Signatures}}.
\subsection{Sigma Protocols}\label{subsec: Sigma Protocol}

A $\mathsf{\Sigma}$-protocol is a $3$-round public-coin interactive proof $(\prover = (\prover_1, \prover_2), \verifier = (\verifier_1, \verifier_2))$ for a relation $R \subseteq \mathcal{I} \times \mathcal{W}$:
First, the prover $\prover_1$ sends a \emph{commitment} $\com$; then the verifier $\verifier_1$ responds with a random \emph{challenge} $\ch \in \mathcal{C}$; and finally, $\prover_2$ sends a \emph{response} $\rsp$.
This final response is evaluated using $\verifier_2$.
The formal definition and well-known properties can be found in \ifthenelse{\boolean{shortver}}{the full version of this work \cite[A.4]{cryptoeprint:2026/293}}{Appendix~\ref{subsec: additional prelims Sigma Protocols}}.
\subsection{QROM}\label{subsec: QROM}

Our main technical proofs rely on reprogramming techniques~\cite{C:DonFehMaj20,AC:GHHM21} and on a framework~\cite{EC:CFHL21} for proving query complexity bounds in the \ac{QROM}.
The latter is a framework exploiting Zhandry's compressed-oracle technique.
We use a slightly adjusted version of the framework by~\cite{C:DFMS22}.
We also use a proof technique for digital signatures from~\cite{AC:BDFLSZ11} based on history-free reductions, which we generalize to the ring signature setting.
We model quantum access to a random oracle $\oracle: \mathcal{X} \to \mathcal{Y}$ via oracle access to a unitary $\bm{U}_{\oracle}$ defined by  $|x\rangle_{\mathcal{X}}|y\rangle_{\mathcal{Y}} \mapsto |x\rangle_{\mathcal{X}}|y \oplus \oracle(x)\rangle_{\mathcal{Y}}$,
and adversaries with quantum access to $\oracle$ act as a sequence of unitaries, interleaved with applications of $\bm{U}_{\oracle}$.
We utilize two reprogramming techniques.

\noindent\textbf{Measure-and-reprogram.} The first, \emph{measure-and-reprogram} \cite{C:DFMS19,C:DonFehMaj20},  allows measuring randomly selected query inputs and reprogramming the random oracle at those inputs to fresh random outputs.
If an algorithm outputs a tuple of query inputs such that the resulting input-output pairs together with an additional output fulfill a predicate, then the same will hold when the measurements and reprogrammings are applied, with respect to the reprogrammed values.
This also yields the order in which the relevant queries were made.

\begin{theorem}[{\cite[Theorem 6]{C:DonFehMaj20}}]\label{thm: measure and reprogram 2.0}
    Let $n$ be a positive integer, and let $\mathcal{X}$ and $\mathcal{Y}$ be finite non-empty sets.
    Let $\adv$ be an arbitrary oracle quantum algorithm that makes $q$ queries to a uniformly random $\hash: \mathcal{X} \to \mathcal{Y}$ and that outputs a tuple $\bm{x} = (x_{1}, \ldots, x_{n}) \in \mathcal{X}^{n}$ and a (possibly quantum) output $z$.
    There exists a black-box polynomial-time $(n+1)$-stage quantum algorithm $\mathcal{S}$, satisfying the following properties.
    $\cal S$ has the following syntactic behavior: in the first stage it outputs a permutation $\pi$ together with $x_{\pi(1)}$ and takes as input $\Theta_{\pi(1)}$, and then for every subsequent stage $1 < i \leq n$ it outputs $x_{\pi(i)}$ and takes as input $\Theta_{\pi(i)}$; eventually, in the final stage (labeled by $n+1$) it outputs $z$. We denote such an execution of $\cal S$ as $(\pi,\pi({\mathbf x}),z) \leftarrow \langle{\cal S}^{\cal A}, \pi({\mathbf \Theta})\rangle$.
    For any $\bm{x}' \in \mathcal{X}^{n}$ without duplicate entries, any predicate $V$ and uniformly random $\Theta\in\mathcal Y^n:$
    \begin{align*}
        \prob{\bm{x} = \bm{x}' \land V(\bm{x}, \Theta, z): (\pi, \pi(\bm{x}), z) \gets \langle \mathcal{S}^{\adv}, \pi(\Theta) \rangle} \\
        \geq \frac{1}{(2q + 1)^{2n}}\prob{\bm{x} = \bm{x}' \land V(\bm{x}, \hash(\bm{x}), z): (\bm{x}, z) \gets \adv^{\hash}}.
    \end{align*}
\end{theorem}

\noindent\textbf{Resampling lemma.}
The resampling lemma~\cite{AC:GHHM21} says that if an input with sufficiently high min-entropy is chosen and the corresponding output is reprogrammed to a fresh, uniformly random value (it is \emph{resampled}), a polynomial-query distinguisher can detect the reprogramming with a small probability only.
For a precise statement, let  $\REPRO_0$ and $\REPRO_1$ refer to two games, where after an initial learning phase with $q$ queries, in $\REPRO_1$ the random oracle is reprogrammed as described, while it is left unchanged in $\REPRO_0$.

\begin{lemma}[{\cite[Proposition 2]{AC:GHHM21}}]\label{thm: adaptive reprogramming}
    Let $\mathcal{X}_1, \mathcal{X}_2, \mathcal{X}'$ and $\mathcal{Y}$ be finite sets, and let $p$ be a distribution on $\mathcal{X}_1 \times \mathcal{X}'$.
    Let $\distinguisher$ be any distinguisher, issuing $q$ (quantum) queries to $\Oracle{}$ and $R$ reprogramming instructions such that each instruction consists of a value $x_2$, together with the fixed distribution $p$.
    Then
    \[\left|\prob{1 \gets \REPRO_1^\distinguisher} - \prob{1 \gets \REPRO_0^\distinguisher}\right| \leq \frac{3R}{2}\sqrt{q\cdot p_{\max}}\]
    where $p_{\max} := \max_{x_1}p(x_1)$.
\end{lemma}

\subsubsection{Query Complexity bounds in the QROM.}\label{subsubsec: query complexity bound}
For query bounds in the \ac{QROM} we rely on
the technique from  \cite{EC:CFHL21} which builds on Zhandry's compressed oracle framework~\cite{C:Zhandry19}.
The compressed oracle is a simulation of the quantum oracle of a random function $H:\mathcal X\to \mathcal Y$. Its internal state of the oracle can be thought of as a (superposition of) lazy-sampling-style databases of input-output pairs.
The probability that $\adv$ succeeds in a search task can be related to the probability of the database $D$, obtained by measuring the internal state of the compressed oracle after the interaction with $\adv$, satisfying a certain property related to the search task (see \cref{lemma: QROM consistency compressed oracle} below). We can think of the database $D$ as a partial function from $\mathcal X$ to $\mathcal Y$. We write $D(x)=\bot$ if $D$ is undefined on $x$. We denote the set of all possible such databases by $\mathcal D$. For $D\in\mathcal D$, $x\in\mathcal X$ and $y\in\mathcal Y\cup\{\bot\}$, we define $D[x\mapsto y]$ by
$D[x\mapsto y](x)=	y$
and $D[x\mapsto y](x')=D(x')$ for $x'\neq x$.
It will suffice for us to use the  \emph{(quantum) transition capacity} formalism from \cite{EC:CFHL21}, so we will not introduce the compressed oracle in detail.

A subset $P \subseteq \mathfrak{D}$ is called a \emph{database property}, we say that $D \in \mathfrak{D}$ \emph{satisfies} $ P$.
The complement is denoted by $\neg P = \mathfrak{D}\setminus P$.
For a database property $P$, a database $D$ and an input $x\in \mathcal X$, we define the corresponding local property as
\[
    P|_{D|^x}=\{y\in\mathcal Y|D[x\mapsto y]\in P\}.
\]
The maximal probability of a database $D$ satisfying P, when $D$ is obtained by measuring the internal state of the compressed oracle after interaction with $\adv$, maximized over all $q$-query algorithms $\adv$, is defined as
\[\llbracket \bot \Rightarrow^q P \rrbracket := \max_{\adv}\sqrt{\prob{D \in P}}.\]
Its square is an upper bound on the probability of $\adv$ producing such a database.

Lemma 5.6 in~\cite{EC:CFHL21} shows that the probability of fulfilling a database property after $q$ queries can be bounded using the quantum transition capacity.
More precisely, a sequence $P_0, P_1, \ldots, P_q$ with $\neg P_0 = \{\bot\}$ and $P_q = P$ yields
\[\llbracket \bot \Rightarrow^q P\rrbracket \leq \sum_{s=0}^{q-1} \llbracket \neg P_s \to P_{s+1} \rrbracket\]
where each term is the quantum transition capacity between the databases. We do not define the transition capacity as we will only use it as a formal tool. The transition capacity can be bounded using the following.
\begin{theorem}[{\cite[Theorem 2.4]{C:DFMS22}}]\label{thm: quantum transition capacity}
    Let $P$ and $P'$ be database properties with trivial intersection, i.e.\, $P \cap P' = \emptyset$, and for every $D \in \mathcal{D}$ and $x \in \mathcal{X}$ let
    \[L^{x,D} := \begin{cases}
            P|_{D|^{x}}  & \text{if } \bot \in P'|_{D|^x} \\
            P'|_{D|^{x}} & \text{if } \bot \in P|_{D|^x}
        \end{cases}\]
    with $L^{x,D}$ being either of the two if $\bot\notin P|_{D|^x} \cup P'|_{D|^x} $.
    Then
    \begin{align}\label{eq:transcap}
        \llbracket P \rightarrow P' \rrbracket \leq \max_{x,D}\sqrt{10\prob{U \in L^{x,D}}},
    \end{align}
    where $U$ is uniform over $\mathcal{Y}$, and the maximization can be restricted to $D \in \mathcal{D}$ and $x \in \mathcal{X}$ for which both $P|_{D|^{x}}$ and $P'|_{D|^{x}}$ are non-empty.
\end{theorem}
The  \emph{fundamental lemma} of the compressed oracle relates the database to the knowledge of a $q$-query adversary, which yields query bounds.
\begin{lemma}[{\cite[Lemma 2.6]{C:DFMS22}}]\label{lemma: QROM consistency compressed oracle}
    Let $\adv$ be an oracle quantum algorithm that outputs $\bm{x} = (x_1, \ldots, x_\ell) \in \mathcal{X}^\ell$ and $z \in \mathcal{Z}$.
    Let $\tilde{\adv}$ be the oracle quantum algorithm that runs $\adv$, makes $\ell$ classical queries on the outputs $x_i$ to obtain $\bm{y} = \hash(\bm{x})$, and then outputs $(\bm{x}, \bm{y}, z)$.
    When $\tilde{\adv}$ interacts with the compressed oracle instead, and at the end $D$ is obtained by measuring the internal state of the compressed oracle, then, conditioned on $\tilde{\adv}$'s output $(\bm{x}, \bm{y}, z)$,
    \[\prob{\bm{y} = D(\bm{x}) \mid (\bm{x}, \bm{y}, z)} \geq 1 - 2\ell|\mathcal{Y}|^{-1}.\]
\end{lemma}

\subsection{Deutsch-Jozsa Algorithm}\label{subsec: Deutsch-Jozsa Algorithm}

Let $f: \bin^n \to \bin$ be an $n$-bit function.
As a quantum circuit, the Deutsch-Jozsa algorithm~\cite{deutsch1992rapid} can be seen as a measurement in the $+/-$ basis on the first $n$ qubits of the state
$U_f \ket{+}^{\otimes n}\ket{-}$, where $U_f\ket{x}\ket{y} := \ket{x}\ket{y\oplus f(x)}$, $x \in \bin^n$ and $y \in \bin$.
The amplitude of the state $\ket{+}^{\otimes n}$ is $Z_f/2^n$ where
\begin{equation}\label{eq: Deutsch-Jozsa amplitude}
    Z_f := \sum_{x \in \bin^n}(-1)^{f(x)}.
\end{equation}
\section{Quantum Oracle Distribution Switching}\label{sec: Quantum Oracle Distribution Switching}

We define an oracle $\Oraclef{P}$ for the i.i.d function $\oraclef{P}: \mathcal{X} \to \mathcal{Y}$ where $\oraclef{P}(x) \sim P$ for all $x \in \mathcal{X}$ independently and $P$ is a probability distribution on $\mathcal{Y}$.

Consider an algorithm $\adv$ interacting with an oracle $\Oraclef{P}$ or an oracle $\Oraclef{Q}$.
We characterize $\adv$'s output behavior based on the number of quantum queries $q$, as well as the statistical distance and Rényi divergence of $P$ and $Q$.

The distinguishing advantage of an algorithm $\adv$ that can make at most $q$ classical oracle queries is bounded as
\begin{equation}
    \label{eq: classical statistical distance bound}
    \left|\prob{1 \gets \adv^{\oraclef{P}}} - \prob{1 \gets \adv^{\oraclef{Q}}}\right| \leq q\Delta(P,Q)
\end{equation}
and for the Rényi divergence assuming $\supp{P} \subseteq \supp{Q}$
\begin{equation}
    \label{eq: classical RD bound}
    \quad \prob{1 \gets \adv^{\oraclef{P}}} \leq (R_{\alpha}(P\|Q)^q\prob{1 \gets \adv^{\oraclef{Q}}})^{\frac{\alpha-1}{\alpha}}.
\end{equation}

\subsection{Oracle Switching using Statistical Distance}

In this section, let $\epsilon = \Delta(P,Q)$.
Boneh et al.~\cite[Lemma 3]{AC:BDFLSZ11} show that the output distributions of $\adv$ have statistical distance $\bigO{q^2\sqrt{\epsilon}}$ when $P$ or $Q$ is the uniform distribution.
This was improved to $\bigO{q^{1.5}\sqrt{\epsilon}}$ in \cite[Section 7.2]{FOCS:Zhandry12}.
Using the compressed oracle technique, we show that the density matrices of $\adv$'s outputs have trace distance at most $8q\hellingerdistance(P,Q)\leq 8q\sqrt{\epsilon}$.

\begin{restatable}{theorem}{OracleDistributionSwitchingSD}\label{thm: Oracle Distribution Switching using Statistical Distance}
    An algorithm $\adv$ making $q$ quantum queries to $\Oraclef{P}$ or $\Oraclef{Q}$
    has the same output distribution, up to statistical distance
    $8q\hellingerdistance(P, Q) \leq 8q\sqrt{\epsilon}$.
\end{restatable}

\ifthenelse{\boolean{shortver}}{
    The proof and an extended discussion can be found in the full version of this work \cite{cryptoeprint:2026/293}.
    Therefore, the quantum analog of \cref{eq: classical statistical distance bound} only incurs a square-root loss in the statistical distance.
    Our result implies a quantum query complexity of $\Omega(1/\hellingerdistance(P, Q))$ for constant distinguishing advantage, matching the tight bound of \cite{belovs2019quantum}, while additionally providing explicit constants at every advantage.
}{
    The proof can be found in Appendix~\ref{subsec: Proof Oracle Distribution Switching using Statistical Distance}.
    Therefore, the quantum analog of \cref{eq: classical statistical distance bound} only incurs a square-root loss in the statistical distance.

    \paragraph{Complexity Theory Literature.}
    \cite{belovs2019quantum} characterizes the query complexity of
    distinguishing $P$ from $Q$ as $\Theta(1/\hellingerdistance(P,Q))$, where
    $\hellingerdistance$ denotes the Hellinger distance.
    \cref{thm: Oracle Distribution Switching using Statistical Distance} is
    expressed in the same quantity and yields the matching lower bound
    $q = \Omega(1/\hellingerdistance(P,Q))$ for constant distinguishing advantage.
    It is, however, obtained by a different route.
    The bound of \cite{belovs2019quantum} is proven via the adversary method,
    which characterizes the query complexity at constant success probability and
    does not provide the explicit constants needed for concrete parameters,
    whereas we bound the trace distance of the output states at every advantage.

    One can use the result to get a per-query statement, but at a square-root loss.
    Consider a distinguisher $\adv$ with advantage $\eta$, that outputs $1$ with probability $1/2 + \eta/2$ when interacting with
    $\Oraclef{P}$ and with probability $1/2 - \eta/2$ when interacting with
    $\Oraclef{Q}$.
    By $\Theta(1/\eta^2)$ independent repetitions (domain-separating the repetitions)
    and a majority vote, the success probability is boosted to $2/3$.
    A single repetition uses $q$ queries, so $\bigO{q/\eta^2}$
    queries in total, while~\cite{belovs2019quantum} requires
    $\Omega(1/\hellingerdistance(P,Q))$ queries for constant success probability.
    Hence, $q/\eta^2 = \Omega(1/\hellingerdistance(P,Q))$, i.e.,
    \[
        \eta = \bigO{\sqrt{q\cdot \hellingerdistance(P,Q)}}.
    \]
    This is weaker than \cref{thm: Oracle Distribution Switching using Statistical Distance}
    for all $q = o(1/\hellingerdistance(P, Q))$, and the two coincide only at
    $q = \Theta(1/\hellingerdistance(P, Q))$, where both give constant advantage.
}
\subsection{Oracle Switching using Rényi Divergence}\label{subsec: Oracle Switching using Rényi Divergence}

The Rényi divergence measures how close two probability distributions are \emph{multiplicatively}.
The post-processing and probability-preservation properties (see \cref{lemma: renyi divergence}) facilitate relative-error bounds on $\adv$'s outputs when restricted to classical queries.
We explore how this behavior generalizes to quantum access.

\subsubsection{Replacing the entire Distribution Fails}\label{subsec: renyi fails}

\cref{eq: classical RD bound} does not translate to the quantum oracle access setting, i.e., a purely multiplicative bound is impossible.
Any such bound has a multiplicative loss depending only on $P$ and $Q$ (in particular, on their Rényi divergences), and is therefore independent of $n$, the input size for $\oraclef{P}$ and $\oraclef{Q}$.
We show that there is always a choice of $n$ for which the quotient of success probabilities exceeds that loss.

\begin{theorem}
    For all distributions $P \neq Q$ on $\mathcal{Y}$ with $\supp{P} \subseteq \supp{Q}$ and fixed $\gamma > 0$, there exists an algorithm $\adv$ making $2$ quantum queries such that
    \[\lim_{n \to \infty} \prob{1 \gets \adv^{{\oraclef{P}}}} /\prob{1 \gets \adv^{{\oraclef{Q}}}}^\gamma = \infty.\]
\end{theorem}

\begin{proof}The high-level strategy is as follows.
    We define a classical deterministic algorithm that takes a bit string as an input, makes $2$ queries to the oracle given to the distinguisher and outputs one bit, such that the boolean function it implements behaves differently for oracle $f_P$ and oracle $f_Q$.
    For oracle $f_Q$, the boolean function is \emph{balanced in expectation} over the choice of $f_Q$, meaning that in expectation, the outputs $0$ and $1$ occur the same number of times.
    For $f_P$, on the other hand, it is not.
    The Deutsch-Jozsa algorithm therefore almost never reports ``constant'' for $f_Q$, with probability inversely proportional to the domain size $2^n$, whereas for $f_P$ it does so with constant probability.
    The quotient of the two therefore diverges, and does so even when the $f_Q$ probability is raised to any fixed power $\gamma > 0$.

    Fix $\mathcal{X} = \bin^n$ and define the point of maximal probability quotient $y_0 := \arg \max_{y \in \supp{P}}\frac{P(y)}{Q(y)}$, for which $Q(y_0) \in (0,1)$ and $P(y_0)/Q(y_0) > 1$ due to $P \neq Q$ and $\supp{P}\subseteq \supp{Q}$.
    From any function $f:\mathcal X\to \mathcal Y$, we construct a function with a binary output,
    $g[f]:\mathcal X\to \bin$, where $g[f](x) = 0$ if $f(x) = y_0$ and $g[f](x)=1$ otherwise.
    Since this function is not necessarily balanced, we define a padded function $h[f]: \bin^{1 + n} \to \bin$ increasing the size of the input domain, where $h[f](b\|x) = g[f](x)$ if $b=0$.
    For $b=1$, $h[f]$ is defined to be $0$ for exactly $N_0 \coloneqq \lfloor2^{n}(1- Q(y_0))\rfloor$ elements (say for concreteness the first $N_0$ inputs in lexicographic order).
    The $b=1$ half is thus chosen so that its bias cancels the expected bias of $g[f_Q]$, up to the rounding error $\delta \coloneqq 2^{n}(1-Q(y_0)) - N_0 \in [0,1)$.
    Each evaluation of $h[f]$ uses two queries to $f$: one to compute the comparison $f(x) \overset{?}{=} y_0$ into an ancilla, and one to uncompute it after the phase has been applied.
    Since the uncomputation restores the ancilla to $\ket{0}$, the analysis below is unaffected.
    After running the DJ algorithm with $f$ (post-processed to $h[f]$), the probability of measuring $\ket{+}^{\otimes (1 + n)}$ is
    \[
        \prob{1 \gets \adv^{f}} = \mathbb{E}_f  \left[\frac{Z_{h[f]}^2}{2^{2(1+n)}}\right]  = \frac{\mathbb{E}_f\left[Z_{h[f]}^2\right]}{2^{2(1+n)}}
        = \frac{{\mathbb{E}_f\left[Z_{h[f]}\right]}^2 +  \operatorname{Var}_f\left(Z_{h[f]}\right)}{{2^{2(1+n)}}}
    \]
    using the notation from \cref{eq: Deutsch-Jozsa amplitude}.
    Since the $b=1$ half contributes the constant $2N_0 - 2^n = -\mathbb{E}_{f_Q}[Z_{g[f_Q]}] - 2\delta$, we have
    \begin{align*}
        \mathbb{E}_{f}[Z_{h[f]}] & = \mathbb{E}_{f}[Z_{g[f]}] - \mathbb{E}_{f_Q}[Z_{g[f_Q]}] - 2\delta,
    \end{align*}
    so that $\mathbb{E}_{f_Q}\left[Z_{h[f_Q]}\right] = -2\delta$ and $\mathbb{E}_{f_P}[Z_{h[f_P]}] = 2^{1 + n}(P(y_0) - Q(y_0)) - 2\delta$.
    Since $\oraclef{P}$ and $\oraclef{Q}$ assign each input an independent sample, the random variables ${(-1)}^{g[f](x)}$ for $x \in \bin^{n}$ are i.i.d., and the $b=1$ half of $h[f]$ is deterministic.
    Hence, the variance is additive:
    \begin{align*}
        \operatorname{Var}_f  \left(Z_{h[f]}\right)
         & = \sum_{x \in \bin^{n}}\operatorname{Var}_f\left({(-1)}^{g[f](x)}\right)                    \\
         & = \sum_{x \in \bin^{n}}4\prob{g[f](x) = 0}\left(1-\prob{g[f](x) = 0}\right) = 2^{2 + n} v_f
    \end{align*}
    where $v_f = P(y_0)(1-P(y_0))$ if $f = f_P$ and $Q(y_0)(1-Q(y_0))$ if $f = f_Q$.
    Using $\delta < 1$, the rounding contributes $\bigO{2^{-2n}}$ to $\prob{1 \gets \adv^{\oraclef{Q}}}$, and $\bigO{2^{-n}}$ to $\prob{1 \gets \adv^{\oraclef{P}}}$ via the cross term $-2\delta\left(P(y_0)-Q(y_0)\right)2^{-n}$.
    Since $Q(y_0) \in (0,1)$, the denominator is nonzero, and thus
    \begin{align*}
        \frac{\prob{1 \gets \adv^{\oraclef{P}}}}{\prob{1 \gets \adv^{\oraclef{Q}}}^\gamma} & =  \frac{\left(P(y_0) - Q(y_0)\right)^2 + \bigO{2^{-n}}}{\left( 2^{-n}Q(y_0)(1-Q(y_0)) + \bigO{2^{-2n}}\right)^\gamma}
        =\Theta(2^{n\gamma})
    \end{align*}
    tends to $\infty$ in $n$ for any fixed choice of $\gamma$.
\end{proof}

\subsubsection{Replacing some Positions using Small-Range Distributions}\label{subsubsec:small-range distribution approach renyi}
The counterexample makes crucial use of the fact that the multiplicative bounds should hold for any absolute value for the probabilities. Adding an additive error term could be sufficient to achieve a meaningful result for the Rényi divergence, even if the classical properties do not translate.
We explore this approach by considering small-range distributions from \cite{FOCS:Zhandry12}.
In short, our approach relies on first sampling $r$ independent values according to a distribution $P$ on $\mathcal{Y}$.
These are used to simulate the oracle $\Oraclef{P}$.
It is indistinguishable for $\adv$ up to an error term of $\bigO{q^3/r}$ using \cite{FOCS:Zhandry12}.
Next, the distribution is changed from $P$ to $Q$.
Now we can use the properties of the Rényi divergence for $r$ replacements.

\begin{restatable}{proposition}{smallRangeRenyiProof}\label{proposition: Renyi divergence proof with sr}
    Let $\alpha \in (1, \infty], r \in \mathbb{N}_{\geq 1}$ and $P, Q$ be classical distributions over $\mathcal{Y}$ with $\supp{P}\subseteq \supp{Q}$.
    For any $\adv$ making at most $q$ quantum queries
    \[\prob{1 \gets \adv^{\oraclef{P}}} \leq \left(R_{\alpha}(P\|Q)^r \left(\prob{1 \gets \adv^{\oraclef{Q}}} + \frac{\ell(q)}{r}\right)\right)^{\frac{\alpha-1}{\alpha}}+ \frac{\ell(q)}{r}
        ,\]
    where $\ell(q) = \pi^2(2q)^3/3 < 27q^3$.
\end{restatable}

\begin{remark}
    The result is partially for exposition and study of how Rényi divergence can be used by accepting a small additive loss while still switching between the entire oracles.
    This technique could be useful in some situations.
    Statistical distance arguments, however, often yield tighter bounds. This is because selecting $r$ such that the error is negligible implies a small statistical distance, sufficient to apply the results from \cref{thm: Oracle Distribution Switching using Statistical Distance}.
    The full discussion on applicability can be found in \ifthenelse{\boolean{shortver}}{the full version of this work \cite[B.2]{cryptoeprint:2026/293}}{Appendix~\ref{subsec: proof of Oracle Distribution Switching using Small-Range Distributions}}.
\end{remark}
\subsection{Reprogramming with Different Distributions}\label{subsection: Reprgramming with Different Distributions}

Though possible for the statistical distance, it is impossible to use the Rényi divergence for \ac{PQ} security as in the classical setting.
For a fixed oracle, replacing all underlying values with another distribution is infeasible when using the Rényi divergence.
However, it is possible as shown in \cref{subsubsec:small-range distribution approach renyi}, but at the cost of an additive error.

In many applications, it turns out to be sufficient to \emph{reprogram} a quantum oracle for an i.i.d. function with output distribution $P$ using outputs sampled according to $Q$, \emph{for a number of randomly sampled inputs}.
The adaptive reprogramming approach can be modified to facilitate this (see \cref{thm: adaptive reprogramming}).
Let $\REPRO_{1, Q}$ refer to a game, modified from $\REPRO_{1}$, where the output is reprogrammed to a fresh, random value sampled from a distribution $Q$.

\begin{restatable}[Adaptive Reprogramming with Distribution Switching]{lemma}{AdRwDS}\label{cor: adaptive reprogramming with other distribution}
    Let $\mathcal{X}_1, \mathcal{X}_2, \mathcal{X}'$ and $\mathcal{Y}$ be finite sets, and let $p$ be a distribution on $\mathcal{X}_1 \times \mathcal{X}'$.
    Let $Q$ be a distribution over $\mathcal{Y}$ and  $\distinguisher$ be any distinguisher, issuing $q$ (quantum) queries to $\Oracle{}$ and $R$ reprogramming instructions such that each instruction consists of a value $x_2$, together with the fixed distribution $p$.
    Then
    \begin{align*}
        \Big|\prob{1 \gets \REPRO_0^\distinguisher} & - \prob{1 \gets \REPRO_{1, Q}^\distinguisher} \Big| \leq  R\cdot \Delta(Q, \mathcal{U}(\mathcal{Y})) + \delta_{repr}                          \\
        \prob{1 \gets \REPRO_{0}^\distinguisher}    & \leq \left(R_\alpha(\mathcal{U}(\mathcal Y)\|Q)^R\prob{1 \gets \REPRO_{1,Q}^\distinguisher}\right)^{\frac{\alpha -1}{\alpha}} + \delta_{repr}
    \end{align*}
    where $\delta_{repr} = \frac{3R}{2}\sqrt{q\cdot p_{\max}}$ with $p_{\max} := \max_{x_1}p(x_1)$ and $\alpha \in (1, \infty]$.
\end{restatable}

The lemma is an immediate consequence of \cref{thm: adaptive reprogramming} when applying the statistical and Rényi divergence properties to the classical samples used in reprogramming.
A proof, outlining the equations, can be found in \ifthenelse{\boolean{shortver}}{the full version of this work \cite[B.3]{cryptoeprint:2026/293}}{Appendix~\ref{subsec: proof of corollary adaptive reprogramming with other distribution}}.
The result can be extended to two arbitrary underlying distributions. %
\section{A Closer Look at \textsc{Falcon} in the QROM}\label{sec: A Closer Look at Falcon in the QROM}

Our analysis in \cref{sec: Quantum Oracle Distribution Switching} can be applied to the \ac{NIST} candidate \textsc{Falcon}~\cite{fouque2018falcon}.
\textsc{Falcon}'s security in the \ac{QROM} has not been rigorously established in existing literature.
\cite{fouque2018falcon} argues that \textsc{Falcon} is secure because it uses the GPV-framework, which was proven using a history-free reduction in \cite{AC:BDFLSZ11}.
However, a closer look reveals that this line of argument is problematic.
The \ac{QROM} proof of \cite{AC:BDFLSZ11} requires quantum oracle indistinguishability, which is proven via statistical distance arguments (\cref{thm: Oracle Distribution Switching using Statistical Distance} could be used to get tighter arguments than those in \cite{AC:BDFLSZ11}).
The parameters of \textsc{Falcon}, however, are set using ROM security considerations via Rényi divergence arguments\footnote{It also requires considering aborts, we address this separately in the proof.}.
Although it is possible to use the Rényi divergence to get bounds on the statistical distance using Pinsker's inequality, the bounds that are derived this way are impractical (\cite{EPRINT:GajJanKil24b} shows that \textsc{Falcon} has statistical distance $ \approx2^{-34}$).
This renders \cite{AC:BDFLSZ11} infeasible for \textsc{Falcon}.

Classically, the Rényi divergence yields a multiplicative bound (see the ROM proof for \textsc{Falcon} in \cite{EPRINT:GajJanKil24b}).
This directly allows using the data processing inequality (see \cref{lemma: renyi divergence}).
A straightforward generalization of this strategy to the QROM is impossible as we have shown in  \cref{subsec: renyi fails}.
Tolerating additional error terms, we could use \cref{proposition: Renyi divergence proof with sr}, but this implies a small statistical distance, rendering this approach unusable for \textsc{Falcon}'s parameters as well.

As we have sufficient entropy in the hash inputs, we can use our results from \cref{subsection: Reprgramming with Different Distributions} to only reprogram those queries where a call on $\ntrusamplepre$ is performed, enabling use of the Rényi divergence.
To obtain a tighter reduction, \textsc{Falcon} uses a multi-target assumption ($\tntrurisis{t}$) when showing $\eufcma$ security.
In the \ac{QROM}, the only technique we are aware of that supports reduction to this assumption is small-range distributions \cite{FOCS:Zhandry12}.
However, the number of targets must increase significantly: classically $t = \nrhashqueries + 1$ suffices with no additive error in $\epsilon$, whereas in the \ac{QROM} an additive error of $\eta$ in $\epsilon$ costs $t > 27(\nrhashqueries + C_s + 1)^3/\eta$
(see the discussion below \cref{thm: seufcma of CoreFalcon}).

Although \cite[Section 5.1]{EPRINT:GajJanKil24b} does make the assumption that $\tntrurisis{t}$ and $\tntrurspisisold{t}$ are as hard as $\ntrursis$ independent of $t$,
we believe this assumption warrants additional scrutiny due to the significantly larger $t$.
We also need to work with a slightly different version $\tntrurspisis{t}$ of the assumption from \cite[Definition 10]{EPRINT:GajJanKil24b} ($\tntrurspisisold{t}$), as conditioning the challenges on the norm bound, as they do, is incompatible with applying the Rényi arguments per sample.
This is only needed when proving $\seufcma$, and not needed for $\eufcma$ (see \ifthenelse{\boolean{shortver}}{\cite[Remark 3]{cryptoeprint:2026/293}}{\cref{remark: spisis}}).
\cite{EPRINT:GajJanKil24b} also acknowledges that more efficient attacks against $\tntrurisis{t}$
for large $t$ may exist \cite{EPRINT:Bernstein22d}.
For this reason, we consider the alternative reduction to single-target
$\ntrurisis$ via measure-and-reprogram (\cref{thm: measure and reprogram 2.0}).
This is simpler and only incurs the standard square root loss compared to the analogous classical reduction. On the other hand, this strategy yields less tight reductions compared to the reduction to multi-target \ntrurisis.

The details of the proofs and \textsc{Falcon} can be found in \ifthenelse{\boolean{shortver}}{the full version of this work \cite[A.6 and E.1]{cryptoeprint:2026/293}}{Appendix~\ref{subsec: NTRU Lattices} and Appendix~\ref{sec: Falcon in the QROM}}.
The version of \textsc{Falcon} we analyze is \textsc{CoreFalcon}$^{+}$ \cite{EPRINT:GajJanKil24b}.

\begin{restatable}[$\seufcma$ of \textsc{CoreFalcon}$^{+}$]{theorem}{SUFCMAFalcon}\label{thm: seufcma of CoreFalcon}
    For any adversary $\adv$ against the $\seufcma$ security of \textsc{CoreFalcon}$^{+}$, making at most $\nrsignqueries$ signing queries and $\nrhashqueries$ \ac{QROM} queries, there exists an adversary $\bdv$ against $\tntrurspisis{C_s}$, an adversary $\cdv$ against $\ntrurisis$ and an adversary $\ddv$ against $\tntrurisis{t}$,
    s.t. for all $C_s, t \in \NN^{\geq 1}, \alpha_u, \alpha_p \in \RR^{> 1}$:
    \[
        \advantage{\seufcma}{\adv, \textsc{CoreFalcon}^{+}, \nrsignqueries} \leq \left(r_u^{C_s}\left(r_p^{C_s}\left(\epsilon + \epsilon'\right)\right)^{\frac{\alpha_p-1}{\alpha_p}} \right)^{\frac{\alpha_u-1}{\alpha_u}} + \delta
    \]
    where $\epsilon$ and the runtimes of $\bdv$ and $\ddv$ are bounded in \cref{tab: seufcma CoreFalcon} and the other parameters, including $\delta$ as a function of $C_s$ which is in a trade-off relationship with $t$ via epsilon, are defined in \cref{tab: CoreFalcon params}.
\end{restatable}

The two reductions to single- and multi-target assumptions have different advantages and drawbacks.
The single-target reduction via measure-and-reprogram has the advantage of only suffering the standard Grover-type loss compared to the corresponding classical reduction.
It is, however, not tight. The multi-target reduction has the same success probability tightness advantage as its classical analgue.
In contrast to the classical reduction, however, the use of the small-range distribution technique (intuitively necessitated by the fact that a single superposition query can explore all QRO outputs) yields a reduction runtime polynomial in the runtime and the inverse success probability of the adversary (seen by balancing the two summands in the bound on $\epsilon$ in \cref{tab: seufcma CoreFalcon}). Additionally, the additive error $\delta$ that results from the use of reprogramming and thus has no analogue in the classical reduction required a relatively large salt length $t$.

\begin{table*}
    \caption{Trade-offs between runtime and assumptions for $\epsilon$ in \cref{thm: seufcma of CoreFalcon} when reducing to \textsc{CoreFalcon}$^+$. $C_s$ is a parameter of the reduction that can be set to $\Theta(q_s/p)$, where $1-p$ is the rejection probability of the rejection sampling routine.
        In \cite[Corollary 3 to 6]{EPRINT:GajJanKil24b} it was for example set as $C_s < 2 \nrsignqueries$.
    }\label{tab: seufcma CoreFalcon}
    \centering
    {\renewcommand{\arraystretch}{1.3}
        \setlength{\tabcolsep}{10pt}
        \begin{tabular}{|c|c|l|}
            \hline
            Assumption    & Reduction Runtime                     & Bound on $\epsilon$                                                                               \\
            \hline
            single-target & $\poly[\secpar, C_s, \nrhashqueries]$ & $(2\nrhashqueries + 2C_s + 3)^{2}\advantage{\ntrurisis}{\cdv, 1, q, \alpha, \beta}$               \\
            multi-target  & $\poly[\secpar, C_s, t]$              & $\advantage{\tntrurisis{t}}{\ddv, 1, q, \alpha, \beta} + \frac{27(\nrhashqueries + C_s+ 1)^3}{t}$ \\\hline
        \end{tabular}}
\end{table*}
\begin{table*}
    \caption{Parameters from \cref{thm: seufcma of CoreFalcon}
        where $\mathcal{Q}_{\bm{h}}$ is the distribution of
        $\bm{h}\bm{s}_1 + \bm{s}_2 \bmod q$ with
        $\bm{s}_1, \bm{s}_2 \leftarrow \mathcal{D}_{\mathcal{R},s}$, and where
        $\bm{\Lambda}_{\bm{c}} := \bm{\Lambda}(\bm{B}) + (\bm{0}, \bm{c})$ denotes the coset of
        the NTRU lattice containing the preimages of $\bm{c}$, the output of
        $\ntrusamplepre$ being read in the order $(\bm{s}_2, \bm{s}_1)$.
    }\label{tab: CoreFalcon params}
    \centering
    {\renewcommand{\arraystretch}{1.3}
        \setlength{\tabcolsep}{10pt}
        \begin{tabular}{|c|l|}
            \hline
                        & Definition                                                                                                                                                                                                    \\
            \hline
            $\epsilon'$ & $\advantage{\tntrurspisis{C_s}}{\bdv, q, \alpha, \beta} + \dfrac{10(C_s + \nrhashqueries + 1)^2 C_s}{|\mathcal{R}_q|}$                                                                                        \\[6pt]
            $\delta$    & $\dfrac{3C_s}{2}\sqrt{\dfrac{C_s + \nrhashqueries + 1}{2^{\saltbits}}} + \displaystyle\sum_{i = 0}^{\nrsignqueries}{C_s\choose i} (1-p)^{C_s - i}(p)^{i}$                                                     \\[6pt]
            $r_u$       & $\displaystyle\max_{(\bm{h}, \cdot, \cdot) \in \supp{\ntrutrapgen}} R_{\alpha_u}(\mathcal{U}(\mathcal{R}_q) \| \mathcal{Q}_{\bm{h}})$                                                                         \\[6pt]
            $r_p$       & $\displaystyle\max_{(\bm{h}, \bm{f}, \bm{g}) \in \supp{\ntrutrapgen},\, \bm{c} \in \mathcal{R}_q} R_{\alpha_p}(\ntrusamplepre(\bm{B}_{\bm{f}, \bm{g}}, s, \bm{c}) \| \mathcal{D}_{\bm{\Lambda}_{\bm{c}}, s})$ \\[6pt]
            $p$         & $\displaystyle\min_{(\bm{h}, \bm{f}, \bm{g}) \in \supp{\ntrutrapgen},\, \bm{c} \in \mathcal{R}_q} \prob{\substack{\|(\bm{s}_1, \bm{s}_2)\|_2 \leq \beta :                                                     \\[2pt] (\bm{s}_1, \bm{s}_2) \gets \ntrusamplepre(\bm{B}_{\bm{f}, \bm{g}}, s, \bm{c})}}$ \\[10pt]
            \hline
        \end{tabular}}
\end{table*}
\section{Ring Signature from Ring Preimage Sampleable Functions}\label{sec: Ring Signature from RPSF}
We analyse the \ac{QROM} security of ring-trapdoor-based ring signatures.
We first formulate a notion of \acp{RTDF} to capture these ring signatures.

\begin{definition}[Ring Trapdoor Function]\label{def: RPSF}
    A \acf{RTDF} is a quadruple of \ac{ppt} algorithms $(\setup, \trapgen, f, \samplepre)$ defined as follows:
    \begin{itemize}
        \item $\setup(\secparam) \rightarrow \pp$: The setup algorithm takes in the security parameter $\secparam$ and outputs public parameters $\pp$, that also define the maximal ring size $\kappa$.
        \item $\trapgen(\pp) \rightarrow (a, t)$: The trapdoor generation algorithm
              takes the public parameter $\pp$ as input and produces a pair $(a,t)$
              with public value $a$ and trapdoor $t$.
              Every ring $\rho = \{a_1, \ldots, a_N\}$ with $N \leq \kappa$ defines a domain $\domain_{\rho}$, a finite range $\range_{\rho}$ and an efficient function $f_{\rho}: \domain_{\rho} \to \range_{\rho}$.
        \item $\samplepre(\rho, t_{i}, r) \rightarrow d$: The presampling algorithm takes in a ring $\rho$, a trapdoor $t$ and an element $r \in \range_{\rho}$.
              The ring must satisfy the size bound, i.e., $N = |\rho| \leq \kappa$, and $\exists i \in [N]$ such that $\mu(t) = \rho_i$\footnote{We assume (w.l.o.g.) that there exists a one-way function $\mu$ such that $\mu(t) = a$ for all $(a,t) \in \supp{\trapgen}$.}.
              It outputs a $d \in \domain_{\rho}$.
    \end{itemize}
    A \ac{RTDF} is $\delta(\kappa)$-correct if for any $\pp \gets \setup(\secparam)$, all $N \leq \kappa$, any $\{(a_j, t_j)\}_{j \in [N]} \subseteq \supp{\trapgen(\pp)}$ defining $\rho = \{a_1, \ldots, a_N\}$ and any $r \in \range_{\rho}$
    \[\prob{f_{\rho}(\samplepre(\rho, t_j, r)) \neq r} \leq \delta(\kappa).\]
\end{definition}

We introduce properties for \acp{RTDF} similar to those of \acp{PSF}~\cite{STOC:GenPeiVai08}, in particular, the existence of an efficient domain sampling algorithm $\sampledom(\rho)$ that takes a ring $\rho$ with $|\rho| \leq \kappa$ as input and outputs $d \in \domain_{\rho}$ satisfying the properties below.
If a \ac{RTDF} satisfies all of these properties, we also refer to it as a \ac{RPSF}.
Let $\pp \gets \setup(\secparam), \{(a_{i}, t_{i})\}_{i \in [\kappa]} \gets \trapgen(\pp)$ and $\rho = \{a_i\}_{i \in [\kappa]}$.

\begin{enumerate}
    \item \label{prop: rpsf domain uniformity}\emph{Domain sampling with uniform output}:
          Define $\rpsfdomtv$ and $\rpsfdomrenyi$
          such that for any $\rho'$ with $\rho \cap \rho' \neq \emptyset$ it holds that
          \begin{align*}
              \Delta(D, \mathcal{U}(\range_{\rho'})) \leq \rpsfdomtv \quad
              \text{ and } \quad R_\alpha(\mathcal{U}(\range_{\rho'})\| D) \leq \rpsfdomrenyi
          \end{align*}
          where $D := f_{\rho'}(\sampledom(\rho'))$.

    \item \label{prop: rpsf preimage sampling undetectable}\emph{Preimage sampling is not detectable}:
          Define three upper bounds $\rpsfpretv, \rpsfprerenyi$ and $\rpsfprekl$ such that for every $\rho'$ with $\rho \cap \rho' \neq \emptyset$ and $\forall i \in [\kappa]$ it holds that
          \begin{align*}
              \Delta(D, Q) \leq \rpsfpretv, \quad
              R_\alpha(D\| Q) \leq \rpsfprerenyi\quad \text{ and } \quad
              D_{KL}(D\| Q) \leq \rpsfprekl
          \end{align*}
          where $r \gets \range_{\rho'}$, $D:= \samplepre(\rho', t_{i}, r)$ and $Q$ is the distribution of $d \gets \sampledom(\rho')$ conditioned on $f_{\rho'}(d) = r$ for $\alpha \in (1, \infty]$.

    \item\label{prop: rpsf onewayness}\emph{One-wayness}:
          For a \ac{ppt} algorithm $\adv = (\adv_1, \adv_2)$, define the one-wayness
          \[\advantage{\oneway}{\adv, \rpsf} = \prob{
                  f_{\rho'}(d)  = r \land \rho' \subseteq \rho:
                  \begin{array}{l}
                      (\state, \rho')\gets\adv_1(\secparam, \rho) \\
                      \land r \gets \range_{\rho'}                \\
                      \land d \gets \adv_2(\state, r)
                  \end{array}}.\]

    \item\label{prop: rpsf preimage min entropy}\emph{Preimage min-entropy}:
          The min-entropy of $\sampledom(\rho')$ conditioned on $f_{\rho'}(d) = r$ for $\rho' \subseteq \rho$ is at least $\minentropy(\secpar)$ if for every $r \in \range_{\rho'}$ and every $D \in \domain_{\rho'}$
          \[\prob{D = d: d \gets \sampledom(\rho') \mid f_{\rho'}(d) = r} \leq 2^{-\minentropy(\secpar)}.\]

    \item\label{prop: rpsf collision resistance}\emph{Collision-resistance}:
          For any \ac{ppt} algorithm $\adv$, define the collision resistance
          \[\advantage{\collision}{\adv, \rpsf} \!= \!\prob{
                  f_{\rho'}(d_{1})  \!= \! f_{\rho'}(d_{2})  \! \land\! d_{1}    \! \neq \!d_{2}
                  \! \land \!\rho'   \!  \subseteq \!\rho
                  : (d_{1},\! d_{2},\! \rho')\gets\adv(\secparam, \rho)}.\]
\end{enumerate}
\label{def: def and security RS}

We construct a generic ring signature from a \ac{RPSF} as in \cref{fig: rsig from rpsf}.
Each party generates a \ac{RPSF} trapdoor ($\sk$) and public value ($\pk$).
To sign a message for a fixed ring, hash the ring and the message together with a salt $\salt $ to compute a target $h$.
We then use the \acp{RPSF} presample algorithm to get a preimage.
The signature is the preimage and the salt.
A signature is verified by recomputing the target and checking that the signature maps to the target under $f$.

\begin{figure*}
    \begin{pchstack}[boxed, center, space=0.5em]
        \procedure[linenumbering]{$\sign(\sk, \rho, \msg)$}{
            \salt \gets \bin^{\saltbits}\\
            h \gets \hash(\rho, \salt, \msg)\\
            \sigma \gets \rpsf.\samplepre(\rho, \sk, h)\\
            \pcreturn (\salt, \sigma)
        }
        \procedure[linenumbering]{$\verify(\rho, \msg, (\salt, \sigma))$}{
            h \gets \hash(\rho, \salt, \msg)\\
            \pcreturn \rpsf.f_\rho(\sigma) = h
        }
    \end{pchstack}
    \caption{Generic construction of a ring signature from a \ac{RPSF} with $k$ bits of salt. The key generation and setup algorithm of the ring signature are identical to those of $\rpsf$.}\label{fig: rsig from rpsf}
\end{figure*}

\subsection{Unforgeability via Adaptive Reprogramming}\label{subsec: Unforgeability via TAR}

The target for the preimage sampler in the signing algorithm in \cref{fig: rsig from rpsf} is computed using a hash function.
If the hash input has sufficiently high min-entropy in its inputs to $\hash$, e.g. $\lambda =\saltbits$, we can model the hash function as a quantum-accessible random oracle and use reprogramming (\cref{cor: adaptive reprogramming with other distribution}) to simulate the signing oracle for an \seufcra~attacker (\cref{cor: adaptive reprogramming with other distribution}).

\begin{theorem}\label{thm: Unf RPSF RS TAR}
    Let $\rpsf$ be a $\delta(\kappa)$-correct \ac{RPSF} with preimage-min entropy $\beta(\lambda)$, and $\rsig$ be the generic ring signature from \cref{fig: rsig from rpsf}.
    Let $\alpha_1, \alpha_2 \in (1,\infty], N \leq \kappa, k \in \NN$ and $\hash$ be modeled in the \ac{QROM}.
    For any $\seufcra$ adversary $\adv$ making at most $\nrsignqueries$ signing and $\nrhashqueries$ hash queries, there exists an adversary $\adv_{\collision}$ against the collision property of the \ac{RPSF}, and an adversary $\adv_{\oneway}$ against the one-wayness of the \acp{RPSF} s.t. $\advantage{\seufcra}{\adv, \rsig, N, \nrsignqueries}$ is bounded as in \cref{tab: seufcra bounds RPSF via TAR}.
\end{theorem}

\begin{table*}
    \caption{Bounds on $\advantage{\seufcra}{\adv, \rsig, N, \nrsignqueries}$ in \cref{thm: Unf RPSF RS TAR} when using statistical distance (SD) or Rényi divergence (RD) where $C$ is the minimal number of elements in $\range_{\rho}$ for any honestly generated ring $\rho$,
    $\epsilon' = \frac{3\nrsignqueries}{2}\sqrt{(\nrsignqueries + \nrhashqueries + 1)\cdot \frac{1}{2^\saltbits}}$ and
    $\epsilon = \advantage{\collision}{\adv_{\collision}, \rpsf} + {20(\nrsignqueries + \nrhashqueries + 1)^2}{C^{-1}}\nrsignqueries + (2\nrsignqueries + 2\nrhashqueries + 3)^2\advantage{\oneway}{\adv_{\oneway}, \rpsf}$.
    }\label{tab: seufcra bounds RPSF via TAR}
    \centering
    \begin{tabular}{|c|c|c|}
        \hline
        \rule{0pt}{11pt} Domain & Preimage & Bound on $\advantage{\seufcra}{\adv, \rsig, N, \nrsignqueries}$                                                                                                                                     \\
        \hline
        RD                      & RD       & $\left((\rpsfdomrenyi[\alpha_1])^\nrsignqueries\left((\rpsfprerenyi[\alpha_2])^\nrsignqueries\cdot \epsilon\right)^{\frac{\alpha_2-1}{\alpha_2}}\right)^{\frac{\alpha_1 -1}{\alpha_1}} + \epsilon'$ \\
        SD                      & RD       & $\left((\rpsfprerenyi[\alpha_2])^\nrsignqueries\cdot \epsilon\right)^{\frac{\alpha_2-1}{\alpha_2}} + \nrsignqueries \rpsfdomtv+ \epsilon'$                                                          \\
        RD                      & SD       & $\left((\rpsfdomrenyi[\alpha_1])^\nrsignqueries\left(\epsilon + \nrsignqueries\rpsfpretv\right)\right)^{\frac{\alpha_1 -1}{\alpha_1}} + \epsilon'$                                                  \\
        SD                      & SD       & $\nrsignqueries(\rpsfdomtv + \rpsfpretv) + \epsilon + \epsilon'$                                                                                                                                    \\\hline
    \end{tabular}
\end{table*}

\begin{proof}
    The correctness is $\delta(\kappa)$ and follows by construction.
    The strategy for the $\seufcra$ proof is as follows.
    The random oracle is reprogrammed at every signature query.
    The hash inputs are domain-separated for different rings.
    We can reprogram the hash value to a fresh uniform random value, which we then in turn can replace by sampling a value $\sigma$ in the domain using $\sampledom$ and using $f_\rho(\sigma)$. Now, the signature $\sigma$ is freshly uniformly sampled and thus does not depend on the private signing keys anymore,
    so we can transition to the $\eufnra$ game.
    Finally, a forgery breaks the one-wayness of the underlying \ac{RPSF}.

    \noindent$\game{0}$: This is the original game, so here we have
    \[\advantage{\seufcra}{\adv, \rsig, N, \nrsignqueries} = \prob{\game{0}^{\adv}}.\]
    $\adv$ eventually outputs a forgery $(\rho^*, (\salt^*, \sigma^*), m^*)$, and we assume that the adversary made a classical query on $(\rho^*, \salt^*, m^*)$.
    This adds one query to $\nrhashqueries$.

    \noindent$\game{1}$:
    The hash value in signature query $i$ for $\rho_i$ is reprogrammed to $f_{\rho_i}(\sigma_i)$, where $\sigma_i \gets \sampledom_{\rho_i}(\secparam)$.
    We use Property \ref{prop: rpsf domain uniformity} and \cref{cor: adaptive reprogramming with other distribution} to get
    \begin{align*}
        \left|\prob{\game{0}^\adv} - \prob{\game{1}^\adv}\right| & \leq \nrsignqueries\cdot \rpsfdomtv + \epsilon'                                                                             \\
        \prob{\game{0}^\adv}                                     & \leq \left((\rpsfdomrenyi[\alpha_1])^{\nrsignqueries}\prob{\game{1}^\adv}\right)^{\frac{\alpha_1 -1}{\alpha_1}} + \epsilon'
    \end{align*}
    where $\epsilon' = \frac{3\nrsignqueries}{2}\sqrt{(\nrsignqueries + \nrhashqueries + 1)\cdot p_{\max}}$ and $\rho' \cap \rho \neq \emptyset$.
    The hash input has at least $\saltbits$ bits of min-entropy from \salt, so $p_{\max} \leq \frac{1}{2}^{\saltbits}$.

    \noindent$\game{2}$: Each signature is replaced by $\sigma_i$.\footnote{This step implicitly contains the correctness properties of $\rpsf$.
        The simulated signatures are all valid, but the ones generated with $\samplepre$ might not be actual preimages.
        This difference in distribution is, however, captured within Property~\ref{prop: rpsf preimage sampling undetectable}, as the property already considers that $\samplepre$ might not output correct preimages.
    }
    This is undetected due to Property \ref{prop: rpsf preimage sampling undetectable}
    \begin{align*}
        \left|\prob{\game{1}^\adv} - \prob{\game{2}^\adv}\right| & \leq \nrsignqueries\cdot \rpsfpretv                                                                              \\
        \prob{\game{1}^\adv}                                     & \leq \left((\rpsfprerenyi[\alpha_2])^{\nrsignqueries}\prob{\game{2}^\adv}\right)^{\frac{\alpha_2 -1}{\alpha_2}}.
    \end{align*}

    We reduce $\game{2}^\adv$ to the properties of the underlying \ac{RPSF}.
    If the final forgery $(\rho^*, (\salt^*, \sigma^*), m^*)$ uses the same target as a previous signature query on $(m, \salt, \rho)$, they define two preimages for the same target.
    If the signatures are the same, then $(m^*,\salt^*) \neq (m,\salt)$,
    and $\adv$'s forgery query hits one of the targets fixed by the signing queries.
    Let $T \subseteq \range_\rho$ denote the set of these targets, so $|T| \le \nrsignqueries$.
    The forgery query lies outside the reprogrammed positions, so its value is uniform, and $\prob{\gamma(U) \in T} \le \nrsignqueries \cdot \gamma_{cl}/|\mathcal{Y}|$ for any fixed $T$.
    Substituting this for $\prob{U \in L^{x,D}} \le s \cdot \gamma_{cl}/|\mathcal{Y}|$
    in the proof of \cref{lemma: transition collision bound} gives a transition capacity of
    $\sqrt{10}\sqrt{\nrsignqueries \cdot \gamma_{cl}/|\mathcal{Y}|}$, which is independent of $s$ and holds at every step.
    Summing over the $\nrsignqueries + \nrhashqueries + 1$ queries and squaring as in the proof of
    \cref{lemma: probability hash collision qrom} bounds the probability that the database contains
    such a query by
    $10(\nrsignqueries + \nrhashqueries + 1)^2 \cdot \nrsignqueries \cdot \gamma_{cl}/|\mathcal{Y}|
        \le 20(\nrsignqueries + \nrhashqueries + 1)^2\nrsignqueries/C$\footnote{The hash function maps directly into
        $\range_\rho$, and hash inputs are domain-separated between rings. Note that
        $\gamma$ can be chosen such that hashing into $\mathcal{Y}$ and mapping to
        $\range_\rho$ satisfies $\gamma_{cl}/|\mathcal{Y}| \le 2/C$ whenever
        $|\mathcal{Y}| \ge C$.}.

    If $\adv$ uses a different uniform target, we reduce to one-wayness by replacing the hash query from the final forgery by a uniform value using measure-and-reprogram.
    This incurs a multiplicative loss of $(2\nrsignqueries + 2\nrhashqueries + 3)^2$.
    In total:
    \[\prob{\game{2}^\adv} \!\leq\! \advantage{\collision}{\adv, \rpsf} \!+\! \frac
        {20(\nrsignqueries + \nrhashqueries + 1)^2\nrsignqueries}{C} + (2\nrsignqueries + 2\nrhashqueries + 3)^2\advantage{\oneway}{\adv, \rpsf}.\]
    The four combinations of transitions between $\game{0}$ to $\game{1}$ and $\game{1}$ to $\game{2}$ yield the four rows of \cref{tab: seufcra bounds RPSF via TAR}.
\end{proof}

Note that removing the salt does not yield a construction provably with $\weufcra$ or $\weufcraone$ security.
The salt is in fact required for the proof technique itself, as the adaptive reprogramming needs the additional min-entropy.

\paragraph{Comparison to the \ac{ROM}.}\label{par:RPSF Proof ROM}
The proof in the \ac{ROM} is structurally the same.
The terms for adaptive reprogramming and \ac{QROM} collisions are replaced by terms for salt and \ac{ROM} collisions.
Both are only additive, but the salt length $\saltbits$ and size of the range can be made smaller.
Lastly, the one-wayness challenge can be included by guessing of one of the $\nrhashqueries + 1$ hash queries and include the challenge in there.
This is roughly a multiplicative difference of $\nrhashqueries$ assuming that $\nrsignqueries \leq \nrhashqueries$.

\subsection{Unforgeability via History-Free Proofs}\label{subsec: History Free Proofs}
If adding a salt is too costly, we can use the history-free reduction technique from \cite{AC:BDFLSZ11}.
Instead of targeted reprogramming, this technique replaces all hash outputs by the RPSF applied to random domain samples, i.e., it switches to a different i.i.d.~function and the results from \cref{thm: Oracle Distribution Switching using Statistical Distance} and \cref{subsec: Oracle Switching using Rényi Divergence} apply.
Hence, we cannot use the Rényi divergence properties in Property~\ref{prop: rpsf domain uniformity}.

We can only show $\weufcraone$ security of the ring signature without the salt.

\begin{theorem}\label{thm: Unf RPSF RS HF}
    Let $\rpsf$ be a $\delta(\kappa)$-correct \ac{RPSF} with preimage-min entropy $\beta(\lambda)$, and $\rsig$ be the generic ring signature from \cref{fig: rsig from rpsf}.
    Let $\alpha \in (1,\infty], N \leq \kappa, \saltbits = 0$ and $\hash$ be modeled in the \ac{QROM}.
    For any $\weufcraone$ adversary $\adv$ making at most $\nrsignqueries$ signing and $\nrhashqueries$ hash queries, there exists an adversary $\adv_{\collision}$ against the collision property of the \ac{RPSF} such that $\advantage{\weufcraone}{\adv, \rsig, N, \nrsignqueries}$ is bounded as in \cref{tab: weufcraone bounds RPSF via HF}.
\end{theorem}

\begin{table}
    \caption{Bounds on
        $\advantage{\weufcraone}{\adv, \rsig, N, \nrsignqueries}$ in
        \cref{thm: Unf RPSF RS HF} when using statistical distance (SD) or
        Rényi divergence (RD) where $\epsilon = 8(\nrsignqueries + \nrhashqueries)(\rpsfdomtv)^{1/2}$.
    }\label{tab: weufcraone bounds RPSF via HF}
    \centering
    \begin{tabular}{|c|c|}
        \hline
        \rule{0pt}{11pt} Preimage & Bound on
        $\advantage{\weufcraone}{\adv, \rsig, N, \nrsignqueries}$                              \\[2pt]
        \hline
        SD                        & $\epsilon + \nrsignqueries\rpsfpretv + 2^{-\beta(\lambda)}
        + \advantage{\collision}{\adv_{\collision}, \rpsf}$                                    \\
        RD                        & $\epsilon + \left[(\rpsfprerenyi)^{\nrsignqueries}
                \left(2^{-\beta(\lambda)} +
                \advantage{\collision}{\adv_{\collision}, \rpsf}
        \right)\right]^{\frac{\alpha-1}{\alpha}}$                                              \\
        \hline
    \end{tabular}
\end{table}

We prove the \ac{QROM} security of this generic ring signature via a history-free reduction.
This includes the history-free simulation of the random oracle and signing oracle, ensuring consistency.
History-free reductions were initially introduced in \cite{AC:BDFLSZ11} for signature schemes.
We extend the concept to ring signatures (see \ifthenelse{\boolean{shortver}}{the full version of this work \cite[C]{cryptoeprint:2026/293}}{Appendix \ref{sec: history free ring signature}}).
We augment the formalism to support explicit security bounds and use \cref{thm: Oracle Distribution Switching using Statistical Distance} to bound the detectability of switching the oracle distribution.
The adaptations and the formal proof can be found in \ifthenelse{\boolean{shortver}}{the full version of this work \cite[C]{cryptoeprint:2026/293}}{Appendix~\ref{sec: history free ring signature}}.

The reduction strategy can be sketched as follows.
The random oracle is simulated by using a different, private random oracle $\mathsf{O}_c$ to generate randomness for $\sampledom$, using it to ``sample'' $\sigma$ in the domain, and then mapping $\sigma$ to the range using $f$.
Then, \cref{thm: Oracle Distribution Switching using Statistical Distance} can be applied.
The signing queries can be simulated by using $\mathsf{O}_c$ and $\sampledom$, which is undetectable due to the preimage sampling being undetectable.
A valid forgery yields a collision for the \ac{RPSF} with high probability due to the min-entropy of the domain sampling.

\subsection{Anonymity}\label{subsec: Anonymity RS}

The anonymity proof uses the same ideas as \cite[Theorem 1]{C:GajJanKil24}.
The strategy is independent of the modelling of the hash oracle, so it applies to the \ac{QROM}.

\begin{corollary}\label{cor: Anon RPSF RS}
    Let $\rpsf$ be a $\delta(\kappa)$-correct \ac{RPSF}, and $\rsig$ be the generic ring signature from \cref{fig: rsig from rpsf}.
    Let $\hash$ be modeled in the \ac{QROM}.
    For any $\anon$ adversary $\adv$ making at most $\nrchalqueries$ challenge queries
        {\allowdisplaybreaks\begin{align}
                \label{eq: den reduction RPSF TV}
                 & \advantage{\anon}{\adv, \rsig, N, \nrchalqueries}  \leq \nrchalqueries\rpsfpretv,                       \\
                \label{eq: den reduction RPSF KL}
                 & \advantage{\anon}{\adv, \rsig, N, \nrchalqueries}  \leq \sqrt{\left(\nrchalqueries\rpsfprekl\right)/2}.
            \end{align}}
\end{corollary}

In the anonymity game, we simulate the $\chal$ oracle using conditional preimage sampling, rendering the signature manifestly independent of the signer and thus the $\chal$ oracle independent of $b$.
The change in output behavior of $\adv$ can be measured using the statistical distance or the Kullback-Leibler divergence.

\begin{proof}
    Replace line 3 in the signing algorithm of \cref{fig: rsig from rpsf} by the conditional preimage sampling, i.e., $d \gets \sampledom_{\rho}(\secparam)$ conditioned on $f_{\rho}(d) = h$.
    In the anonymity game, observe that the signature is independent of $b$, so the adversary can only output the correct $b$ with probability $1/2$.
    It remains to consider how to bound the difference between the original and the modified signing oracle.

    To get \cref{eq: den reduction RPSF TV}, we use the statistical distance $\rpsfpretv$.
    For both $b=0$ or $b=1$, the statistical distance for each query is $\rpsfpretv$, and with $\nrchalqueries$, the overall difference in winning in the unmodified and the modified game is at most $\nrchalqueries \rpsfpretv$.

    For \cref{eq: den reduction RPSF KL}, we use the Kullback-Leibler divergence.
    For both $b=0$ and $b=1$, the only change is in the signing oracle, and the difference between signing oracle and the domain sampling is given by the Kullback-Leibler divergence with $\rpsfprekl$.
    By application of \cref{lemma: KL bound} for $\nrchalqueries$ queries, we get the desired relation.
\end{proof} \section{AOS Ring Signatures}\label{sec: AOS Ring Signatures}

Abe et al.~\cite{AC:AbeOhkSuz02} provide simple ring signatures from any $\Sigma$-protocol using a circular version of the Fiat-Shamir transform~\cite{C:FiaSha86}, the AOS-transform.
It is used in several linear-size ring signatures~\cite{CiC:BorLaiLer24,cryptoeprint:2025/1090}.
Its security was recently proven by Yuen et al.~\cite{C:YELAD21} for canonical verification and established for general AOS-like ring signatures by Borin et al.~\cite{CiC:BorLaiLer24}.
However, both proofs are in the \ac{ROM} only.

The specification of AOS ring signatures is in \cref{fig: AOS Ring Signature}.
Given a $\Sigma$-protocol, we refer to the ring signature obtained using the AOS-framework as $\AOS(\Sigma)$.
\begin{figure*}
    \centering
    \begin{pcvstack}[boxed, space=0.5em,]
        \begin{pchstack}[space=0.5em]
            \procedure[linenumbering]{$\sign(\sk = w, \rho = \{\instance_{1}, \ldots, \instance_{N}\}, \msg)$}{
                \pclinecomment{let $i\in [N]$ be the index s.t. $\mu(w) = \instance_i$}\\
                (\com_{i}, \state_{i}) \gets \Sigma.\prover_{1}(\instance_{i})\\
                \pcfor j = i+1, \ldots, N, 1, \ldots, i-1\\
                \t \ch_{j} \gets \hash(j, \rho, \com_{j-1}, \msg)\\
                \t (\com_{j}, \rsp_{j}) \gets \simulator(\instance_{j}, \ch_{j})\\
                \ch_{i} \gets \hash(i, \rho, \com_{i-1}, \msg)\\
                \rsp \gets \Sigma.\prover_{2}(\state, w, \ch_{i})\\
                \sigma \gets ((\com_{1},\rsp_{1}), \ldots, (\com_{N},\rsp_{N}))\\
                \pcreturn \sigma
            }

            \procedure[linenumbering]{$\verify(\rho = \{\instance_{1}, \ldots, \instance_{N}\}, \msg, \sigma)$}{
                \ch_{1} = \hash(1, \rho, \com_{N}, \msg)\\
                \pcfor j = 1, \ldots, N\\
                \t \pcif \neg \ID.\verifier_{2}(\instance_{j}, \com_{j}, \ch_{j}, \rsp_{j})\\
                \t \t \pcreturn \false\\
                \t \ch_{j+1} \gets \hash(j+1, \rho, \com_{j}, \msg)\\
                \pcreturn \true
            }
        \end{pchstack}
    \end{pcvstack}
    \caption{Construction of a ring signature $\AOS(\Sigma)$ from a $\mathsf{\Sigma}$-protocol $\Sigma$ and its \ac{HVZK} simulator $\simulator$ using \cite{AC:AbeOhkSuz02}.
        The construction can be modified if $\Sigma$ has commitment recoverability.
        In the modified version, $\ch_1$ can be sent instead of all the commitments, and the commitments and challenges are computed on the fly.
        Finally, there would be a consistency check, if $\ch_1$ is the challenge that can be recovered from the last commitment.
        $\setup$ only fixes the maximal ring size and $\kgen$ is identical to $\Sigma.\gen$.
    }\label{fig: AOS Ring Signature}
\end{figure*}
We will show $\seufcra$ security of AOS ring signatures in the \ac{QROM} to complement our results on the \ac{RPSF}-based  \ac{RSS}.
The techniques we use are unrelated to the ones developed in \cref{sec: Quantum Oracle Distribution Switching}.
First, we use \cref{thm: adaptive reprogramming} to relate the $\seufcra$ security to the $\eufnra$ security.
In \cref{subsec: General Strategy}, we employ the measure-and-reprogram technique from \cref{thm: measure and reprogram 2.0} to generically relate the $\eufnra$ security to impersonation security of the $\Sigma$-protocol.
For $C\&O$ $\Sigma$-protocols, we give a tighter alternative in \cref{subsec: Commit and Open}, generalizing the approach from \cite{C:DFMS22}.
In this section, we assume perfect correctness for simplicity, but the technique also works for imperfect correctness.

\subsection{\seufcra~to \eufnra}\label{subsec: seufcra to ufnra}

\begin{restatable}{theorem}{AOSSUFCRAtoUFNRA}\label{thm: UFCRA to UFNRA}
    Let $\Sigma$ be a special-sound $\mathsf{\Sigma}$-protocol with \ac{HVZK} simulator $\simulator$ and simulator commitment min-entropy\footnote{If the min-entropy is too small, we can either add commitment entropy by appending a random string, or we can instantiate the AOS construction with a \emph{salted} \ac{RO}, resulting in a higher effective commitment min-entropy.} $\beta(\lambda)$.
    Consider the AOS \ac{RSS} with maximal ring size $\kappa = \kappa(\lambda) \geq N$.
    For any $\seufcra$ adversary $\adv$ making at most $\nrsignqueries$ signing queries and $\anon$ adversary $\bdv$ making at most $\nrchalqueries$ challenge queries, each making at most $\nrhashqueries$ quantum queries to the \ac{RO} $\hash$, there exists a $\eufnra$ adversary $\adv_{\eufnra}$, a $\cur$ adversary $\adv_{\cur}$ and a $\wrec$ adversary $\adv_{\wrec}$
    such that $\advantage{\seufcra}{\adv, \AOS(\Sigma), N,
            \nrsignqueries}$ and $\advantage{\anon}{\bdv, \AOS(\Sigma), N,
            \nrchalqueries}$ are bounded as in \cref{tab: aos bounds} and the parameters are defined as in \cref{tab: aos params}.
\end{restatable}

\begin{table*}
    \caption{Unforgeability and anonymity bounds in \cref{thm: UFCRA to UFNRA} using statistical distance (SD), Rényi divergence (RD) or Kullback-Leibler divergence (KL) for $\alpha \in (1, \infty]$.}\label{tab: aos bounds}
    \centering
    {\renewcommand{\arraystretch}{1.3}
        \setlength{\tabcolsep}{10pt}
        \begin{tabular}{|c|c|l|}
            \hline
            Security & Method                                              & Bound \\
            \hline
            \multirow{2}{*}{$\advantage{\seufcra}{\adv, \AOS(\Sigma), N,
                        \nrsignqueries}$}
                     & SD
                     & $\epsilon_{\eufnra} + \nrsignqueries \cdot \hvzktv
            + \epsilon_{repr} + \epsilon_{repr}'$                                  \\[4pt]
                     & RD
                     & $\left[(\hvzkrenyi)^{\nrsignqueries}
                    \epsilon_{\eufnra}\right]^{\frac{\alpha-1}{\alpha}}
            + \epsilon_{repr} + \epsilon_{repr}'$                                  \\[4pt]
            \hline
            \multirow{2}{*}{$\advantage{\anon}{\bdv, \AOS(\Sigma), N,
                        \nrchalqueries}$}
                     & SD
                     & $\nrchalqueries \cdot \hvzktv + \epsilon_{repr}''$          \\[4pt]
                     & KL
                     & $\sqrt{\left(\nrchalqueries \cdot \hvzkkl\right)/2}
            + \epsilon_{repr}''$                                                   \\[4pt]
            \hline
        \end{tabular}}
\end{table*}

\begin{table*}
    \caption{Parameters from \cref{thm: UFCRA to UFNRA} where $\epsilon_{\eufnra} := \advantage{\eufnra}{\adv_{\eufnra}, \AOS(\Sigma), N}$.
    }\label{tab: aos params}
    \centering
    {\renewcommand{\arraystretch}{1.3}
        \setlength{\tabcolsep}{10pt}
        \begin{tabular}{|c|l|}
            \hline
             & Definition                                           \\
            \hline
            $\epsilon_{repr}$
             & $\dfrac{3\nrsignqueries}{2}\sqrt{(\nrhashqueries
                    + \kappa \cdot \nrsignqueries + N)
            \cdot 2^{-\minentropy(\secpar)}}$                       \\[6pt]
            $\epsilon_{repr}'$
             & $N \cdot \left(\advantage{\cur}{\adv_{\cur}, \Sigma}
                + \advantage{\wrec}{\adv_{\wrec}, \Sigma}\right)
            + \hashcollisionqrom$                                   \\[6pt]
            $\epsilon_{repr}''$
             & $\dfrac{3\nrchalqueries}{2}\sqrt{(\nrhashqueries
                    + \kappa \cdot \nrchalqueries)
            \cdot 2^{-\minentropy(\secpar)}}$                       \\[6pt]
            \hline
        \end{tabular}}
\end{table*}

\begin{proof}[Sketch]
    The full proof can be found in \ifthenelse{\boolean{shortver}}{the full version of this work \cite[B.4]{cryptoeprint:2026/293}}{Appendix~\ref{subsec: proof of AOS ufcra to ufnra}}.
    For unforgeability we use \cref{thm: adaptive reprogramming} to reprogram every signature query with uniform challenges, which can be done assuming sufficient min-entropy in the hash input.
    This introduces the typical adaptive reprogramming error ($\epsilon_{repr}$), and an additive error term ($\epsilon_{repr}'$) to check whether the final forgery is valid for the non-reprogrammed random oracle.
    Then, we use the \ac{HVZK} simulator to remove use of the secret key which can be done using Rényi divergence and statistical distance.

    For anonymity we reprogram the random oracle, and then we use the \ac{HVZK} simulator.
    Statistical \ac{HVZK} gives a straightforward bound based on the number of challenge queries.
    Alternatively, the Kullback-Leibler divergence can be used, which gives a direct bound using \cref{lemma: KL bound}.
\end{proof}

\paragraph{Proof in the \ac{ROM}}\label{par: AOS Proof in the ROM}
As in \cref{par:RPSF Proof ROM}, the main changes arise from natural transitioning between collision resistance in the \ac{QROM} to the \ac{ROM} and the adaptive reprogramming.
Both are only additive terms, and when viewing in the \ac{ROM}, smaller choices for the min-entropy and the challenge space are likely possible.
A proof of $\weufcra$ in the \ac{ROM} can be found in \cite[Proposition A.3]{CiC:BorLaiLer24}.

\subsection{\eufnra: General Strategy using Measure-and-Reprogram}\label{subsec: General Strategy}

When considering the classical reduction from $\seufcra$ to $\eufnra$, a time-ordered list of the queries is essential.
The \ac{ROM} naturally provides this time-ordered list (and only has a guessing cost of $\nrhashqueries\cdot (\nrhashqueries -1)$~\cite[In the proof of Proposition A.3]{CiC:BorLaiLer24}), but unfortunately not the \ac{QROM}.
We can get such a list using the measure-and-reprogram technique~\cite{C:DonFehMaj20}.
We can get a query order for a subset of queries, allowing us to reprogram these queries in that specific order.

\begin{restatable}[$\eufnra$ using Measure-and-Reprogram 2.0]{theorem}{AOSUFNRAMeasureRepr}\label{thm: UFNRA from measure and reprogram}
    Let $\Sigma$ be a $\mathsf{\Sigma}$-protocol.
    For any $\eufnra$ adversary $\adv$ making at most $\nrhashqueries$ quantum queries to $\hash$, there exists an impersonation adversary $\bdv$ that can convince an honest verifier $\verifier$ such that

    \[\advantage{\eufnra}{\adv, \AOS(\Sigma), N} \leq N\cdot {(2(\nrhashqueries+N) + 1)}^{2N} \cdot \advantage{\imp}{\bdv, \Sigma}.\]
\end{restatable}

\begin{proof}[Sketch]
    The full proof can be found in \ifthenelse{\boolean{shortver}}{the full version of this work \cite[B.5]{cryptoeprint:2026/293}}{Appendix~\ref{subsec: proof of AOS ufnra from measure-and-reprogram}}.
    We use \cref{thm: measure and reprogram 2.0} for an ordering of all $N$ queries used within the verification of the adversaries final forgery.
    There always exists an index in this list, where we can successfully include an interactive impersonation adversary against $\Sigma$.
    The first message must be made before the respective challenge for it is fixed by the previous ring member, and with the list, such an index always exists.
    The application of \cref{thm: measure and reprogram 2.0} yields the multiplicative loss in the order of $\nrhashqueries^{2N}$.
\end{proof}

\subsection{\eufnra: Commit-and-Open Sigma Protocols}\label{subsec: Commit and Open}
In this section, we consider commit-and-open ($C\&O$) $\Sigma$-protocols to construct AOS \ac{RSS}.
In a $C\&O$ $\Sigma$-protocol, the prover commits to responses within the first message of the protocol.
The challenge defines which responses need to be opened in the third protocol message.
Formally, we define them as follows.

\begin{definition}[Commit-and-Open Sigma Protocols]\label{def: cando sigma}
    A commit-and-open $\mathsf{\Sigma}$-protocol $\Sigma$ is a $\mathsf{\Sigma}$-protocol with a special form that uses a hash function $\hash$.
    In this protocol, the first message consists of $\ell$ commitments $y_{1}, \ldots, y_{\ell}$ computed as $y_i = \hash(m_i)$ for $m_i \in \mathcal{M}$ and possibly an additional string $a_{\circ}$.\footnote{Note that $m_i \in \mathcal{M}$ may contain some additional randomness.}    The challenge is then picked uniformly at random from $\mathcal{C} \subseteq 2^{[\ell]}$.
    This identifies the responses as $\bm{m}_c = (m_i)_{i \in c}$.
    The verifier accepts iff $\hash(m_i) = y_i$ for all $i\in c$ and some given predicate $\verifier(\instance, c, \bm{m}_c, a_\circ)$ is satisfied.
\end{definition}

Don et al.~\cite{EC:DFMS22} showed that $C\&O$ $\Sigma$-protocols have online extractability after application of the Fiat-Shamir transform.
We introduce their soundness definition for $\Sigma$-protocols, which will also be essential in our proof.
We recall the definition of $\mathfrak{S}^*$-soundness for $C\&O$ $\Sigma$-protocols from \cite{C:DFMS22}.

Consider $\mathcal{C}\subseteq 2^{[\ell]}$, and let $\mathfrak{S}\subseteq 2^{\mathcal{C}}$ be an arbitrary non-empty, monotone increasing (i.e. $S \in \mathfrak{S} \land S \subseteq S' \Rightarrow S' \in \mathfrak{S}$) set of subsets $S\subseteq \mathcal{C}$. Let $\mathfrak{S}_{\min{}} :=\{S \in \mathfrak{S}\mid S' \subsetneq S \Rightarrow S' \notin \mathfrak{S}\}$.
As an example, consider $\mathfrak{S}$ to be $\{S \subseteq \mathcal{C} \mid |S| \geq k\}$, then $\mathfrak{S}_{min{}}$ is defined as the subsets of $\mathcal{C}$ with exactly $k$ elements.

\begin{definition}[{\cite[Def. 5.2]{EC:DFMS22}}]
    A $C\&O$ $\mathsf{\Sigma}$-protocol $\Sigma$ is $\mathfrak{S}$-sound$^{*}$ if there exists an efficient deterministic algorithm $\extractor^{*}_{\mathfrak{S}}(\instance, m_1, \ldots, m_\ell, a_\circ)$
    that takes on input an instance $\instance \in \mathcal{I}$, messages $m_1, \ldots, m_\ell \in \mathcal{M} \cup\{\perp\}$ and a string $a_\circ$, outputs a witness for $\instance$, if $\exists S \in \mathfrak{S}$ such that $\verifier(\instance, c,\bm{m}_c, a_\circ)$ for all $c \in S$.
\end{definition}

Informally, this soundness implies the existence of an extractor capable of recovering a witness if the transcript is valid for many challenges.
This is modeled using $\mathfrak{S}_{\min}$ \emph{without} explicitly being given the set $S$.
Similar to \cite{EC:DFMS22}, we define
\[sz_{triv}^{\mathfrak{S}} := \max_{\hat{S} \notin \mathfrak{S}} |\hat{S}|.\]
capturing the number of challenges a prover can successfully answer by first picking a set $\hat{S} \notin \mathfrak{S}$ of challenges, and then preparing a message $\hat{\bm{m}}$ and $a_0$ such that $\verifier(\instance, c, \hat{\bm{m}}_c, a_0)$ holds if $c \in \hat{S}$.
If an adversary prepares the same for an $S \in \mathfrak{S}$, running the soundness-extractor retrieves the witness.
\cite{EC:DFMS22} defines the probability of a successful attack by dividing $sz_{triv}^{\mathfrak{S}}$ by $|\mathcal{C}|$ as $p_{triv}^{\mathfrak{S}}$.

\subsubsection{UF-NRA using $C\&O$ $\Sigma$-Protocols}\label{subsubsec: UFNRA CO}

Our goal is to show that AOS \acp{RSS} are $\eufnra$ secure by utilizing the online-extractability of $C\&O$ $\Sigma$-protocols.
The main tool for this analysis uses the framework introduced in \cref{subsubsec: query complexity bound}.
We follow a similar procedure to that of Don et al.~\cite[Section 4]{C:DFMS22} when proving the online extractability of $C\&O$ $\Sigma$-protocols under the Fiat-Shamir transformation.
(Alternatively, the formalism of \cite{TCC:ChiManSpo19} could be used.)
First, we introduce three database properties, similar to
those in \cite{C:DFMS22}.
The properties capture collisions, limited size, and the event where a proof can be produced, and no extraction is possible.
The size property can be defined as
\[\SZ_{\leq s} := \left\{D \,\middle|\,|\{x|D(x) \neq \bot\}|\leq s\right\}.\]

We now define the prover-success-extractor-fail property. Consider a set of $N$ instances $\{\instance_{1}, \ldots, \instance_{N}\} \subseteq \mathcal{I}$.We define the sets $\mathcal{M}$ for messages, $\mathcal{C}$ for challenges, and $\mathcal{Y}$ for commitments referring to the sets from the $C\&O$ $\Sigma$-protocols.
We define commitment queries as $D(m)$ for $m \in \mathcal{M}$, returning a commitment $y \in \mathcal{Y}$ and extend this notation to vectors of messages $\mathbf m$.
For a database $D$ and a commitment $y \in \mathcal{Y}$, we define $D^{-1}(y)$ to be the smallest $m \in \mathcal{M}$ with $D(m) = y$.
If no such $m$ exists, we define $D^{-1}(y) = \bot$, as well as $D^{-1}(\bot) = \bot$.
We denote challenge queries as $D(\beta, \rho, \bm{y}, \mu)$ where $\bm{y} \in \mathcal{Y}^{\ell}$, $\mu$ acts as auxiliary input, and $\rho$ is a collection of instances.
The output is a string from $\mathcal{Y}$, which means that we need an additional map $\gamma: \mathcal{Y} \to \mathcal{C}$.
With this, we can define the remaining two database properties of collision\footnote{This collision property captures slightly more than what we need, as we could also separate between collisions for commitments and challenges.}
\[\CL  := \left\{D \,\middle|\,
    \begin{array}{c}
        \exists x\neq x' \in \mathcal{X}: D(x) \neq \bot \neq D(x') \land \gamma \circ D(x) = \gamma \circ D(x')
    \end{array}\right\}\]
and adversarial success ($\rho' \subseteq \rho \subseteq \supp{\Sigma.\gen(\secparam)}$, $N' = |\rho'| \leq |\rho| = N$)
\[\SUC := \left\{D \,\middle|\,
    \begin{array}{c}
        \exists \mu\,\text{and}\,\bm{y}_{\beta} \in \mathcal{Y}^{\ell}\,\text{s.t.\ } \forall \beta \in [N']:                                                   \\
        \verifier(\instance_{\beta}, c_{\beta}, \bm{m}_{\beta, c_{\beta}})\land (\instance_{\beta}, \extractor^{*}(\instance_{\beta}, \bm{m}_{\beta})) \notin R \\
        \text{where}\, c_{\beta + 1} := \gamma \circ D(\beta+1, \rho', \bm{y}_{\beta}, \mu) \land
        \bm{m}_{\beta} := D^{-1}(\bm{y}_{\beta})
    \end{array}\right\}.\]

Our intermediate goal is to bound the transition capacity from an empty compressed database to satisfy either $\SUC$ or $\CL$ after $q$ queries.

\begin{lemma}\label{lemma: transition combined}
    Let $\gammacldef$ and $N < q$, then
    \begin{align*}\llbracket \bot & \Rightarrow^{q} \SUC \cup \CL \rrbracket \leq q\sqrt{{10}/{|\mathcal{Y}|}}\left(\sqrt{(q-1) \gamma_{cl}} + \sqrt{\max\left\{(q-1)\ell, \gamma_{cl} (sz_{triv}^{\mathfrak{S}})^N\right\}}\right).\end{align*}
\end{lemma}

\begin{proof}
    Using the same arguments as in \cite[Proof of Lemma 4]{C:DFMS22}, we get that
    \[\llbracket \bot \Rightarrow^{q} \SUC \cup \CL \rrbracket \leq \sum_{s=0}^{q-1}\Big(\llbracket \SZ_{\leq s}\setminus \CL \rightarrow\CL \rrbracket + \llbracket \SZ_{\leq s}\setminus \SUC \setminus \CL \rightarrow \SUC \rrbracket\Big).\]
    Compared to \cite{C:DFMS22}, we have implicitly removed the first term, as it vanishes regardless, and we have not simplified the last term, as we explicitly need the collision property.
    The first term is bounded similarly to \cite[Example 5.28]{EC:CFHL21}, but we use \cref{lemma: transition collision bound} to get a tighter bound.
    The second term is bounded by \cref{lemma: transition success bound}.
    Finally, we combine the results and use that $s \leq q-1$.
\end{proof}

\begin{restatable}{lemma}{TransitionCollisionBound}\label{lemma: transition collision bound}
    Let $\gammacldef$, then
    \[\llbracket \SZ_{\leq s} \setminus \CL \rightarrow \CL \rrbracket \leq \sqrt{10}\sqrt{s\cdot \gamma_{cl}/|\mathcal{Y}|}.\]
\end{restatable}
Similar results can be found in \cite[Example 5.28]{EC:CFHL21}, so we defer the proof to \ifthenelse{\boolean{shortver}}{the full version of this work \cite[B.6]{cryptoeprint:2026/293}}{Appendix \ref{subsec: proof lemma transition capacity database collision}}.
We can now finalize \cref{thm: UFCRA to UFNRA} using the previously established bound on the quantum transition capacity.

\begin{lemma}\label{lemma: probability hash collision qrom}
    Let $\gammacldef$, then for any oracle quantum algorithm $\adv$ with query complexity $q$, it holds that
    \[\Pr_{\adv}[D \in \CL] \leq 10 q^2 (q-1) \cdot \gamma_{cl}/|\mathcal{Y}|. \]
\end{lemma}

\begin{proof}
    Consider $\llbracket \bot \Rightarrow^{q} \CL\rrbracket \leq \sum_{s=0}^{q-1} \llbracket \SZ_{\leq s} \setminus \CL \rightarrow \CL\rrbracket$.
    Using \cref{lemma: transition collision bound} we get
    \[{\sum}_{s=0}^{q-1} \llbracket \SZ_{\leq s} \setminus \CL \rightarrow \CL\rrbracket \leq q \sqrt{10}\sqrt{(q-1)\cdot \gamma_{cl}/|\mathcal{Y}|}.\]
    Squaring this term yields the upper bound for any $\adv$ with $q$ queries.
\end{proof}

\begin{restatable}{lemma}{transitionSuccessBound}\label{lemma: transition success bound}
    Let $\gammacldef$, then
    \[ \llbracket \SZ_{\leq s}\setminus \SUC \setminus \CL \rightarrow \SUC\rrbracket  \leq \sqrt{10}\sqrt{\max\left\{N \cdot \ell,s\cdot \ell, \gamma_{cl}\cdot(sz_{triv}^{\mathfrak{S}})^N\right\}/|\mathcal{Y}|}.\]
\end{restatable}

A similar result can be found in \cite[Proof of Lemma 4.1]{C:DFMS22}, so we defer the full proof to the \ifthenelse{\boolean{shortver}}{the full version of this work \cite[B.7]{cryptoeprint:2026/293}}{Appendix~\ref{subsec: proof lemma transition success bound}}.
The main modifications are based on having to account for more potential targets due to the ring size $N$.
Specifically, $sz_{triv}^{\mathfrak{S}}$ captures the number of challenges that can be answered without allowing for extraction.
This size is now propagated through the entire ring, giving a total number of potentially valid targets of $(sz_{triv}^{\mathfrak{S}})^N$, where it was only $sz_{triv}^{\mathfrak{S}}$ in \cite{C:DFMS22}.

\begin{theorem}\label{thm: CO UFNRA}
    Let $\Sigma$ be a $\mathfrak{S}^{*}$-sound $C\&O$ $\mathsf{\Sigma}$-protocol with challenge space $\mathcal{C}= \mathcal{C}_{\lambda}$, $\ell = \ell(\secparam)$ commitments, $\hash: \{0,1\}^{*} \to \mathcal{Y}$ be a hash function modeled in the \ac{QROM} and $\gamma:\mathcal{Y}^{\ell} \to \mathcal{C}_{\lambda}$
    be a function with $\gammacldef$.
    Set $\omega = \omega(\lambda) := \max_{c \in \mathcal{C}_\lambda} |c|$.
    Furthermore, let $N \in \mathbb{N}$ be fixed with $N < \nrhashqueries$.
    For any $\eufnra$ adversary $\adv$ making at most $\nrhashqueries$ quantum queries to $\hash$, there exists an efficient adversary $\bdv$ against witness-recovery of the $\Sigma$ with
    \begin{align*}
         & \advantage{\eufnra}{\adv, \AOS(\Sigma), N}  \leq N\cdot \advantage{\wrec}{\bdv, \Sigma}
        + 2\cdot N(\omega(\secpar) + 1)\cdot|\mathcal{Y}|^{-1}                                                                                                                                                                     \\
         & \hspace{1cm} + \nrhashqueries^{2}10/|\mathcal{Y}|\left(\sqrt{(\nrhashqueries-1) \cdot \gamma_{cl}} + \sqrt{\max\left\{(\nrhashqueries-1)\cdot \ell, \gamma_{cl} \cdot (sz_{triv}^{\mathfrak{S}})^N\right\}}\right)^{2}. \\
         & \hspace{1cm} \leq N\cdot \advantage{\wrec}{\bdv, \Sigma}
        + 2\cdot N(\omega(\secpar) + 1)\cdot|\mathcal{Y}|^{-1}                                                                                                                                                                     \\
         & \hspace{1cm} + \nrhashqueries^{2}(\nrhashqueries-1)\ell 10\gamma_{cl}/|\mathcal{Y}|\left(1 + 2(sz_{triv}^{\mathfrak{S}})^{N/2} + (sz_{triv}^{\mathfrak{S}})^N\right).
    \end{align*}
\end{theorem}

\begin{proof}
    Let $\adv$ be an $\eufnra$-adversary. We will construct a witness-recovery adversary $\bdv$ for $\Sigma$ that succeeds if $\adv$ produces a forgery in the $\eufnra$ game.
    $\bdv$ is given  $\instance$ where $(\instance, w) \gets \gen(\secparam)$.
    Now, $\bdv$ guesses an index $i^{*} \leftarrow [N]$.
    For all other $i \in [N] \setminus \{i^{*}\}$, $\bdv$ computes $(\instance_{i}, w_{i}) \gets \gen(\secparam)$.
    Set $\rho = \{\instance_{1}, \ldots, \instance_{N}\}$.
    For simplicity, we assume that all challenges $c \in \mathcal{C}$ are of the same length.

    $\bdv$ gives $\rho$ to $\adv$ and plays the role of the challenger simulating the random oracle of $\adv$ as a compressed oracle. We denote this by $\adv^{\mathsf{CO}}$.
    After making $\nrhashqueries$ queries to the random oracle $\hash$, \adv~outputs $(\rho^{*}, m^{*}, \sigma^{*})$ where $\rho^{*} \subseteq \rho$.
    If $i^{*}$ is not part of the ring, we abort.
    Now, each $\sigma_{i}$ with $i \in [N^*]$ where $\sigma^{*} = (\sigma_{1}, \ldots, \sigma_{N^*})$ consists of commitments $\bm{y}_{i} \in \mathcal{Y}^{\ell}$ and messages $\bm{m}_{i, c_i} = (m_{i,j})_{j \in c_i}$.
    Now, we obtain $c_{i+1} = \gamma \circ\hash(i+1, \rho^{*}, \bm{y}_{i}, \mu)$ as well as $\bm h_{i,c_i}=\hash(\mathbf m_{i, c_i})$.
    For a signature to be valid, the relation $\tilde{R}_{\rho^{*}}$ must hold for the triple $(\sigma^*, \bm c,\bm h)$ which is defined as
    \[\forall i\in[N^*]: \bm{h}_{i, {c_i}} = \bm{y}_{i, c_i} \quad \land \quad \verifier(\rho_{i}, c_{i}, \bm{m}_{i, c_{i}}).\]

    Next, we measure the internal state of the compressed oracle to obtain a database $D$.
    Let $j^{*}$ be the index with which $i^{*}$ is identified in $\rho^{*}$.
    Now, we obtain $m_{j,c}=D^{-1}(y_{j,c})$ for $c\in\mathcal C$ and run $\extractor^{*}(\rho^{*}, \bm{m}_{j^{*}})$ to obtain a witness $w$.

    Defining $\hat{\bm c}$ and $\hat{\bm h}$ like $\bm c$ and $\bm h$, but using $D$ in place of $\hash$, we get
    \begin{align*}
        \advantage{\eufnra}{\adv, \rsig, N} & := \prob{
            \verify(\rho^{*}, m^{*}, \sigma^{*}) \land \rho^{*}\subseteq \rho
        }
        \leq \prob{
            (\sigma^*, \bm c,\bm h)\in \tilde{R}_{\rho^*} \land \rho^{*}\subseteq \rho
        }                                                 \\
                                            & \leq \prob{
            (\sigma^*, \hat{\bm c},\hat{\bm h})\in \tilde{R}_{\rho^*} \land \rho^{*}\subseteq \rho
        }
        +              2\cdot N^{*}(\omega(\secpar) + 1)\cdot|\mathcal{Y}|^{-1}
    \end{align*}
    In the last step, we have used \cref{lemma: QROM consistency compressed oracle}.
    Now, we bound the first term by distinguishing whether the database $D$ satisfies $\SUC \cup \CL$.
    For simplicity, we drop the explicit notation for how $(\rho^*, m^*, \sigma^*)$ are generated.
    We bound
    \begin{align*}
        \prob{	(\sigma^*, \hat{\bm c},\hat{\bm h})\in \tilde{R}_{\rho^*} \land \rho^{*}\subseteq \rho} & \le \prob{ 	(\sigma^*, \hat{\bm c},\hat{\bm h})\in \tilde{R}_{\rho^*} \land \rho^{*}\subseteq \rho\land D\in\SUC\cup\CL}  \\
                                                                                                       & +\prob{ 	(\sigma^*, \hat{\bm c},\hat{\bm h})\in \tilde{R}_{\rho^*} \land \rho^{*}\subseteq \rho\land D\notin\SUC\cup\CL}.
    \end{align*}
    For the first case, we have
    \begin{align*}
         & \prob{	(\sigma^*, \hat{\bm c},\hat{\bm h})\in\tilde{R}_{\rho^{*}}\land \rho^{*}\subseteq \rho \land D \in \SUC \cup \CL} \leq \prob{D \in \SUC \cup \CL} \\
         & \leq \llbracket \bot \Rightarrow^{\nrhashqueries} \SUC \cup \CL\rrbracket^{2}.
    \end{align*}
    In the second case, we get that $(\sigma^*, \hat{\bm c},\hat{\bm h})\in \tilde{R}_{\rho^{*}}$ and $D \notin \SUC \cup \CL \subseteq \SUC$.
    Considering the definition of $\SUC$, there exists $i \in [N^{*}]: (\instance_i, \extractor^{*}(\instance_i, D^{-1}(\bm{y}_i))) \in R$ succeeds because the other terms are all satisfied due to $\tilde{R}_{\rho^{*}}$. Formally,
    \begin{align*}
         & \prob{ 	(\sigma^*, \hat{\bm c},\hat{\bm h})\in \tilde{R}_{\rho^*} \land \rho^{*}\subseteq \rho\land D\notin\SUC\cup\CL} \\
         & \quad\quad\le \prob{\exists i \in [N^{*}]: (\instance_i, \extractor^{*}(\instance_i, D^{-1}(\bm{y}_i))) \in R}.
    \end{align*}
    With probability $\frac{1}{N^{*}}$, this is the index where the challenge instance of $\bdv$ is placed, and $\bdv$ outputs $\extractor^{*}(\instance_i, \bm{m}_i)$ as a witness. As the index where the challenge instance was placed is used in the signature with probability $\frac{N^*}{N}$, we get
    \[N^{-1}\prob{(\sigma^*, \hat{\bm c},\hat{\bm h})\,\text{satisfies}\,\tilde{R}_{\rho^{*}} \land D \notin \SUC \cup \CL} \leq \advantage{\wrec}{\bdv, \Sigma},\]
    and combining these results with \cref{lemma: transition combined} shows our final bound.
\end{proof}

In \cite{C:DFMS22}, the results for $C\&O$ protocols are generalized to allow for Merkle tree based vector commitments.
We perform the same generalization for the AOS transform, see
\ifthenelse{\boolean{shortver}}{the full version of this work \cite[D]{cryptoeprint:2026/293}}{Appendix~\ref{sec: merkle-tree based co protocols}}.
\section{Instantiation and Discussion of Ring Signatures}\label{sec: Explicit Ring Signature Constructions}

\paragraph{AOS-based Ring Signature.}

There are several AOS-based \ac{RSS} in the literature.
\textsc{Erebor}~\cite{CiC:BorLaiLer24} and \textsc{MayoRS}~\cite{cryptoeprint:2025/1090} are recent examples.
Neither is based on a $C\&O$ $\Sigma$-protocol, but \cref{thm: UFNRA from measure and reprogram}  yields the first \ac{QROM} security bound for both. The bounds are not tight, but at least rule out fundamental quantum vulnerabilities.

Explicit examples for $C\&O$ $\Sigma$-protocols are \textsc{ZKBOO}~\cite{USENIX:GiaMadOrl16} and \textsc{ZKB}++~\cite{CCS:CDGORR17} which are the underlying $\Sigma$-protocols in \textsc{Fish} and \textsc{Picnic}~\cite{CCS:CDGORR17}.
More generally, $C\&O$ $\Sigma$-protocols can be constructed using the MPC-in-the-head paradigm~\cite{STOC:IKOS07}.

\paragraph{RPSF-based Ring Signature.}

As an example, consider the ring signature \textsc{Gandalf}~\cite{C:GajJanKil24}.
The translation of \textsc{Gandalf}-like ring signatures into our \ac{RPSF} framework is relatively straightforward (for details see \ifthenelse{\boolean{shortver}}{the full version of this work \cite[E.2]{cryptoeprint:2026/293}}{Appendix~\ref{subsec: Application to Gandalf}}).
Compared to the classical proof technique, we get the modified hash-collision term from the \ac{QROM}, the additional additive adaptive reprogramming term that requires sufficiently many salt bits $\saltbits$ and the typical square-root loss when reducing to the one-wayness via measure-and-reprogram.
Our framework can also be applied (both for \cref{thm: Unf RPSF RS TAR} and \cref{thm: Unf RPSF RS HF}) to \cite{ICICS:WanSun11}, a lattice-based ring signature directly based on \cite{STOC:GenPeiVai08}.

\paragraph{Comparison and Trade-offs.}The generic AOS bound is exponential in ring size $N$ and is not expected to be tight, but gives meaningful guarantees for small rings.
The $C\&O$ bound is multiplicatively tight and avoids exponential scaling when $sz_{\mathrm{triv}} = 1$, e.g., for protocols with 2-special soundness.
\ac{RPSF}-based bounds are tight but constructing \acp{RPSF} for large rings is non-trivial, as the assumed hardness of finding collisions is impacted by the ring size.

\paragraph{Computational Assumptions.}In our work, we relied on statistical \ac{HVZK} properties of $\Sigma$-protocols.
For $\Sigma$-protocols in AOS, it is sufficient to consider weak special \ac{HVZK} properties for unforgeability, i.e., an adversary given a transcript from the simulator or the actual algorithm cannot distinguish them.
Analogously, for \acp{RPSF}, computational indistinguishability between presampling and conditional domain sampling would suffice.
For domain sampling, computational properties suffice when considering \cref{thm: Unf RPSF RS TAR}, but not when considering \cref{thm: Unf RPSF RS HF}.

Anonymity from computational assumptions requires additional access to \emph{the witness/trapdoor} (called strong \ac{HVZK} of $\Sigma$-protocols \cite{CiC:BorLaiLer24}).

\appendix

\begin{credits}
    \subsubsection{\ackname}

    \ifthenelse{\boolean{anonymous}}{}{We thank the anonymous reviewers of Eurocrypt 2026, Crypto 2026 and Asiacrypt 2026 for their valuable feedback.
        We used Claude for language editing and to improve presentation; all technical content and proofs are the authors' own.
        The authors acknowledge support from the Danish Ministry of Defense Acquisition
        and Logistics Organization (FMI). CM acknowledges support by the Independent Research Fund Denmark via a DFF Sapere Aude grant (IM-3PQC,
        grant ID 10.46540/2064-00034B). }
\end{credits}
\newpage

\bibliographystyle{splncs04}

\renewcommand{\theHsection}{\Alph{section}}    \renewcommand{\theHsubsection}{\Alph{section}.\arabic{subsection}}

\ifthenelse{\boolean{shortver}}{}{

\section{Additional Preliminaries}
\subsection{Probabilities}
\label{app:prob}

Let $P$ and $Q$ be probability distributions on the same countable alphabet $\Omega$.
Then, the statistical distance between $P$ and $Q$ is $\Delta(P, Q) :=\sum_{x \in \Omega} |P(x) - Q(x)|/2$.

The Bhattacharyya coefficient of two probability distributions $P, Q$ on a
countable alphabet $\Omega$ is
\[\BC(P, Q) := \sum_{x \in \Omega}\sqrt{P(x)Q(x)}.\]
The Hellinger distance is defined via
\[\hellingerdistance(P, Q)^2 := 1 - \BC(P, Q)
    = \frac12\sum_{x \in \Omega}\left(\sqrt{P(x)} - \sqrt{Q(x)}\right)^2.\]
It is bounded by the statistical distance as $\hellingerdistance(P,Q) \leq \sqrt{\Delta(P,Q)}$.

The Rényi divergence of order $\alpha \in (1, \infty)$ of distributions $P, Q$ with $\supp{P} \subseteq \supp{Q}$ is
\[R_{\alpha}(P \| Q) = \left({\sum}_{x \in \supp{P}} {P(x)^\alpha\cdot Q(x)^{1-\alpha}}\right)^{\frac{1}{\alpha -1}}\]
and the Rényi divergence of order $\infty$ is defined as
\[R_\infty(P \| Q) = {\sup}_{x \in \supp{P}}P(x)/Q(x) \in [1, \infty].\]
These are the multiplicative relations often used in cryptography. The divergences are defined as
\[D_{\alpha}(P\|Q):=\ln(R_\alpha(P\|Q)).\]

The Kullback-Leibler divergence for two distributions over the same countable set $\Omega$ is defined as
\[D_{KL}(P \| Q) = \sum_{x \in \Omega}\ln\left(\frac{P(x)}{Q(x)}\right)\cdot P(x)\]
with the convention that $\ln(x/0) = + \infty$ for any $x > 0$ and $0$ for $x=0$.
Furthermore, Pinsker's inequality gives an implicit bound on the statistical distance $\Delta(P,Q) \leq \sqrt{D_{KL}(P\| Q)/2}$.
As $D_{KL}(P\| Q) \leq D_{\alpha}(P\| Q)$ for all $\alpha > 1$, the Rényi divergence can be used to bound the statistical distance.
\subsection{Signatures}
\label{subsec: additional prelims Signatures}
We recall the syntax of digital signatures.

\begin{definition}[Signatures]\label{def: signatures}
    A signature scheme $\sig$ is defined as a tuple $(\gen, \sign, \verify)$ of the following three algorithms:
    \begin{itemize}
        \item $\gen(\secparam) \rightarrow (\sk, \pk)$: The probabilistic key generation algorithm returns a pair of secret key $\sk$ and public key $\pk$, where the public key defines the message space.
        \item $\sign(\sk, m) \rightarrow \sigma$: The probabilistic signature algorithm takes in a secret key and a message $m$, and outputs a signature $\sigma$.
        \item $\verify(\pk, m, \sigma) \rightarrow b$: The deterministic verification algorithm takes in the public key $\pk$, a message $m$ and a signature $\sigma$.
              It outputs either $\false$ ($b=0$) or $\true$ ($b=1$).
    \end{itemize}
    $\sig$ has correctness error $\epsilon$ if for all $(\pk, \sk) \in \supp{\gen}$ and any $m$:
    \[\prob{\verify(\pk, m, \sign(\sk, m)) \neq 1} \leq \epsilon,\]
    where the probability is taken over the randomness of $\sign$.
\end{definition}

\begin{definition}[$\seufcma$ of Signatures]
    Consider the unforgeability game in \cref{fig: digital signatures}.
    We say that a signature $\sig$ is strongly unforgeable under $\nrsignqueries$ chosen message attacks, if for any \ac{ppt} adversary $\adv$
    \[\advantage{\seufcma}{\adv, \nrsignqueries} := \prob{1 \gets \seufcma_{\adv, \sig, \nrsignqueries}} = \negl.\]
\end{definition}

\begin{figure}
    \begin{pchstack}[boxed, center, space=0.5em]
        \procedure[linenumbering]{$\seufcma_{\adv, \sig, \nrsignqueries}$}{
            q \gets 0; \mathcal{L} \gets \emptyset\\
            (\pk, \sk) \gets \gen(\secparam)\\
            (m^*, \sigma^*) \gets \adv^{\Oracle{\sign}}(\pk)\\
            \pcif (m^*, \sigma^*) \in \mathcal{L}: \\
            \t \pcreturn 0 \pccomment{return $\false$}\\
            \pcreturn \verify(\pk, m^*, \sigma^*)
        }
        \procedure[linenumbering]{$\Oracle{\sign}(m)$}{
            \pcif q = \nrsignqueries:\\
            \t \pcreturn \bot\\
            q \gets q + 1\\
            \sigma \gets \sign(\sk, m)\\
            \mathcal{L} \gets \mathcal{L} \cup \{(m, \sigma)\}\\
            \pcreturn \sigma
        }
    \end{pchstack}
    \caption{Unforgeability games for signatures.}\label{fig: digital signatures}
\end{figure}
\subsection{Ring Signatures}\label{subsec: additional prelims Ring Signatures}
\begin{definition}[Ring Signatures]\label{def: ring signature}
    A \acf{RSS} $\rsig$ is a quadruple of \ac{ppt} algorithms $(\setup, \kgen, \sign, \verify)$ such that:
    \begin{itemize}
        \item $\setup(\secparam) \rightarrow \pp$: The setup algorithm takes as input the security parameter $\secparam$ and outputs public parameters $\pp$.
              These parameters also define the message space and an upper bound on the ring size $\kappa$.
        \item $\kgen(\pp) \rightarrow (\pk, \sk)$: The key generation algorithm takes as input the public parameters $\pp$ and produces a pair of public and private keys $(\pk, \sk)$.
        \item
              $\sign(\sk, \rho, \msg) \rightarrow \sigma$: The signing algorithm takes as input the secret key of the signer $\sk$, a list of public keys $\rho = \{\pk_{1}, \ldots, \pk_{N}\}$ defining the ring and a message $\msg$.
              The ring must satisfy the size bound, i.e., $N \leq \kappa$, and $\exists i \in [N]: \mu(\sk) = \pk_i$.\footnote{We assume (w.l.o.g.) that there exists a one-way function $\mu$ such that $\mu(\sk) = \pk$ for all $(\pk, \sk) \in \supp{\kgen}$.}
              It outputs a signature $\sigma$.
        \item
              $\verify(\rho, \msg, \sigma) \rightarrow b$: The deterministic verification algorithm takes as input a list of public keys $\rho = \{\pk_{1}, \dots, \pk_{N}\}$ with $N \leq \kappa$, a message $\msg$ and a signature $\sigma$.
              It outputs either $\false$ ($b=0$) or $\true$ ($b=1$).
    \end{itemize}
    $\rsig$ is $\delta(\kappa)$-correct if for any $\pp \gets \setup(\secparam)$, all $N \leq \kappa$,  any $\{(\pk_j, \sk_j)\}_{j \in [N]} \subseteq \supp{\kgen(\pp)}$ defining $\rho = \{\pk_1, \ldots, \pk_N\}$, any $\msg$ and any $i \in [N]$
    \[\prob{\verify(\rho, \msg, \sign(\sk_{i},\rho, \msg))\neq 1
        } \leq \delta(\kappa).\]
\end{definition}

\paragraph{Unforgeability of Ring Signatures.}
As in standard EUF-CMA, a forgery adversary against a \ac{RSS} must output a forgery $\sigma^*$ for a message $\msg^*$ for a ring $\rho^* \subseteq \{\pk_{1}, \ldots, \pk_{N}\}$ of valid public keys.
The adversary can make adaptive queries to a signing oracle with a message, a signer position, and a list of public keys. The signer position must refer to a valid public key.
This is referred to as ``insider security''~\cite{JC:BenKatMor09}.
In the following, we treat the size bound for signing queries and verification queries as implicit, i.e.\ the ring size of queries cannot exceed $\kappa$.

\begin{figure*}
    \centering
    \begin{pchstack}[boxed, center, space=0.5em]
        \procedure[linenumbering]{$\seufcra_{\adv, \rsig, N, \nrsignqueries}(\secpar)$}{
            q \gets 0;\mathcal{L} \gets \emptyset\\
            \pp \gets \setup(\secparam)\\
            \pcfor i \in [N]:\\
            \t (\pk_{i},\sk_{i})\gets\kgen(\pp)\\
            \rho \gets \{\pk_{1}, \ldots, \pk_{N}\}\\
            (\rho^{*}, \msg^{*}, \sigma^{*}) \gets \adv^{\Oracle{\sign}}(\rho)\\
            \pcif (\rho^{*}, \msg^{*}, \sigma^{*}) \in \mathcal{L} \lor \rho^{*}\not\subseteq \rho:\\
            \t \pcreturn 0\pccomment{return \(\false\)}\\
            \pcreturn \verify(\rho^{*}, \msg^{*}, \sigma^{*})
        }
        \procedure[linenumbering]{$\Oracle{\sign}(i, \rho'=\{\pk_{1}', \ldots, \pk_{N'}'\}, \msg)$}{
            \pcif q = \nrsignqueries \lor \pk_{i}' \notin \rho:\\
            \t \pcreturn \bot\\
            \pclinecomment{let $j\in [N]$ be the index s.t. $\mu(\sk_j) = \pk_i'$}\\
            q \gets q+1\\
            \sigma \gets \sign(\sk_j, \rho', \msg)\\
            \mathcal{L} \gets \mathcal{L} \cup \{(\rho', \msg, \sigma)\}\\
            \pcreturn \sigma
        }

    \end{pchstack}
    \caption{Unforgeability game for ring signatures.}\label{fig: rsig unforgeability}
\end{figure*}

\begin{definition}[$\eufnra$ of Ring Signatures]\label{def: UFNRA}
    A \ac{RSS} $\rsig$ is unforgeable under no-ring attacks for $N \leq \kappa$ if, for any \ac{ppt} adversary $\adv$
    \[\advantage{\eufnra}{\adv, \rsig, N} := \prob{1 \gets \eufnra_{\adv, \rsig, N}(\secpar)}= \negl,\]
    where $\eufnra_{\adv, \rsig, N}=\seufcra_{\adv, \rsig, N, 0}$.
\end{definition}

\begin{definition}[$\seufcra$ of Ring Signatures]\label{def: SUFCRA}
    Consider the unforgeability game in \cref{fig: rsig unforgeability}.
    We say a \ac{RSS} $\rsig$ is strongly unforgeable under $\nrsignqueries$ chosen ring attacks for $N \leq \kappa$ if, for any \ac{ppt} adversary $\adv$
    \[\advantage{\seufcra}{\adv, \rsig, N, \nrsignqueries} := \prob{1 \gets \seufcra_{\adv, \rsig, N, \nrsignqueries}(\secpar)}= \negl.\]
\end{definition}

Consider the unforgeability game in \cref{fig: rsig unforgeability}.
Define the $\weufcraone$ game as the $\seufcra$ game, except that i) $\mathcal{L}$ contains pairs $(\rho,m)$ (without the signatures), ii) a forgery counts as fresh if it is for a pair $(\rho^*, m^*)\notin \mathcal L$ and iii) the signing oracle cannot be queried for $(\rho', \msg)$ already in $\mathcal{L}$.
\begin{definition}[$\weufcraone$ of Ring Signatures]\label{def: WUFCRA}
    We say a \ac{RSS} $\rsig$ is one-per-message weakly unforgeable under $\nrsignqueries$ chosen ring attacks for $N \leq \kappa$ if, for any \ac{ppt} adversary $\adv$
    \[\advantage{\weufcraone}{\adv, \rsig, N, \nrsignqueries} := \prob{1 \gets \weufcraone_{\adv, \rsig, N, \nrsignqueries}(\secpar)}= \negl.\]
\end{definition}

\paragraph{Anonymity of Ring Signatures.}
We consider \emph{anonymity under full key exposure} introduced in \cite{JC:BenKatMor09}.
As in~\cite{PKC:BFGJS22,C:GajJanKil24}, we parameterize the anonymity notion with $\nrchalqueries$, defining the number of allowed calls to the challenge oracle.

\begin{figure}
    \centering
    \begin{pchstack}[boxed, center, space=0.5em]
        \procedure[linenumbering]{$\anon_{\adv, \rsig, N, \nrchalqueries}(\secpar)$}{
            q \gets 0; b \gets \bin\\
            \pp \gets \setup(\secparam)\\
            \pcfor i \in [N]:\\
            \t (\pk_{i}, \sk_{i})\gets\kgen(\pp)\\
            \rho \gets {\{\pk_{i}\}}_{i \in [N]};\SK \gets {\{\sk_{i}\}}_{i \in [N]}\\
            b' \gets \adv^{\chal}(\rho, \SK)\\
            \pcreturn b=b'
        }

        \procedure[linenumbering]{$\chal(i_0, i_{1},\rho'=\left\{\pk_1', \ldots, \pk_{N'}'\right\}, \msg)$}{
            \pcif q = \nrchalqueries \lor \pk_{i_{0}}'\notin \rho \lor \pk_{i_{1}}' \notin \rho: \\
            \t\pcreturn \mathord{\perp}\\
            \pclinecomment{let $j\in [N]$ be the index s.t. $\mu(\sk_j) = \pk_{i_b}'$}\\
            q \gets q + 1\\
            \pcreturn \sign(\sk_j, \rho', \msg)
        }
    \end{pchstack}
    \caption{The anonymity game for \ac{RSS}s under full key exposure.}\label{rsig: anonymity under full key exposure}
\end{figure}

\begin{definition}[Anonymity of Ring Signatures]\label{def: anonymity}
    Consider the anonymity game in \cref{rsig: anonymity under full key exposure}.
    We say a \ac{RSS} $\rsig$ is anonymous under full key exposure with $\nrchalqueries$ challenge queries if, for $N\leq \kappa$ and any \ac{ppt} adversary $\adv$
    \[\advantage{\anon}{\adv, \rsig, N, \nrchalqueries} := \left|\prob{\anon_{\adv, \rsig, N, \nrchalqueries}(\secpar)\Rightarrow \true} - 1/2\right|= \negl\]
\end{definition}

\subsection{Sigma Protocols}\label{subsec: additional prelims Sigma Protocols}

\begin{definition}[$\mathsf{\Sigma}$-protocol]\label{def: sigma protocol}
    A $\mathsf{\Sigma}$-protocol $\Sigma=(\gen,\prover,\verifier)$ is a three-move identification protocol consisting of five \ac{ppt} algorithms $\gen, \prover_{1}, \prover_{2}, \verifier_{1}, \verifier_{2}$ where $\verifier_{2}$ is deterministic.
    We assume $\prover_{1}, \prover_{2}$  to share state, and likewise  $\verifier_{1}$ and $\verifier_{2}$.
    Denote the challenge space by $\mathcal{C}$.
    Then, $\Sigma$ behaves as follows.
    \begin{enumerate}
        \item The prover gets  a pair of instance and witness $(\instance, w) \gets \gen(\secparam)$ where $w \in \mathcal{W}$ and publishes its instance $\instance \in \mathcal{I}$.\footnote{We assume (w.l.o.g.) that there exists a function $\mu$ such that $\mu(w) = \instance$ for all $(\instance, w) \in \supp{\gen}$. If this is not the case, we can re-define the witness to contain a copy of the instance.}
        \item The prover runs $\prover_{1}$ with input $\instance$ to get a commitment $(\com, \state) \gets \prover_{1}(\instance)$, and sends $\com$ to the verifier.
        \item The verifier runs $\verifier_{1}$ with input $\instance, \com$, and generates an independent random challenge from the challenge set $\mathcal{C}$ and sends it to the prover.
        \item The prover runs $\prover_{2}$ with input $\state, w, \ch$ to compute a response $\rsp \gets \prover_{2}(\state, w, \ch)$ for the commitment $\com$, and sends it to the verifier.
        \item The verifier runs $\verifier_{2}$ with input $\instance, \com, \ch, \rsp$ and outputs either $\false$ or $\true$.
    \end{enumerate}
    A protocol transcript $(\com, \ch, \rsp)$ is said to be valid in case $\verifier_{2}(\instance, \com, \ch, \rsp)$ outputs $\true$.
    We say a $\mathsf{\Sigma}$-protocol is complete, if
    \[\prob{\verifier_{2}(\instance, \com, \ch, \rsp) : \substack{
                (\instance, w) \gets \gen(\secparam);\
                (\com, \state) \gets \prover_{1}(\instance);\\
                \ch \gets \verifier_{1}(\instance, \com);\
                \rsp \gets \prover_{2}(\state, w, \ch)
            }} = 1.\]
\end{definition}

If a response and a challenge uniquely define the commitment, then we call the $\mathsf{\Sigma}$-protocol \emph{commitment-recoverable}.

\paragraph{Security for the Verifier.}

The verifier needs to be convinced that the prover knows the witness.
Therefore, an adversary without knowing the witness should not be able to impersonate an honest prover.

\begin{definition}[Impersonation under Passive Attacks]\label{def: impersonation}
    We say that a $\mathsf{\Sigma}$-protocol $\Sigma$ is secure against impersonations under passive attacks, if for any \ac{ppt} $\adv = (\adv_1, \adv_2)$ and honest verifier $\verifier = (\verifier_1, \verifier_2)$, we have
    \[\advantage{\imp}{\adv, \Sigma} := \prob{\verifier_{2}(\instance, \com, \ch , \rsp) : \substack{
                (\instance, w) \gets \Sigma.\gen(\secparam);\\
                (\com, \state) \gets \adv_1(\instance);\\
                \ch \gets \verifier_1(\instance, \com);\\
                \rsp \gets \adv_2(\state, \ch)
            }
        } = \negl.\]
\end{definition}

An important security notion for $\mathsf{\Sigma}$-protocols is  \emph{special soundness}.

\begin{definition}[Special Soundness]\label{def: special soundness}
    We say that a $\mathsf{\Sigma}$-protocol $\Sigma$ is \emph{special sound} if there exists an \ac{ppt} $\extractor$ that on input of any $\instance$ from $(\instance, w) \gets \gen(\secparam)$ and any two valid transcripts $(\com, \ch_{0}, \rsp_{0})$ and $(\com, \ch_{1}, \rsp_{1})$ with $\ch_{0} \neq \ch_{1}$, outputs some witness $w'$ such that $(\instance, w') \in R$.
\end{definition}

A $\mathsf{\Sigma}$-protocol has unique responses if there is only one valid response for a commitment-challenge-pair.
A relaxation is computational unique responses.

\begin{definition}[Computational Unique Responses]\label{def: computationally unique responses}
    We say that a $\mathsf{\Sigma}$-protocol $\Sigma$ has \emph{computational unique responses} if for any \ac{ppt} $\adv$, we have
    \[\advantage{\cur}{\adv, \Sigma} := \prob{
            \begin{array}{c}\rsp_0 \neq \rsp_1                           \\
                \land \verifier_2(\instance, \com, \ch, \rsp_0) \\
                \land \verifier_2(\instance, \com, \ch, \rsp_1)
            \end{array}
            :
            \substack{
                (\instance, w) \gets \gen(\secparam); \\
                (\com, \ch, \rsp_0, \rsp_1)\gets \adv(\instance)
            }} = \negl.\]
\end{definition}

\paragraph{Security for the Prover.}
The prover wants to ensure that the witness is not revealed to others.
There are two levels to this.
First, the instance should not reveal anything about the witness, and secondly, the commitments and responses should not reveal anything about the witness.

\begin{definition}[Witness Recovery]\label{def: witness recovery}
    We say that a $\mathsf{\Sigma}$-protocol $\Sigma$ for a witness relation $R$ is secure against witness recovery, if for any \ac{ppt} $\adv$, we have
    \[\advantage{\wrec}{\adv, \Sigma} := \prob{(\instance, w') \in R :
            \substack{
                (\instance, w) \gets \gen(\secparam); \\
                w' \gets \adv(\instance)
            }
        } = \negl.\]
\end{definition}

\begin{definition}[Honest Verifier Zero Knowledge]\label{def: HVZK}
    We say that a $\mathsf{\Sigma}$-protocol $\Sigma$ is \emph{statistically} special $\hvzktv$ \ac{HVZK} if there exists a \ac{ppt} simulator, such that for any $\secpar, (\instance, w) \gets \gen(\secparam)$ and any challenge $\ch \in \mathcal{C}$, we have
    \[\Delta(P_{\ch}(\instance, w),\simulator(\instance, \ch)) \leq \hvzktv \]
    where $P_{\ch}(\instance, w)$ denotes the distribution of honestly generated transcripts according to the algorithm's description, conditioned on having challenge $\ch$.
    Furthermore, we call it $\hvzkrenyi$-divergence \ac{HVZK} for $\alpha \in (1, \infty]$ if
    \[R_{\alpha}(P_\ch(\instance, w)\|\simulator(\instance, \ch)) \leq \hvzkrenyi.\]
    We also define the Kullback-Leibler divergence for \ac{HVZK} as
    \[D_{KL}(P_\ch(\instance, w)\|\simulator(\instance, \ch)) \leq \hvzkkl.\]
\end{definition}

A simulator produces commitments and responses.
However, the distribution of commitments in the range of all commitments is not immediate, so we consider the notion of commitment min-entropy for a \ac{HVZK} simulator for \emph{any} input.

\begin{definition}[Simulator Commitment Min-Entropy]\label{def: sim commitment-min-entropy}
    We say that simulator $\simulator$ for $\Sigma$ has $\minentropy(\secpar)$ commitment min-entropy, if
    \[\max_{C, \ch \in \mathcal{C}, \instance \in \{0,1\}^*}\prob{C = \com: (\com, \rsp) \gets \simulator(\instance, \ch)} \leq 2^{-\minentropy(\secpar)}.\]
\end{definition}

For protocols with statistical \ac{HVZK}, the notion is equivalent to the more common notion of commitment min-entropy.     \subsection{Quantum Computation Basics}
The length of a vector $\ket{\psi} \in \CC^d$ can be measured with the Euclidean norm
\[\euclnorm{\ket{\psi}} := \sqrt{\braket{\psi}{\psi}}.\]
Given an operator $M \in \mathcal{L}(\CC^d)$, define its trace norm and operator norm as
\[\tracenorm{M} :=  \trace{\sqrt{M^\dagger M}} \quad\text{and}\quad \operatornorm{M} := \max_{\text{unit vectors}\, \ket{\psi} \in \CC^d}\euclnorm{M\ket{\psi}}.\]
For a vector $\ket{\psi} \in \CC^d$ and two unitaries $U, V$ on $\CC^d$ it holds that
\begin{equation}
    \label{eq: trace norm same initial state}\tracenorm{U^\dagger\ketbra{\psi}{\psi}U - V^\dagger \ketbra{\psi}{\psi}V} \leq 2 \operatornorm{U - V}.
\end{equation}
Both the trace and the operator norm are also invariant under multiplication with unitaries.
Any rank-1 operator $M = \ketbra{\psi}{\phi} \in \mathcal{L}(\CC^d)$ satisfies
\begin{equation}
    \label{eq: rank 1 operator norm}\operatornorm{M} = \euclnorm{\ket{\psi}}\euclnorm{\ket{\phi}}.
\end{equation}

\subsubsection{Grover Search}\label{subsubsec: Grover Search}
Let $f: \bin^n \to \bin$ be an $n$-bit function.
An element $x \in \bin^n$ is marked, if and only if $f(x) = 1$.
Grover's algorithm~\cite{STOC:Grover96} allows for finding a marked element in time $\bigO{\sqrt{N/M}}$ where $M$ is the number of marked elements and $N = 2^n$.
We follow the notation of \cite{nielsen2010quantum} in this work.
The initial state of Grover's algorithm is an equal superposition
$\ket{\psi} = \frac{1}{\sqrt{N}}\sum_{x \in \bin^n}\ket{x}$.
By separating the marked and unmarked elements, the state can be written as
\[
    \ket{\psi}
    =
    \sqrt{\frac{N-M}{N}}\ket{\alpha}
    +
    \sqrt{\frac{M}{N}}\ket{\beta},
\]
where
\[
    \ket{\alpha}
    =
    \frac{1}{\sqrt{N-M}}
    \sum_{\substack{x \in \bin^n \\ f(x)=0}}
    \ket{x},
    \qquad
    \ket{\beta}
    =
    \frac{1}{\sqrt{M}}
    \sum_{\substack{x \in \bin^n \\ f(x)=1}}
    \ket{x}.
\]
This is interpreted geometrically by defining an angle $\theta$ such that $\cos(\theta/2) = \sqrt{(N-M)/N}$, which yields
$\ket{\psi} = \cos\left(\frac{\theta}{2}\right)\ket{\alpha} + \sin\left(\frac{\theta}{2}\right)\ket{\beta}$.
The Grover iteration $G$ applied $q$ times takes the state to
\begin{equation}\label{eq: Grover state after q iterations}
    \ket{\psi^{(q)}} := G^q\ket{\psi} = \cos\left(\frac{2q+1}{2}\theta\right)\ket{\alpha} + \sin \left(\frac{2q + 1}{2}\theta\right)\ket{\beta}.
\end{equation}     \subsection{NTRU Lattices}\label{subsec: NTRU Lattices}
We introduce the lattice notation used in this appendix. Most definitions are adapted from \cite{C:GajJanKil24,EPRINT:GajJanKil24b}, with details omitted where they are not essential.

We work with polynomial rings of the form $\mathcal{R} = \ZZ[X]/f(X)$ where $f(X) = X^M +1$ and $M$ is a power of two, i.e., $\exists k: M=2^k$.
Polynomials are denoted by lower-case bold letters.
We write $\mathcal{R}_q = \ZZ_q[X]/f(X)$.

A lattice $\bm{\Lambda} \subseteq \RR^n$ is defined by a basis $\bm{B} = \{\bm{b}_1, \ldots, \bm{b}_m\}$ of linearly independent vectors $\bm{b}_i \in \RR^n$.
They form the lattice
\[\bm{\Lambda}(\bm{B}) = \bm{B}\ZZ^m = \left\{\sum_{i \in [m]}c_i \bm{b}_i : c_1, \ldots, c_m \in \ZZ\right\}.\]
For a basis $\bm{B}$, denote the Gram-Schmidt norm from \cite{STOC:GenPeiVai08} as $\norm{\bm{B}}_{GS} = \max_{i \in [m]} \|\tilde{\bm{b}}_i\|$, where $\tilde{\bm{B}}$ is the Gram-Schmidt orthogonalization of $\bm{B}$.

A discrete Gaussian over a lattice $\bm{\Lambda}$ is defined via the $n$-dimensional Gaussian function $\rho_{s, \bm{c}}: \RR^n \to (0,1]$ on $\RR^n$ centered at $\bm{c} \in \RR^n$ with standard deviation $s > 0$:
\[\rho_{s,\bm{c}}(\bm{x}) := \exp\left({-\frac{\euclnorm{\bm{x}-\bm{c}}^2}{2s^2}}\right).\]
A discrete Gaussian over $\bm{\Lambda}$ is the Gaussian function normalized by all points of the lattice
\[\forall \bm{x} \in \bm{\Lambda}: \mathcal{D}_{\bm{\Lambda}, s, \bm{c}}(\bm{x}):=\frac{\rho_{s, \bm{c}}(\bm{x})}{\sum_{\bm{z} \in \bm{\Lambda}}\rho_{s, \bm{c}}(\bm{z})}.\]
We may omit $\bm{c}$ if it is $\bm{0}$, and we may omit $\bm{\Lambda}$, if it is $\ZZ^M$.
We write $\bm{f} \gets \mathcal{D}_{\mathcal{R}, s}$ to denote the polynomial $\bm{f}$ whose coefficient embedding\footnote{This embedding identifies polynomials in $\mathcal{R}$ with their coefficient vectors in $\ZZ^M$.} is sampled as $\mathcal{D}_{\ZZ^M, s}$.

W denote by $\eta_{\epsilon}(\bm{\Lambda})$  the smoothing parameter of a lattice $\bm{\Lambda}$ (introduced and properly defined in \cite{FOCS:MicReg04}).
Intuitively, the smoothing parameter is the smallest $s$ such that a discrete Gaussian distribution with standard deviation $s$ looks ``smooth'' (nearly uniform), when it is wrapped over the quotient $\RR^n/\bm{\Lambda}$.

\begin{lemma}[{\cite[Lemma 2.11]{TCC:PeiRos06}}]\label{lemma: min-entropy shifted discrete Gaussians}
    For any $n$-dimensional lattice $\bm{\Lambda}$, center $\bm{c} \in \RR^n$, $\epsilon >0$, and $s \geq 2 \eta_{\epsilon}(\bm{\Lambda})$, and for every $\bm{x} \in \Lambda$, we have
    \[\mathcal{D}_{\bm{\Lambda}, s, \bm{c}}(\bm{x}) \leq \frac{1+\epsilon}{1-\epsilon}2^{-n}.\]
    In particular, for $\epsilon < 1/3$, the min-entropy is at least $n-1$.
\end{lemma}

\begin{lemma}[{\cite{EC:Lyubashevsky12,banaszczyk1993new}}]\label{lemma: length bound discrete Gaussian}
    Let $n > 1$ and $s > 0$.
    \begin{enumerate}
        \item For any $\tau > 0$, $\prob{|z| > \tau s: z \gets \mathcal{D}_{\ZZ, s}} \leq 2e^{\frac{-\tau^2}{2}}$.
        \item For any $\tau > 1$, $\prob{\euclnorm{\bm{z}} > \tau s\sqrt{n}: \bm{z} \gets \mathcal{D}_{\ZZ^n, s}} \leq \tau^ne^{\frac{n}{2}(1-\tau^2)}$.
    \end{enumerate}
\end{lemma}

\begin{definition}[NTRU Lattice]
    Let $M$ be a power of two, $q$ prime, $\bm{f}, \bm{g} \in \mathcal{R} = \ZZ[X]/(X^M + 1)$ and $\bm{h} = \bm{g}\bm{f}^{-1} \mod q$.
    The NTRU lattice of $\bm{h}$ is defined by the coefficient embedding of
    \[\{(\bm{u}, \bm{v}) \in \mathcal{R}^2 \mid \bm{h}\bm{u} + \bm{v} = \bm{0}\mod q\}.\]
    We denote this lattice as $\bm{\Lambda}_{\bm{h}, q}$.
\end{definition}
The basis of the lattice can be computed easily from the anticirculant matrix of $\bm{h}$.
The sampling procedure also generates $\bm{h}$ typically such that it is in $\mathcal{R}_q^{*}$, i.e., that it is coprime to $q$.

\begin{lemma}[NTRU Trapdoor Generation~{\cite{HofPipSil98,prestthesis}}]\label{lem: NTRU trapgen}
    Let $\mathcal{R} = \ZZ[X]/(X^M + 1)$.
    There exists a \ac{ppt} algorithm, $\ntrutrapgen(q, \alpha)$, that given a modulus $q$, and a target quality $\alpha$, returns a public key $\bm{h} \in \mathcal{R}_q$, and the trapdoor $(\bm{f}, \bm{g}) \in \mathcal{R} \times \mathcal{R}$ such that they each define a basis $\bm{B_h}$ and $\bm{B_{f,g}}$ for the same lattice.
    Furthermore, $\norm{\bm{B_{f,g}}}_{GS} \leq \alpha \sqrt{q}$.
\end{lemma}

\begin{corollary}[{\cite[Corollary 2]{EPRINT:GajJanKil24b}}]\label{cor: Renyi uniformity of NTRU}
    Let $q$ be prime, $\bm{h} \in \mathcal{R}_q\setminus\{\bm{0}\}, \alpha \in (1, \infty), \epsilon \in (0,1/2), s \geq \eta_{\epsilon(\bm{\Lambda}_{\bm{h}, q})}$ and $Q$ be the distribution of $\bm{h}\bm{u} + \bm{v}$ where $\bm{u}, \bm{v} \gets \mathcal{D}_{\mathcal{R}, s, \bm{0}}$.
    Then it holds that
    \[R_{\alpha}(\mathcal{U}(\mathcal{R}_q)\|Q) \lesssim 1 + \frac{2\alpha \epsilon^2}{(1-\epsilon)^2}.\]
\end{corollary}

\begin{corollary}[{\cite[Corollary 1 that uses \cite{prestthesis}]{C:GajJanKil24}}]\label{cor: presampler KL and Rényi divergence}
    Let $M$ be a positive integer, $\alpha > 1$ and $\epsilon \in (0, 1/4)$.
    Then, there exists a preimage sampling algorithm $\ntrusamplepre(\bm{B}, s, \bm{c})$, such that for any basis $\bm{B}$, standard deviation $s \geq \eta_{\epsilon(\ZZ^{2M})}\cdot \|\bm{B}\|_{GS}$ and arbitrary syndrome $\bm{c}$, the Kullback-Leibler divergence and Rényi divergence are bounded as
    \[D_{KL}(\ntrusamplepre(\bm{B}, s, \bm{c})\|\mathcal{D}_{\bm{\Lambda}(\bm{B}), s, \bm{c}}) \lesssim 2\epsilon^2\]
    and
    \[R_\alpha(\ntrusamplepre(\bm{B}, s, \bm{c})\|\mathcal{D}_{\bm{\Lambda}(\bm{B}), s, \bm{c}}) \lesssim 1 + 2\alpha\epsilon^2.\]
\end{corollary}

We define two hardness problems for NTRU lattices with respect to the trapdoor algorithm $\ntrutrapgen$.
\begin{definition}[NTRU-SIS]
    Let $\mathcal{R} := \ZZ[X]/(X^M + 1)$.
    The \emph{short integer solution} problem relative to the NTRU trapdoor algorithm $\ntrutrapgen$ with parameters $m,q,\alpha,\beta > 0$ is defined via the $\ntrursis$ game from \cref{fig: ntrugames}.
    We define the advantage of $\adv$ in $\ntrursis$ as
    \[\advantage{\ntrursis}{\adv, m,q,\alpha, \beta} := \prob{1 \gets \ntrursis_{\adv,m,q,\alpha,\beta}}.\]
\end{definition}

\begin{definition}[{$t$}-NTRU-ISIS]
    Let $\mathcal{R} := \ZZ[X]/(X^M + 1)$.
    The $t$-target \emph{inhomogeneous short integer solution} problem relative to the NTRU trapdoor algorithm $\ntrutrapgen$ with parameters $m,q,\alpha,\beta > 0$ is defined via the $\tntrurisis{t}$ game from \cref{fig: ntrugames}.
    We define the advantage of $\adv$ in $\tntrurisis{t}$ as
    \[\advantage{\tntrurisis{t}}{\adv, m,q,\alpha, \beta} := \prob{1 \gets \tntrurisis{t}_{\adv,m,q,\alpha,\beta}}.\]
    We may omit $t$, if it is $1$.
\end{definition}

\begin{definition}[NTRU-SPISIS$^*$ {\cite[Slightly modified from Definition 10]{EPRINT:GajJanKil24b}}]
    Let $\mathcal{R} := \ZZ[X]/(X^M + 1)$and let $t \geq 1$. The \emph{second preimage inhomogeneous short integer solution} problem relative to the NTRU trapdoor algorithm $\ntrutrapgen$ with parameters $q, \alpha, \beta > 0$ is defined via the $\tntrurspisis{t}$ game from \cref{fig: ntrugames}.
    We define the advantage of $\adv$ in $\tntrurspisis{t}$ as
    \[\advantage{\tntrurspisis{t}}{\adv, q, \alpha, \beta} := \prob{1 \gets \tntrurspisis{t}_{\adv, q, \alpha, \beta}}.\]
\end{definition}

\begin{remark}[Relation to {\cite[Definition 10]{EPRINT:GajJanKil24b}}]
    Our formulation differs from \cite[Definition 10]{EPRINT:GajJanKil24b} in two respects.
    First, the syndromes are not uniform: we sample $\bm{u}_i, \bm{v}_i \gets \mathcal{D}_{\mathcal{R},s}$ and set $\bm{c}_i := \bm{h}\bm{u}_i + \bm{v}_i$, so that $\bm{c}_i$ follows $\mathcal{Q}_{\bm{h}}$ and, conditioned on $\bm{c}_i$, the pair $(\bm{u}_i, \bm{v}_i)$ follows the coset Gaussian.
    These are exactly the distributions against which $r_u$ and $r_p$ are measured.
    Second, the game performs no rejection: only the returned preimage must be short.
    Rejection is done by the reduction instead, which reprograms every preimage sampling attempt and discards attempt $i$ exactly when $\euclnorm{(\bm{u}_i, \bm{v}_i)} > \beta$, mirroring the loop of $\sign^{+}$.

    A divergence hop requires unconditioned samples: conditioning them costs up to $p^{-\alpha/(\alpha-1)}$ each.
    In \cite{EPRINT:GajJanKil24b} the conditioning is therefore never inside a hop.
    It is either a stopping rule applied identically on both sides of the hop, as in their $\game{3} \to \game{4}$ and in their Theorem 1, where the random oracle is likewise programmed on every iteration, or it is introduced statistically, as in the transition preceding their reduction.
    Doing something similar in the \ac{QROM} would place a conditioned distribution inside the reprogramming hop which is not possible with how we use the Rényi divergence.

    The two assumptions are nevertheless close: on the targets whose given preimage happens to be short they differ only in the syndrome being drawn from $\mathcal{Q}_{\bm{h}}$ rather than $\mathcal{U}(\mathcal{R}_q)$, and the remaining targets have a large preimage as the only help.
    We therefore expect the hardness discussion of \cite[Section 5.1]{EPRINT:GajJanKil24b} to apply here as well.
\end{remark}

\begin{figure}
    \begin{pcvstack}[boxed, center, space=0.5em]
        \procedure[linenumbering]{$\ntrursis_{\adv,m,q,\alpha,\beta}$}{
            \pcfor i \in [m]:\\
            \t (\bm{h}_i, \cdot, \cdot) \gets \ntrutrapgen(q, \alpha)\\
            (\bm{u}_1, \bm{u}_2, \ldots, \bm{u}_m, \bm{v}) \gets \adv(\bm{h}_1, \ldots, \bm{h}_m)\\
            \pcreturn \bm{v} + \sum_{i \in [m]}\bm{h}_i \cdot \bm{u}_i = \bm{0} \land \euclnorm{(\bm{u}_1, \bm{u}_2, \ldots, \bm{u}_m, \bm{v})} \leq \beta
        }
        \procedure[linenumbering]{$\tntrurisis{t}_{\adv,m,q,\alpha,\beta}$}{
            \pcfor i \in [m]:\\
            \t (\bm{h}_i, \cdot, \cdot) \gets \ntrutrapgen(q, \alpha)\\
            \pcfor j \in [t]:\\
            \t \bm{c}_j \gets \mathcal{R}_q\\
            (j, \bm{u}_1, \bm{u}_2, \ldots, \bm{u}_m, \bm{v}) \gets \adv(\bm{h}_1, \ldots, \bm{h}_m, \bm{c}_1, \ldots, \bm{c}_t)\\
            \pcreturn \bm{v} + \sum_{i \in [m]}\bm{h}_i \cdot \bm{u}_i = \bm{c}_j \land \euclnorm{(\bm{u}_1, \bm{u}_2, \ldots, \bm{u}_m, \bm{v})} \leq \beta
        }
        \procedure[linenumbering]{$\tntrurspisis{t}_{\adv, q, \alpha, \beta}$}{
            (\bm{h}, \cdot, \cdot) \gets \ntrutrapgen(q, \alpha)\\
            \pcfor i \in [t]:\\
            \t \bm{u}_i, \bm{v}_i \gets \mathcal{D}_{\mathcal{R}, s}\\
            \t \bm{c}_i := \bm{h}\bm{u}_i + \bm{v}_i \in \mathcal{R}_q\\
            (j, \bm{u}, \bm{v}) \gets \adv(\bm{h}, \{(\bm{c}_i, \bm{u}_i, \bm{v}_i)\}_{i \in [t]})\\
            \pcreturn \bm{h}\bm{u} + \bm{v} = \bm{c}_j \land \euclnorm{(\bm{u}, \bm{v})} \leq \beta \land (\bm{u}, \bm{v}) \neq (\bm{u}_j, \bm{v}_j)
        }
    \end{pcvstack}
    \caption{Games defining $\ntrursis_{\adv,m,q,\alpha,\beta}, \tntrurisis{t}_{\adv,m,q,\alpha,\beta}$ and $\tntrurspisis{t}_{\adv, q, \alpha, \beta}$.}\label{fig: ntrugames}
\end{figure}

\section{Omitted Proofs}\label{sec: omitted proofs}

\subsection{Proof of \cref{thm: Oracle Distribution Switching using Statistical Distance}}\label{subsec: Proof Oracle Distribution Switching using Statistical Distance}

\OracleDistributionSwitchingSD*

We rely on techniques and perspectives on the \ac{QROM} from \cite{EC:CFHL21} utilizing the compressed oracle framework from \cite{C:Zhandry19}.
We recall notation from \cite{EC:CFHL21}, simplified in \cite{EC:DFMS22}.
Refer to the original work for a more detailed introduction.

\newcommand{\initialstate}{\hat{0}}
Consider the multiregister $D = (D_x)_{x \in \mathcal{X}}$, where the state space of $D_x$ is given by $\mathcal{H}_{D_x} = \CC[\bin^n \cup \{\bot\}]$, meaning that it is spanned by an orthonormal set of vectors $\ket{y}$ labelled by $y \in \bin^n \cup \{\bot\}$.
The initial state is set to be $\ket{\bm{\bot}} := \bigotimes_x\ket{\bot}_{D_x}$.
Let $P$ be a probability distribution over $\bin^n$.
Define the state
$\ket{\phi_P} := \sum_{y \in \bin^n} \sqrt{P(y)} \ket{y}$
and the unitary
\[F^P = \ketbra{\bot}{\phi_P} + \ketbra{\phi_P}{\bot} + I_{\bin^n} - \ketbra{\bot}{\bot} - \ketbra{\phi_P}{\phi_P}.\]
\newcommand{\singlecompressedoraclequary}[1]{O_{YD_x}^{#1, x}}
\newcommand{\compressedoraclequery}[1]{O_{XYD}^{#1}}
When the compressed oracle is queried, a unitary $\compressedoraclequery{P}$, acting on the query registers $X$ and $Y$ and the oracle register $D$, is applied, given by
\[\compressedoraclequery{P} = \sum_{x \in \mathcal{X}} \ketbra{x}{x}_X \otimes \singlecompressedoraclequary{P},\quad \text{with} \quad \singlecompressedoraclequary{P} = F_{D_x}^P\CNOT_{YD_x}F_{D_x}^P\]
where $\CNOT_{YD_x}\ket{y}\ket{y_x} = \ket{y \oplus y_x}\ket{y_x}$ for $y, y_x \in \bin^n$ and acting as the identity on $\ket{y}\ket{\bot}$.
Instantiated with the uniform distribution over $\bin^n$, the formalism describes the ``compressed'' oracle as formalized in \cite{EC:DFMS22}.

A $q$-query adversary $\adv$ querying the compressed oracle with underlying distribution $P$ can be described as in \cref{fig: General Strategy q query adversary compressed oracle}
where the measurement describes the final output of $\adv$.
The adversary describes the strategy via a unitary $U$\footnote{One can also consider a sequence of unitaries $U_i$, but they can be combined into a single unitary with a counter and a controlled unitary to get the same behavior.} on $LXY$ and the eventual outcome is described by a measurement on the state
\[\ket{\psi_P} := \underbrace{\left(U \compressedoraclequery{P}\right)^{q}U}_{:= V_P}\ket{0}_{LXY}\ket{\bm{\bot}}_{D}.\]
\begin{figure*}
    \centering
    \begin{quantikz}
        \lstick{$\ket{\bot}_{D}$} &  & \gate[wires=2]{\compressedoraclequery{P}} & &\dots & \gate[wires=2]{\compressedoraclequery{P}} &\qw&\qw\\
        \lstick{$\ket{0}_{XY}$} & \gate[wires=2]{U}      &   & \gate[wires=2]{U}&\dots & &\gate[wires=2]{U} &\meter[wires=2]{}\\
        \lstick{$\ket{0}_{L}$} &       &   & &\dots & & &
    \end{quantikz}
    \caption{General $q$-query quantum adversary using oracle with the underlying distribution of $P$.}\label{fig: General Strategy q query adversary compressed oracle}
\end{figure*}

\newcommand{\coswitch}{\operatorname{Comp}}
\newcommand{\initialstatep}{\initialstate_P}

\subsubsection{Quantum Oracle Switching via Statistical Distance and the Compressed Oracle}\label{subsubsec: Statistical Distance with the Compressed Oracle}

We are interested in the trace distance of $\ketbra{\psi_P}{\psi_P}$ and $\ketbra{\psi_Q}{\psi_Q}$ where we are given the statistical distance of $P$ and $Q$.

\begin{lemma}\label{lemma: trace distance co}
    Let $\epsilon$ be the statistical distance of $P$ and $Q$ over $\bin^n$.
    Then,
    \[\|\ketbra{\psi_Q}{\psi_Q} - \ketbra{\psi_P}{\psi_P}\|_{tr} \leq 8q \hellingerdistance(P,Q) \leq 8q\sqrt{\epsilon}.\]
\end{lemma}
\begin{proof}
    Both $\psi_Q$ and $\psi_P$ are constructed by applying $V_Q$ or $V_P$ to the same initial state.
    The trace norm can now be bounded by the operator norm of the unitaries that are applied to the initial state as in \cref{eq: trace norm same initial state}
    \[\tracenorm{\ketbra{\psi_Q}{\psi_Q} - \ketbra{\psi_P}{\psi_P}} \leq 2 \operatornorm{V_Q - V_P}.\]
    By adding $2q$ null terms where we replace one $F_Q$ with $F_P$, we get that the operator norm is bounded by the sum of the operator norms of the compressed oracle switches.

    Let $j \in [2q]$ be the oracle switch, we are replacing in the current step.
    Now, add the null-term of the form
    \[\pm \left(U\compressedoraclequery{P}\right)^{q - \lfloor (j+1)/2\rfloor}UW_j\left(U\compressedoraclequery{Q}\right)^{\lfloor (j-1)/2\rfloor}U\]
    where
    \[W_j := \begin{cases}
            \sum_{x} \ketbra{x}{x}\otimes F_{D_x}^P \CNOT_{YD_x} F_{D_x}^Q & \text{if}\, j \equiv 1 \mod 2, \\
            \sum_{x} \ketbra{x}{x}\otimes F_{D_x}^Q \CNOT_{YD_x} F_{D_x}^Q & \text{if}\, j \equiv 0 \mod 2. \\
        \end{cases}\]
    Now, we consider the original operator norm and add this nullterm for every $j \in [2q]$.
    By using the triangular inequality and the fact that the operator norm is invariant under application of unitaries, it follows that
    \[ \operatornorm{V_Q - V_P} \leq 2q \cdot \operatornorm{\sum_{x} \ketbra{x}{x}\otimes(F_{D_x}^Q - F_{D_x}^P)}.\]
    We bound it in the next lemma.
\end{proof}

\begin{lemma}
    Let $\epsilon$ be the statistical distance of $P$ and $Q$ over $\bin^n$.
    Then,
    \[\operatornorm{\sum_{x} \ketbra{x}{x}\otimes(F_{D_x}^Q - F_{D_x}^P)} \leq 2\hellingerdistance(P,Q) \leq 2\sqrt{\epsilon}.\]
\end{lemma}

\begin{proof}
    \begin{align*}
         & \operatornorm{\sum_{x} \ketbra{x}{x}\otimes(F_{D_x}^Q - F_{D_x}^P)} = \max_{x}\operatornorm{F_{D_x}^Q - F_{D_x}^P}                                                 \\
         & = \operatornorm{\ketbra{\bot}{\phi_Q} - \ketbra{\bot}{\phi_P} + \ketbra{\phi_Q}{\bot} - \ketbra{\phi_P}{\bot} + \ketbra{\phi_P}{\phi_P} - \ketbra{\phi_Q}{\phi_Q}} \\
         & = \operatornorm{\ket{\bot}(\bra{\phi_Q} - \bra{\phi_P}) + (\ket{\phi_Q} - \ket{\phi_P})\bra{\bot} + \ketbra{\phi_P}{\phi_P} - \ketbra{\phi_Q}{\phi_Q}}             \\
    \end{align*}
    We set $\ket{\Delta_{+}} := \ket{\phi_Q} + \ket{\phi_P}$ and $\ket{\Delta_{-}} := \ket{\phi_Q} - \ket{\phi_P}$.
    Now we express the above using these orthogonal states.
    \begin{align*}
        \operatornorm{\sum_{x} \ketbra{x}{x}\otimes(F_{D_x}^Q - F_{D_x}^P)} & =  \big\|\ket{\bot}\bra{\Delta_{-}} + \ket{\Delta_{-}}\bra{\bot}                                   \\
                                                                            & \quad + 1/4(\ket{\Delta_{+}} - \ket{\Delta_{-}})(\bra{\Delta_{+}} - \bra{\Delta_{-}})              \\
                                                                            & \quad - 1/4(\ket{\Delta_{+}} + \ket{\Delta_{-}})(\bra{\Delta_{+}} + \bra{\Delta_{-}})\big\|_\infty
    \end{align*}
    Considering an orthonormal basis that includes normalized versions of $\ket{\bot}, \ket{\Delta_{-}}, \ket{\Delta_{+}}$, it is sufficient to consider how the operator behaves:
    \begin{align*}
        \ket{\bot}                                      & \mapsto \ket{\Delta_{-}}                                                                                                                                                                        \\
        \frac{1}{\euclnorm{\Delta_{-}}}\ket{\Delta_{-}} & \mapsto \frac{1}{\euclnorm{\Delta_{-}}}\left(\braket{\Delta_{-}}{\Delta_{-}}\ket{\bot} - 1/2\braket{\Delta_{-}}{\Delta_{-}}\ket{\Delta_{+}}\right)                                              \\
                                                        & = \euclnorm{\Delta_{-}}\left(\ket{\bot} - 1/2\ket{\Delta_{+}}\right)                                                                                                                            \\
        \frac{1}{\euclnorm{\Delta_{+}}}\ket{\Delta_{+}} & \mapsto \frac{1}{\euclnorm{\Delta_{+}}}\left(-1/2\braket{\Delta_{+}}{\Delta_{+}}\ket{\Delta_{-}}                                             \right)= -1/2\euclnorm{\Delta_{+}}\ket{\Delta_{-}} \\
    \end{align*}
    In matrix form, we thus get
    \begin{align*}
        M = \left(\begin{array}{ccc}
                      0                     & \euclnorm{\Delta_{-}}                                & 0                                                    \\
                      \euclnorm{\Delta_{-}} & 0                                                    & -\frac 12 \euclnorm{\Delta_{-}}\euclnorm{\Delta_{+}} \\
                      0                     & -\frac 12 \euclnorm{\Delta_{-}}\euclnorm{\Delta_{+}} & 0
                  \end{array}\right).
    \end{align*}
    The operator $F_{D_x}^Q - F_{D_x}^P$ is Hermitian, so $M$ is (real) symmetric, and its operator norm is the largest absolute value among its eigenvalues.
    Writing $a := \euclnorm{\Delta_{-}}$ and $b := \frac12\euclnorm{\Delta_{-}}\euclnorm{\Delta_{+}}$, the characteristic polynomial of $M$ is $-\lambda(\lambda^{2} - (a^{2}+b^{2}))$, so its eigenvalues are $0$ and $\pm\sqrt{a^{2}+b^{2}}$, giving $\operatornorm{M} = \sqrt{a^2+b^2} = \euclnorm{\Delta_{-}}\sqrt{1+\frac 14\euclnorm{\Delta_{+}}^2}$.
    Trivially $\euclnorm{\Delta_{+}}\le 2$, and thus $\operatornorm{M} \le \sqrt 2\euclnorm{\Delta_{-}}$.

    We thus have
    \[\operatornorm{\sum_{x} \ketbra{x}{x}\otimes(F_{D_x}^Q - F_{D_x}^P)} \leq \sqrt{2} \euclnorm{\ket{\phi_P} - \ket{\phi_Q}}.\]
    Now we can proceed considering that both are unit vectors and we get
    \[\euclnorm{\ket{\phi_P} - \ket{\phi_Q}} = \sqrt{2 - \braket{\phi_P}{\phi_Q} - \braket{\phi_Q}{\phi_P}}.\]
    With $\braket{\phi_P}{\phi_Q} = \braket{\phi_Q}{\phi_P}$, we can bound this term separately.
    As $\braket{\phi_P}{\phi_Q} = \sum_{y \in \bin^n}\sqrt{P(y)Q(y)} = \BC(P,Q)$ we get that
    \begin{align*}
        \sqrt{2 - \braket{\phi_P}{\phi_Q} - \braket{\phi_Q}{\phi_P}} & = \sqrt{2}\sqrt{1 - \BC(P, Q)}                            \\
                                                                     & = \sqrt{2}\hellingerdistance(P, Q) \leq \sqrt{2\epsilon},
    \end{align*}
    and in total
    \[\operatornorm{\sum_{x} \ketbra{x}{x}\otimes(F_{D_x}^Q - F_{D_x}^P)} \leq 2\hellingerdistance(P, Q) \leq 2\sqrt{\epsilon}.\]
\end{proof}

This makes it possible to get to our main result.

\begin{theorem}\label{thm: query distinguishing bound}
    Let $\adv$ be a quantum algorithm given oracle access to an oracle instantiated with $P$ or $Q$ and
    \[|\prob{\adv \rightarrow 1 | P} - \prob{\adv \rightarrow 1 | Q}| \geq \zeta,\]
    where $\adv$ can make $q$ queries. Then
    $\Delta(P, Q) = \Omega(\zeta^2/q^2)$.
\end{theorem}

\begin{proof}
    View the algorithm in the compressed oracle formalism.
    This perfectly simulates the interaction with the oracle.
    The final output of $\adv$ corresponds to a measurement on the final state, which by \cref{lemma: trace distance co} has trace distance $8q\sqrt{\epsilon}$ where $\epsilon$ is the statistical distance of
    $P$ and $Q$.
    The trace distance also provides an upper bound for the difference in output probability, such that
    \[|\prob{\adv \rightarrow 1 | P} - \prob{\adv \rightarrow 1 | Q}| \leq 8q\sqrt{\epsilon}.\]
    Rewriting it, we get
    \[\zeta \leq 8q\sqrt{\epsilon} \Rightarrow \zeta^2/(64q^2) \leq \epsilon = \Delta(P, Q).\]
\end{proof}

This result improves the bound from \cite[Section 7.2]{FOCS:Zhandry12} by a factor of $q$, which itself improved the bound by a factor of $q$ compared to \cite{AC:BDFLSZ11}.
After finishing this, we found that there is another work achieving the same asymptotic bounds \cite{belovs2019quantum}, but using the quantum query complexity.
Our approach therefore only gives a different proof technique to reach the same final statement, albeit they do not provide explicit bounds needed in cryptography.
Their result, however, only provides a bound for constant advantage, and our result gives a bound for any advantage $\zeta$.

\subsection{Proof of \cref{proposition: Renyi divergence proof with sr}}\label{subsec: proof of Oracle Distribution Switching using Small-Range Distributions}
The counterexample presented in \cref{subsec: renyi fails} shows that while the ratio of success probabilities increases, the absolute success probability decreases.
This opens the door to an approach with a mix of additive and multiplicative loss terms.

We use the small-range distributions and subsequent results from \cite{FOCS:Zhandry12} towards such a mixed approach.
We note that this approach opens a path to exploiting a small Rényi divergence even between distributions that have large statistical distance.
The small-range distribution approach replaces the i.i.d.\ function with the composition of two i.i.d functions where the first one is uniform with small range $r$, and the second one uses the actual distribution. This introduces a loss based on $\adv$'s capabilities of distinguishing this oracle from the honest oracle.
The small-range distribution with range size $r$ gets $r$ classical samples from the original i.i.d.\ function, this approach is amenable to applying the classical Rényi divergence machinery.

Formally, given a distribution $D$ on $\mathcal{Y}$, define $\mathsf{SR}_{r}^{D}(\mathcal{X})$ as the following distribution on functions from $\mathcal{X}$ to $\mathcal{Y}$:
\begin{itemize}
    \item For each $i \in [r]$, sample a value $y_i \in \mathcal{Y}$ according to $D$.
    \item For each $x \in \mathcal{X}$, pick a uniformly random $i \in [r]$ and set $O(x) = y_i$.
\end{itemize}

\begin{lemma}[{\cite[Corollary 7.5]{FOCS:Zhandry12}}]\label{cor: small-range-distribution}
    The output distributions of a quantum algorithm making $q$ quantum queries to an oracle either drawn from $\mathsf{SR}_{r}^{D}$ or $D^{\mathcal{X}}$\footnote{This notation corresponds to our oracle notation $\Oraclef{D}$.} have statistical distance at most $\ell(q)/r$, where $\ell(q) = \pi^2(2q)^3/3 < 27q^3$.
\end{lemma}

\smallRangeRenyiProof*

\begin{proof}[Proof of \cref{proposition: Renyi divergence proof with sr}]
    Apply \cref{cor: small-range-distribution}, then apply post-processing of the Rényi divergence on the classical samples of which there are $r$, and then apply \cref{cor: small-range-distribution} again:
    \begin{align}
        \prob{1 \gets \adv^{{\oraclef{P}}}} & \leq \prob{1 \gets \adv^{\mathsf{SR}_r^P}} + \frac{\ell(q)}{r}                                                           \nonumber                               \\
                                            & \leq \left(\delta^r \prob{1 \gets \adv^{\mathsf{SR}_r^Q}}\right)^{\frac{\alpha-1}{\alpha}}+ \frac{\ell(q)}{r}                \nonumber                           \\
                                            & \leq \left(\delta^r \left(\prob{1 \gets \adv^{\oraclef{Q}}} + \ell(q)/r\right)\right)^{\frac{\alpha-1}{\alpha}}+ \frac{\ell(q)}{r}. \label{eq:apply-small-range}
    \end{align}
\end{proof}

\paragraph{Applicability}

This technique can be used in some regimes, but in some situations, using a generic relationship between R\'enyi divergence and the statistical distance in combination with \cref{thm: Oracle Distribution Switching using Statistical Distance} will yield stronger bounds.
Consider a Rényi divergence of factor $(1 + \epsilon)$ which is applied $r$ times.
The multiplicative factor becomes
\[(1 + \epsilon)^r = e^{r\ln(1 + \epsilon)} \approx e^{r\epsilon}.\]
The factor should not be super-polynomial in $q$, so $\epsilon$ must be chosen as $\frac{k\ln q}{r}$ where $k$ is some positive constant.
The factor $\ell(q)/r$  should be negligible, forcing a super-polynomial $r$.
Therefore, $\epsilon$ needs to be selected to be negligible.
Following standard inequalities of the Rényi divergence (see Appendix \ref{app:prob}), it bounds the statistical distance as $\Delta(P,Q) \leq \sqrt{\ln(1 + \epsilon)} \approx \sqrt{\epsilon}$, which is negligible.

Using Pinsker's inequality and \cref{thm: Oracle Distribution Switching using Statistical Distance} we can obtain a bound on the trace distance, namely \[C\cdot q\cdot  \sqrt{\sqrt{\epsilon}},\] for some small constant $C$.
In \cref{eq:apply-small-range}, clearly, $r\le -\log \prob{1 \gets \adv^{\oraclef{Q}}}/\epsilon $, otherwise the $\delta^r\prob{1 \gets \adv^{\oraclef{Q}}}\ge 1$.
This implies that the bound in \cref{cor: small-range-distribution} is at least

\[ \frac{\ell(q)}{r} \geq \frac{q^3}{-\log \prob{1 \gets \adv^{Q^{\mathcal{X}}}}/\epsilon}=\epsilon\frac{q^3}{-\log \prob{1 \gets \adv^{\oraclef{Q}}}}.\]
The statistical distance error bound $C\cdot q\cdot  \sqrt{\sqrt{\epsilon}}$ is smaller than the additive error of the small range distribution technique at least whenever
\[q \geq \sqrt{\frac{-C \log(\prob{1 \gets \adv^{\oraclef{Q}}})}{\epsilon^{3/4}}}.\]

\subsection{Proof of \cref{cor: adaptive reprogramming with other distribution}}\label{subsec: proof of corollary adaptive reprogramming with other distribution}
\AdRwDS*

\begin{proof}[Proof of \cref{cor: adaptive reprogramming with other distribution}]
    The proof is quite straightforward.
    Apply \cref{thm: adaptive reprogramming} and sample all values that need resampling.
    This gives the additive terms from the reprogramming.
    Replace the $R$ classical samples from the uniform distribution, by values from distribution $Q$ and use the properties of the Rényi divergence.
    The factor of $R_\alpha(\mathcal{U}(\mathcal{Y})\|Q)^R$ follows from observing that each sample to be reprogrammed is independent, so we can apply \cref{lemma: renyi multiplication} to get the overall bound on the Rényi divergence.
    \allowdisplaybreaks{
        \begin{align*}
            \Big|\prob{1 \gets \REPRO_{1,Q}^{\dist}} & - \prob{1 \gets \REPRO_{0}^{\dist}}\Big|                                                                                                       \\
                                                     & \leq \left|\prob{1 \gets \REPRO_{1}^{\dist}} - \prob{1 \gets \REPRO_{0}^{\dist}}\right|                                                        \\
                                                     & \quad + \left|\prob{1 \gets \REPRO_{1,Q}^{\dist}} - \prob{1 \gets \REPRO_{1}^{\dist}}\right|                                                   \\
                                                     & \leq \delta_{repr} + R \Delta(Q, \mathcal{U}(\mathcal{Y}))                                                                                     \\
            \prob{1 \gets \REPRO_{0}^{\dist}}        & \leq \prob{1 \gets \REPRO_{1}^{\dist}} + \delta_{repr}                                                                                         \\
                                                     & \leq \left(R_\alpha(\mathcal{U}(\mathcal{Y})\|Q)^R\prob{1 \gets \REPRO_{1,Q}^\distinguisher}\right)^{\frac{\alpha -1}{\alpha}} + \delta_{repr}
        \end{align*}}
    where $\delta_{repr} = \frac{3R}{2}\sqrt{q p_{\max}}$.
\end{proof}
\subsection{Proof of \cref{thm: UFCRA to UFNRA}}\label{subsec: proof of AOS ufcra to ufnra}
\AOSSUFCRAtoUFNRA*

\begin{proof}
    For $\seufcra$ security, we define a sequence of games.
    The proof follows essentially the same pattern as~\cite[Theorem 3]{AC:GHHM21}.

    The first game is the normal $\seufcra$ game.
    In the second game, we make two changes.
    The $j$th signature query $\Oracle{\sign}(i_{j}, \msg_{j}, \rho_{j})$ will be modified as follows.
    We sample a value $\ch_j \gets \mathcal{C}$ uniformly at random and everything in the ring stays as it is, until the query to $\hash(i_j, \rho_j, \com_{i_{j} - 1}, \msg)$ is made.
    This query will now be programmed to the initially sampled value of $\ch_j$.
    All values $\ch_j$ for $j \in [\nrsignqueries]$ can be sampled at the beginning.
    In the next game, we have already defined $\ch_j$ at the beginning, so we can use the HVZK simulator for every participant of the ring and get a valid signature.
    This removes the secret key, and the oracle can be simulated, i.e., we can simulate the oracle and transition to $\eufnra$.
    The games are formally depicted in \cref{fig: ufcra to ufnra and deniability games}.

    \begin{figure*}
        \begin{pcvstack}[boxed, center]
            \begin{pchstack}[space=0.5em ]

                \procedure[linenumbering]{Unforgeability games $\game{0}-\game{2}$}{
                    q \gets 0;\mathcal{L} \gets \emptyset\\
                    \pcfor i \in [N]:\\
                    \t (\instance_{i}, w_{i}) \gets \gen(\secparam)\\
                    \rho = \{\instance_{1}, \ldots, \instance_{N}\}\\
                    (\rho^{*}, \msg^{*}, \sigma^{*}) \gets \adv^{\Oracle{\sign}}(\rho)\\
                    \pcif (\rho^{*}, \msg^{*}, \sigma^{*}) \in \mathcal{L} \lor \rho^{*} \not\subseteq \rho: \\
                    \t \pcreturn \false\\
                    \{\instance^{*}_1, \ldots, \instance^{*}_{N^{*}}\} \gets \rho^{*}\\
                    ((\com_i^*,\rsp_{i}^*))_{i \in [N^*]} \gets \sigma^{*}\\
                    \pcfor i \in [N^{*}]:\\
                    \t \ch_{i+1}^* \gets \hash(i+1, \rho^{*}, \com_{i}^*, \msg^{*})\\
                    \pcreturn \bigwedge_{i \in [N^{*}]} \verifier_{2}(\instance_{i}^{*}, \com_i^*, \ch_i^*, \rsp_i^*)
                }

                \begin{pcvstack}

                    \procedure[linenumbering]{Anonymity games $\game{0}$ to $\game{2}$}{
                    q \gets 0; b \gets \bin\\
                    \pcfor i \in [N]:\\
                    \t (\instance_{i}, w_{i})\gets\gen(\secparam)\\
                    \rho = {\{\instance_{i}\}}_{i \in [N]};\SK = {\{w_{i}\}}_{i \in [N]}\\
                    b' \gets \bdv^{\chal}(\rho, \SK)\\
                    \pcreturn b=b'
                    }

                    \procedure[linenumbering]{$\chal(i_0, i_1,\rho' = \left\{\instance_i'\right\}_{i \in [N']},\msg)$}{
                        \pcif \instance_{i_0}' \notin \rho \lor \instance_{i_1}' \notin \rho:\\
                        \t \pcreturn \bot\\
                        \pcreturn \Oracle{\sign}(i_b, \rho', \msg)
                    }
                \end{pcvstack}
            \end{pchstack}

            \begin{pchstack}[space=0.5em]
                \procedure[linenumbering]{$\Oracle{\sign}(i, \rho' = \{\instance_1', \ldots, \instance_{N'}'\}, \msg)$}{
                    \pcif q = q_{\max{}} \lor i \notin [N'] \lor \instance_{i}' \notin \rho: \pccomment{$q_{\max{}}$ is $\nrsignqueries$ or $\nrchalqueries$ in the respective games}\\
                    \t \pcreturn \bot\\
                    q \gets q+1\\
                    \pclinecomment{let $w$ be the witness such that  $\mu(w) = \instance_i'$}\\
                    \ch_{i}\gets \mathcal{C} \pccomment{$\game{1}-\game{2}$}\\
                    (\com_i, \st) \gets \prover_{1}(\instance_{i}') \pccomment{$\game{0}-\game{1}$}\\
                    (\com_i, \rsp_i) \gets \simulator(\instance_{i}', \ch_i)\pccomment{$\game{2}$}\\
                    \pcfor j = i+1, \ldots, N', 1, \ldots, i-1:\\
                    \t \ch_{j} \gets \hash(j, \rho', \com_{j-1}, \msg)\\
                    \t (\com_{j}, \rsp_{j}) \gets \simulator(\instance_{j}', \ch_{j})\\
                    \ch_{i}\gets \hash(i, \rho', \com_{i-1}, \msg)\pccomment{$\game{0}$}\\
                    \rsp_i \gets \prover_{2}(\state, w, \ch_{i})\pccomment{$\game{0}-\game{1}$}\\
                    \hash := \hash^{(i, \rho', \com_{i-1},  \msg) \mapsto \ch_{i}} \pccomment{$\game{1}-\game{2}$}\\
                    \sigma \gets ((\com_1,\rsp_{1}), \ldots, (\com_{N'}, \rsp_{N'}))\\
                    \mathcal{L} \gets \mathcal{L} \cup \{(\rho', \msg, \sigma)\}\\
                    \pcreturn \sigma
                }

            \end{pchstack}
        \end{pcvstack}
        \caption{Unforgeability and anonymity games $\game{0}$ to $\game{2}$ in \cref{thm: UFCRA to UFNRA}. All indices are interpreted to be modulo $N, N'$, and modulo $N^*$, respectively.
            In $\game{0}$, we have the normal unforgeability/anonymity game. In $\game{1}$, the \ac{RO} is reprogrammed to a \ac{RO} value, and in $\game{2}$, the signature is generated entirely without the secret key $w$.
        }\label{fig: ufcra to ufnra and deniability games}
    \end{figure*}
    \vspace{.1cm}

    \noindent$\mathsf{Game}_0$: This is the original game, so here we have
    \[\advantage{\seufcra}{\adv, \AOS(\Sigma), N, \nrsignqueries} = \prob{\game{0}^{\adv}}.\]

    \noindent $\mathsf{Game}_1$:     In game $\game{1}$, each signing query is adapted as described in \cref{fig: ufcra to ufnra and deniability games}.
    The adaptive reprogramming lemma, \cref{thm: adaptive reprogramming} from \cite{AC:GHHM21}, yields a bound on the difference of $\game{0}$ and $\game{1}$.
    At the cost of at most $N$ additional queries, assume that the adversary runs the verifier on its output signature. We construct a reprogramming distinguisher: Run \adv, calling the reprogramming oracle for each signing query.
    The distinguisher issues $q_H + \kappa \cdot \nrsignqueries + N$ queries to $\hash$ with $\nrsignqueries$ many reprogramming queries, so \cref{thm: adaptive reprogramming} yields
    \[|\prob{\game{0}(\adv)} - \prob{\game{1}(\adv)}| \leq \frac{3 \nrsignqueries}{2}\sqrt{(\nrhashqueries + \kappa \cdot \nrsignqueries + N) \cdot p_{\max}} + \epsilon_{repr}'\]
    where $p_{\max} \leq 2^{-\minentropy(\secpar)}$ by assumption on the simulator commitment min-entropy, and
    $\epsilon_{repr}'$ is an additional error as the eventually produced signature forgery must be valid for the non-reprogrammed hash function $\hash$.
    For the eventual signature $(\rho^*, \msg^*, \sigma^*)$ to be valid, there can not be a pair $(\rho^*, \msg^*, \sigma')$ that was produced by an oracle query where $i$ was the signer and $\com_{i-1}' = \com_{i-1}^*$.
    For this index, there was reprogramming, so the actual signature might not be valid under the non-reprogrammed oracle.
    We bound this probability as $\epsilon_{repr}'$.
    Assume a signing query exists where $\com_{i-1}' = \com_{i-1}^*$.
    Let $\ch_j^*$ and $\ch_j'$ for $j \in [N^*]$ be challenges computed with the reprogrammed oracle $\hash$.
    If $\ch_{i-1}' \neq \ch_{i-1}^*$ and both signatures are valid, we can extract the witness via special soundness, creating an adversary $\adv_{\wrec}$.
    Now assume that $\ch_{i-1}' = \ch_{i-1}^*$, then, either the same commitments for the previous ring member were used, i.e., $\com_{i-2}' = \com_{i-2}^*$, or a hash collision was found.
    For a hash collision with ring $\rho^*$, there are at most $\nrhashqueries + N^* \cdot \nrsignqueries + N^*$ queries to $\hash$ and a collision happens with probability at most $\hashcollisionqrom$ using \cref{lemma: probability hash collision qrom}.\footnote{
        Here, we make use of domain-separation between different rings.
        The hash function we are using in the definition maps directly into $\mathcal{C}$.
        Note that $\gamma$ can be chosen such that hashing into $\mathcal{Y}$ and mapping to $\mathcal{C}$ with $\gamma$ satisfies $\gamma_{cl}/\mathcal{Y} \leq \frac{2}{|\mathcal{C}|}$.}
    If the commitments are the same for $i-2$, we can repeat the same argument.
    Assuming no hash collisions, all commitments and (challenges) are the same.
    If all commitments and challenges are the same, then there must be a response that is different for one index $j \in [N^*]$ because the signatures are different.
    This can be used to construct a $\cur$ adversary $\adv_{\cur}$ against $\Sigma$.
    The $\wrec$ and $\cur$ adversaries rely on guessing the correct index when placing the $\Sigma$-protocol instances before giving them to the adversary, which introduces the multiplicative factor of $N$.
    Altogether, this allows to bound \[\epsilon_{repr}' \leq N\cdot (\advantage{\cur}{\adv_{\cur}, \Sigma} + \advantage{\wrec}{\adv_{\wrec}, \Sigma}) + \hashcollisionqrom.\]

    \noindent $\mathsf{Game}_2$: In game $\game{2}$, we remove the signing key.
    We do not compute an initial commitment, but rather compute the commitment and response from the challenge using the simulator, and reprogram as in the previous game.

    Lets consider the statistical distance approach first.
    There are $\nrsignqueries$ queries, and for each replacement, the statistical difference is at most $\hvzktv$ such that
    \[\left|\prob{\game{1}(\adv)} - \prob{\game{2}(\adv)}\right| \leq \nrsignqueries \cdot \hvzktv.\]
    $\adv$ now makes no actual signing queries, i.e., we can define a $\eufnra$ adversary $\adv_{\eufnra}$ that plays the role of the challenger for $\adv$ by simulating all transcripts themselves and reprogramming the \ac{RO}
    \[\prob{\game{2}(\adv)} \leq \advantage{\eufnra}{\adv_{\eufnra}, \AOS(\Sigma), N}.\]

    For the Rényi divergence approach we can use \cref{lemma: renyi divergence,lemma: renyi multiplication} as the underlying distributions are independent.
    This introduces the multiplicative factor $(\hvzkrenyi)^{\nrsignqueries}$.

    We now prove anonymity.
    Classically, we would rely on reprogramming and running the simulator.
    We start with the normal anonymity game $\game{0} = \anon_{\bdv, \rsig, N, \nrchalqueries}(\secpar)$, and from it, we define two small modifications formally depicted in \cref{fig: ufcra to ufnra and deniability games}.
    The first one is just a reprogramming of the \ac{QROM}, and the second game uses the simulator to remove the use of the secret key.
    With the same reasoning as in the unforgeability part, except for the additional hash assumption that was needed for the eventual forgery, it follows that
    \[\left|\prob{\game{0}(\bdv)} - \prob{\game{1}(\bdv)}\right| \leq 3\nrchalqueries/2\sqrt{(\nrhashqueries + \kappa \cdot \nrchalqueries) \cdot 2^{-\minentropy(\secpar)}}.\]
    We distinguish two cases now.
    First, we can use the statistical distance, and observe that the distance between a signature from the oracle in the two games is the difference of using the simulator or the actual proof algorithm for one of the ring members.
    Hence, the statistical distance for each oracle query is bounded by $\hvzktv$.
    There are $\nrchalqueries$ queries which yield the final bound of $\nrchalqueries \cdot \hvzktv$.
    Second, we can use the Kullback-Leibler divergence to bound the difference in probability between $\game{1}$ and $\game{2}$ using \cref{lemma: KL bound}.
    This introduces a factor of $\sqrt{\left(\nrchalqueries\cdot  \hvzkkl\right)/2}$.
    Considering $\game{2}$, the signing oracle is independent of the choice of $b$, so the probability that $\bdv$ outputs $b$ is exactly $1/2$.
\end{proof}

\subsection{Proof of \cref{thm: UFNRA from measure and reprogram}}\label{subsec: proof of AOS ufnra from measure-and-reprogram}
\AOSUFNRAMeasureRepr*\begin{proof}
    We construct an adversary $\bdv$ from $\adv$.
    $\bdv$ receives an instance $\instance$, where $(\instance,w) \gets \Sigma.\gen(\secparam)$.

    $\bdv$ samples $i \gets [N]$, sets $\instance_i = \instance$ and generates $(\instance_j, w_j) \gets \Sigma.\gen(\secparam)$ for $j \in [N]\setminus\{i\}$.
    It forwards ${(\instance_k)}_{k \in [N]}$ to $\adv$.
    Since all instances are identically distributed, $\adv$'s view is independent of $i$.

    Let $\rho^{*}$ with $N' = |\rho^*| \leq N$ be the ring of $\adv$'s final forgery.
    We ensure that $\adv$ must have made at least a classical query for every $(1, \rho^*,\com_{N'}^*, \msg^*), \ldots,\allowbreak (N',\rho^*,\com_{N'-1}^*, \msg^*)$ from the final forgery.
    \cref{thm: measure and reprogram 2.0} yields an order $\pi$ of the queries in which we may reprogram, at a multiplicative loss of ${(2(\nrhashqueries+N)+1)}^{2N}$.
    This mimics the time-ordered list that is immediate in the classical case.
    One reprogramming carries the $\Sigma$-protocol challenge; the remaining $N-1$ only commit $\adv$ to a classical ordering.

    Recall that the query with index $j$ takes $\com_{j-1}$ as input and defines $\ch_j$ as its output.
    Hence, for $\bdv$ to place its instance at a position $p$ of $\rho^*$, it must read $\com_p$ from the input of query $p+1$ and program the challenge received from $\verifier$ as the output of query $p$, which requires $\pi(p+1) < \pi(p)$ (indices modulo $N'$).
    Such a $p$ always exists.
    As $\pi$ is independent of $i$, the guess $i$ lands on a position that lies in $\rho^*$ and satisfies this condition with probability at least $1/N$.

    For every other query $j$ that must be reprogrammed, $\bdv$ samples $\ch_j \gets \mathcal{C}$ uniformly and programs it into the oracle.
    When $\adv$ outputs a forgery $((\com_1,\rsp_1),\allowbreak \ldots, (\com_{N'},\rsp_{N'}))$, $\bdv$ sends $\rsp_p$ to $\verifier$.
    If the signature verifies, then in particular $\verifier_2(\instance_p, \com_p, \ch, \rsp_p)$ holds, which concludes the proof.
\end{proof}
\subsection{Proof of \cref{lemma: transition collision bound}}\label{subsec: proof lemma transition capacity database collision}
\TransitionCollisionBound*

\begin{proof}[\cref{lemma: transition collision bound}]
    We apply \cref{thm: quantum transition capacity} with $\PProp:=\SZ_{\leq s}\setminus \CL$ and $\PProp' := \CL$.
    Consider an arbitrary $D \in \mathcal{D}$ and input $x \in \mathcal{X}$.
    By \cite[Remark 2.5]{C:DFMS22}, the argument of the maximum in the RHS of \cref{eq:transcap} does not depend on the value $D(x)$, but only on the values of $D$ outside $x$.
    Therefore, we may assume that $D(x) = \bot$ when bounding $\prob{U \in L^{x,D}}$.
    Since $\PProp\subset\SZ_{\leq s}$,
    if $\PProp|_{D|^{x}}$ is non-empty, it follows that $D\in\SZ_{\leq s}$.
    We thus assume $D \in \SZ_{\leq s}$,
    otherwise, \cref{thm: quantum transition capacity} allows us to set $L^{x, D} = P|_{D|^{x}}=\emptyset$.
    The convention of $D(x) = \bot$ implies that $D \in \CL \Rightarrow D[x \mapsto U] \in \CL$ for all $U\in \mathcal{Y}$.
    Hence, with the same argument as before, we may assume that $D \notin \CL$, as we could set $L^{x, D} = P|_{D|^{x}}$ which is empty.

    It remains to consider $D \in \neg \CL \cap \SZ_{\leq s}$.
    Assuming $D(x) = \bot$, we have that $\bot \notin P'|_{D|^{x}}$, and so we choose $L^{x, D} = P'|_{D|^{x}}$.
    We observe that
    \[U \in L^{x, D} \Leftrightarrow D[x \mapsto U] \in P' \Rightarrow \exists x': D(x') \neq \bot \land \gamma \circ D(x') = \gamma \circ D(x).\]
    As $D \in \SZ_{\leq s}$, there are at most $s$ possible values $x'$ with $D(x') \neq \bot$.
    Therefore, we consider $\prob{\gamma(U) \in \{\gamma(D(x')) \mid x': D(x') \neq \bot\}}$.
    This holds if and only if $U \in \gamma^{-1}(\{\gamma(D(x')) \mid x': D(x') \neq \bot\})$ which is of size at most $s \cdot \gamma_{cl}$.
    Thus,
    \[\prob{U \in L^{x, D}} \leq s \cdot \gamma_{cl}/|\mathcal{Y}|\]
    holds. Applying \cref{thm: quantum transition capacity}, this yields
    \[\llbracket \SZ_{\leq s} \setminus \CL \rightarrow \CL \rrbracket \leq \sqrt{10}\sqrt{s\cdot \gamma_{cl}/|\mathcal{Y}|}.\]
\end{proof}
\subsection{Proof of \cref{lemma: transition success bound}}\label{subsec: proof lemma transition success bound}

\transitionSuccessBound*

\begin{proof}
    We apply \cref{thm: quantum transition capacity} with $\PProp:=\SZ_{\leq s}\setminus \SUC \setminus \CL$ and $\PProp' := \SUC$.
    Consider an arbitrary but fixed $D \in \mathcal{D}$ and input $x \in \mathcal{X}$.
    Furthermore, by \cite[Remark 2.5]{C:DFMS22}, the transition capacity does not depend on the value $D(x)$, but only on the values of $D$ outside $x$.
    Therefore, we may assume that $D(x) = \bot$ for any $x$ for which we consider the transition capacity.
    When considering $\PProp$, we observe that it satisfies the database property $\SZ_{\leq s}$.
    If $\PProp|_{D|^{x}}$ is non-empty, the database $D$ must satisfy $\SZ_{\leq s}$, i.e., it is bounded in size.

    We now consider different cases as depicted in \cref{tab: case distinction co proof}.
    We can assume $D \in \SZ_{\leq s}$.
    Otherwise, \cref{thm: quantum transition capacity} allows us to set $L^{x, D} = P|_{D|^{x}}$ which is empty.
    The convention of $D(x) = \bot$ implies that $D \in \CL \Rightarrow D[x \mapsto U] \in \CL$ for all $U\in \mathcal{Y}$.
    Therefore, with the same argument as before, we may assume that $D \notin \CL$, as we could set $L^{x, D} = P|_{D|^{x}}$ which is empty.

    \begin{table*}

        \caption{We have to consider two types of transitions regarding success, namely the transitions from $\SUC$ to $\neg \SUC$ and the transitions from $\neg \SUC$ to $\SUC$.
            We separate the individually domain-separated commitment and challenge queries.}\label{tab: case distinction co proof}
        \centering
        \begin{tabular}{|c|c|c|c|}
            \hline
            Case & Condition                                                         & Query Type & Bound \\
            \hline
            $1a$
                 & \multirow{2}{*}{$D \in \SUC \cap \neg \CL \cap \SZ_{\leq s}$}
                 & commitment
                 & $\leq N \cdot\ell/|\mathcal{Y}|$                                                       \\
            \cline{1-1} \cline{3-4}
            $1b$
                 &
                 & challenge
                 & $= 0$                                                                                  \\
            \hline
            $2a$
                 & \multirow{2}{*}{$D\in \neg \SUC \cap \neg \CL \cap \SZ_{\leq s}$}
                 & commitment
                 & $\leq s\cdot \ell/|\mathcal{Y}|$                                                       \\
            \cline{1-1} \cline{3-4}
            $2b$
                 &
                 & challenge
                 & $\leq \gamma_{cl}(sz_{triv}^{\mathfrak{S}})^N/|\mathcal{Y}|$                           \\
            \hline
        \end{tabular}
    \end{table*}

    \noindent\textbf{1.)}   Let $D \in \SUC \cap \neg \CL \cap \SZ_{\leq s}$.
    We assume $D(x) = \bot$, which implies that $\bot \in \PProp'_{D|^{x}}$.
    Following \cref{thm: quantum transition capacity}, we must set $L^{x, D} := \PProp_{D|^{x}}$.
    As $D \in \SUC$, there are a $(\bm{y}_{i})_{i \in [N']}$ and $\mu$ s.t.\
    \[\verifier(\instance_{i}, c_{i}, \bm{m}_{i, c_{i}})\land (\instance, \extractor^{*}(\instance_{i}, \bm{m}_{i})) \notin R\]
    where $c_{i + 1} := \gamma \circ D(i+1, \rho', \bm{y}_{i}, \mu) \land \bm{m}_{i} := D^{-1}(\bm{y}_{i})$.

    \noindent\textbf{(a)}
    Let $x$ describe a commitment query to $D$ with input $m$, i.e.\ $x=m$.
    By assumption, $D(m) = \bot$, and so, $D(m)$ is not a coordinate of any $\bm{y}_{i} \in \mathcal{Y}^{\ell}$ which we express here as $D(m)\notin \bm{y}_{i}$ for ease of notation.
    Furthermore,
    \[U \in L^{x,D} \Leftrightarrow D[x\mapsto U] \in \PProp \Rightarrow D[x \mapsto U] \notin \SUC \Rightarrow \exists i \in [N'] : U \in \bm{y}_{i},\]
    where the last property can be seen by contraposition: Assuming that $U \notin \bm{y}_{i}$, then the defining properties of $\SUC$ remain as before.
    This means that
    \[\prob{U \in L^{x, D}} \leq \prob{\exists i \in [N']: U \in \bm{y}_{i}} \leq N\cdot \ell/|\mathcal{Y}|.\]

    \noindent\textbf{(b)}
    Let $x$ describe a challenge query.
    Because $x$ describes a challenge query, it has no impact on preimages of $(\bm{y}_i)_{i \in [N]}$ and $\mu$.
    They remain as they are, and hence $D[x \mapsto U] \in \SUC$ for any $U$ which means that
    \[\prob{U \in L^{x, D}} = 0.\]

    \noindent\textbf{2.)}
    Let $D \in \neg \SUC \cap \neg \CL \cap \SZ_{\leq s}$.
    We assume $D(x) = \bot$ which implies that $\bot \in \PProp_{D|^{x}}$.
    Following \cref{thm: quantum transition capacity}, we must set $L^{x, D} := \PProp'_{D|^{x}}$.

    \noindent\textbf{(a)} Let $x$ be a commitment query to $D$ with input $m$, i.e.\ $x=m$.
    We then have
    \[U \in L^{x, D} \Leftrightarrow D[x \mapsto U] \in \PProp' \Rightarrow \exists \bm{y} : D(\beta, \rho', \bm{y}, \mu) \neq \bot \land U \in \bm{y}\]
    for some $\beta \in [N']$ and $\mu$.
    The final implication can be seen by contradiction.
    Assume that $U \notin \bm{y}$, but $D[x \mapsto U] \in \SUC$.
    This implies that for every instance satisfying the $\SUC$ property, $m$ can not be a preimage, because for a valid instance it must hold that each $c_{\beta + 1} \neq \bot$.
    Therefore, the extraction would have also worked before setting $x$ in the database, which means that $D \in \SUC$ already forms a contradiction.
    Thus, for a transition, it must hold that there exists a $\bm{y}$ such that $D(\beta, \rho', \bm{y}, \mu) \neq \bot$ with $U \in \bm{y}$ in the database.
    By $\SZ_{\leq s}$ there are at most $s$ queries with $\ell$ possible $y_i \in \mathcal{Y}$ as targets.
    Hence,
    \[\prob{U \in L^{x, D}} \leq \prob{\exists \bm{y}, \mu, \beta: D(\beta, \rho', \bm{y}, \mu) \neq \bot \land U \in \bm{y}} \leq s \cdot \ell\cdot |\mathcal{Y}|^{-1}.\]

    \noindent\textbf{(b)} Let $x$ describe a challenge query, i.e.\ $x = (\beta, \rho', \bm{y}, \mu)$.
    We then have
    \begin{align*}
        U & \in L^{x, D}  \Leftrightarrow D[x \mapsto U] \in P' = SUC                                                                                                \\
          & \Rightarrow \exists \bm{y}_{i} \in \mathcal{Y}^{k} \,\text{for}\, i \in [N']\setminus\{\beta -1\} \text{ and } \bm{y}_{\beta -1}:=\bm{y} \,\text{s.t.\ } \\
          & \forall j \in [N']: \verifier(\instance_{j}, c_{j}, \bm{m}_{j, c_{j}}) \land (\instance_{j}, \extractor(\instance_{j}, \bm{m}_{j})) \notin R
    \end{align*}
    where $\bm{c}_{j} := D(j, \rho', \bm{y}_{j-1}, \mu)$ and $\bm{m}_{j} = D^{-1}(\bm{y}_{j})$.
    The last implication can be seen by observing that because $x$ is a challenge query, if it is not part of the proof satisfying the database property, the property would hold for $D$, too.

    Considering the satisfying instance, we can observe that $\bm{y}$ uniquely defines $\bm{m}_{\beta - 1}$.
    For these, we consider the set of possible challenges $S_{\beta -1} := \{c \in \mathcal{C} | \verifier(\instance_{\beta -1}, c, \bm{m}_{\beta -1, c})\}$.
    This set must be bounded by $sz_{triv}^{\mathfrak{S}}$, because there would be an extractor otherwise that could extract a witness from $m_{\beta-1}$ which would in turn mean, that $m_{\beta}$ can not be contained in the instance satisfying $\SUC$ causing a contradiction.
    Each of these challenges $c$ has at most one preimage, because there are no collisions in our database.
    This continues around the ring until we reach $\beta $.
    Here, there is again such a set $S_{\beta}$.
    Now we consider all possible paths and thereby get an upper bound on the number of possible values $U$ can take on while causing the transition.
    \[\prob{U \in L^{x, D}} \leq \frac{\gamma_{cl}\cdot \prod_{i \in [N']} |S_{i}|}{|\mathcal{Y}|} \leq \frac{\gamma_{cl} \cdot (sz_{triv}^{\mathfrak{S}})^N}{|\mathcal{Y}|}.\]

    Applying \cref{thm: quantum transition capacity}, we now get
    \begin{align*}
        \llbracket \SZ_{\leq s}\setminus \SUC \setminus \CL \rightarrow \SUC\rrbracket & \leq \sqrt{10}\sqrt{|\mathcal{Y}|^{-1}\max\left\{N \cdot \ell,s\cdot \ell, \gamma_{cl} \cdot (sz_{triv}^{\mathfrak{S}})^N\right\}}.
    \end{align*}
\end{proof}     \section{History-Freeness for Ring Signatures}\label{sec: history free ring signature}

We now formulate a version of history-freeness~\cite[Definition 4]{AC:BDFLSZ11} for ring signatures.
Additionally, we modify and generalize it in the following ways.
First, we remove the \ac{PRF} used in the original version of \cite[Definition 4]{AC:BDFLSZ11} and use $2q$-wise independent functions \cite{C:Zhandry12} instead.
Second, we also add explicit bounds to the final security statement.

\begin{definition}[{Problem~\cite[Definition 2]{AC:BDFLSZ11}}]\label{def: Problem}
    A problem $P$ is defined by a pair $(\game{P}, \alpha_P)$ where $\game{P}$ specifies a game that a (possibly quantum) adversary plays with a classical challenger.
    The game works as follows:
    \begin{itemize}
        \item On input $\secparam$, the challenger computes a value $x$, which it sends to the adversary as its input.
        \item The adversary is then run on $x$, and is allowed to make classical queries to the challenger.
        \item The adversary then outputs a value $y$, which it sends to the challenger.
        \item The challenger then looks at $x, y$, and the classical queries made by the adversary, and outputs $1$ or $0$.
    \end{itemize}
    Here, $\alpha_P\in[0,1)$.
    It may also depend on $n$.
\end{definition}

We say that an adversary $\adv$ wins the game $\game{P}$ if the challenger outputs $1$.
We define the advantage $\advantage{\adv}{P}$ of $\adv$ in problem $P$ as
\[\advantage{\adv}{P} = \left|\prob{\adv\, \text{wins in game}\, \game{P}} - \alpha_P\right|.\]

\begin{definition}[{Hard Problem~\cite[Definition 3]{AC:BDFLSZ11}}]\label{def: hard problem}
    A problem $P$ is hard for quantum computers, if for all polynomial time quantum adversaries $\adv$, $\advantage{\adv}{P}$ is negligible.
\end{definition}

The following definition is a direct adaptation of history-freeness for signatures~\cite[Definition 4]{AC:BDFLSZ11}.
The original definition only allowed for statistical distance as a measure between the actual sign algorithm and the simulated one.
We observe that the Rényi divergence can also be used, and added this to our definition as an alternative.

\begin{definition}[History-free Reduction for Ring Signatures]\label{def: history-free reduction ring signature}
    A reduction from a hard problem $P = (\game{P}, 0)$ to the $\weufcraone$ security of a ring signature $\rsig = (\setup, \kgen, \sign^{\oracle}, \verify^{\oracle})$ in the \ac{ROM} is called history-free, if it has the following form.
    Let $\adv$ be an adversary against the $\weufcraone$ security of $\rsig$ and $\bdv$ be the algorithm $\bdv$ for problem $P$ such that:
    \begin{itemize}
        \item Algorithm $\bdv$ for $P$ contains four (explicit, classical) algorithms: $\START$, $\RAND^{\oracle_{c}}$, $\SIGN^{\oracle_{c}}$, and $\FINISH^{\oracle_{c}}$.
              The latter three algorithms have access to a shared random oracle $\oracle_{c}$.
              These algorithms, except for $\RAND^{\oracle_{c}}$, may also make queries to the challenger for problem $P$.
              The algorithms are used as follows:
              \begin{enumerate}
                  \item Given an instance $x$ for problem $P$ as input, algorithm $B$ first runs $\START(x, N)$ to obtain $(\pk_{i}, z_{i})$ for $i \in [N]$, $\pk_i$ is a public key, and $z_i$ is a private state.
                        $\bdv$ gets these where $z = (z_1, \ldots, z_N)$, and sends $\rho = \{\pk_{1}, \ldots, \pk_{N}\}$ to $\adv$ and plays the role of the challenger for $\adv$.
                  \item When $\adv$ makes a random oracle query to $\oracle(\rho', r)$, algorithm $\bdv$ responds with $\RAND^{\oracle_{c}}(\rho', r, z)$.
                        Note that $\RAND$ is given the current query as input, but is unaware of previous queries and responses.
                  \item When $\adv$ makes a signature query $\sign(\rho', \sk_{i}, m)$, algorithm $\bdv$ responds with $\SIGN^{\oracle_{c}}(\rho', m, z)$.
                  \item When $\adv$ outputs a signature forgery $(m, \sigma)$ for $\rho^*\subseteq \rho$, algorithm $\bdv$ outputs $\FINISH^{\oracle_{c}}(\rho^*, m, \sigma, z)$.
              \end{enumerate}
        \item There is an efficiently computable function $\INSTANCE(\rho)$ which produces an instance $x$ of problem $P$ such that $\START(x) = (\rho, z)$ for some $z$.
              Consider the process of first generating $\pp \gets \setup(\secparam)$, then $(\sk_i, \pk_i)$ from $\kgen$ where $\rho = \{(\sk_i, \pk_i)\}_{i \in [N]}$, and then computing $x = \INSTANCE(\rho)$.
              The distribution of $x$ generated in this way is $\hfinstancetv$ close to the distribution of $x$ generated in $\game{P}$.
        \item For fixed $z$, consider the (classical) random oracle $\oracle(\rho', r) = \RAND^{\oracle_{c}}(\rho', r, z)$.
              Define a quantum oracle $\oracle_{quant}$, which transforms a basis element $|x,y\rangle$ to $|x,y \oplus \oracle(x) \rangle$.
              The oracle, although noted slightly differently, is the same, remains the same.
              We require that $\oracle_{quant}$ is quantum computationally indistinguishable from a random oracle.
              The distinguishing advantage for $\nrhashqueries$ queries is bounded by $\hfrandtv$.\item $\SIGN^{\oracle_{c}}$ either aborts (and hence $B$ aborts) or it generates a valid signature relative to the oracle $\oracle(\rho', r) = \RAND^{\oracle_{c}}(\rho', r, z)$ with a distribution
              that for $\nrsignqueries$ sign queries is \emph{either} $\hfsigntv$ close to the correct signing algorithm or it has a small Rényi divergence $\hfsignrenyi$ compared to the correct signing algorithm.\footnote{The statement about the Rényi divergence was added to capture proofs that use a conceptually different proof technique requiring the Rényi divergence.}
              The probability that none of the signature queries abort is at least $1 - \hfabort$.
        \item If $(m, \sigma)$ is a valid signature forgery relative to $\rho^{*}$ and oracle $\oracle(\rho^{*}, r) = \RAND^{\oracle_{c}}(\rho^{*}, r, z)$ then the output of $\bdv$ (i.e.\ $\FINISH^{\oracle_{c}(\rho^{*}, m, \sigma, z)}$) causes the challenger for problem $P$ to output $1$ with probability at least $1-\hffinish$.
    \end{itemize}

\end{definition}

The following theorem is adapted from the corresponding theorem for signatures~\cite[Theorem 1]{AC:BDFLSZ11}. We give a full proof and fill in the details needed for a concrete analysis of the security loss.
\begin{theorem}\label{thm: ring signature history free}
    Let $\rsig = (\kgen, \sign, \verify)$ be a ring signature scheme.
    Suppose that there is a history-free reduction for $\rsig$ with $\hfinstancetv$, $\hfrandtv$, $\hfsigntv$, $\hfsignrenyi$, $\hfabort$, $\hffinish$.
    Then, for every $\weufcraone$ adversary $\adv$ against $\rsig$ making $\nrsignqueries$ sign queries and $\nrhashqueries$ queries to the \ac{QROM}, we can construct a quantum adversary $\bdv$ for problem $P$ such that
    \[\begin{aligned}
            \advantage{\weufcraone}{\adv, \rsig} & \leq \prob{\bdv\,\text{solves}\, P} +\hfrandtv + \hfinstancetv \\
                                                 & + \hfsigntv +  \hfabort + \hffinish.
        \end{aligned}\]
    and using the Rényi divergence, we can get
    \[\begin{aligned}
            \advantage{\weufcraone}{\adv, \rsig} & \leq  \left[\hfsignrenyi\left(\prob{\bdv\,\text{solves}\, P} + \hfinstancetv + \hfabort + \hffinish\right)\right]^{\frac{\alpha-1}{\alpha}} \\
                                                 & + \hfrandtv
        \end{aligned}\]
\end{theorem}
\begin{proof}
    The history-free reduction includes five (classical) algorithms $\START$, $\RAND$, $\SIGN$, $\FINISH$, and $\INSTANCE$.
    We consider a sequence of games where the first game is the standard $\weufcraone$ ring signature game with respect to $\rsig$ in the QROM.

    \begin{gamedescription}
        \describegame The quantum adversary $\adv_Q$ plays the standard $\weufcraone$ game for $\rsig$.
        Assume that $\adv_Q$ has a non-negligible winning probability.

        \describegame Perform the following modifications compared to $\game{1}$.
        After the challenger generates $(\pk_i, \sk_i) \gets \kgen(\pp)$ for $i \in [N]$ where $\pp \gets \setup(\secparam)$, it computes $x \gets \INSTANCE(\rho)$ as well as $(\rho, z) \gets \START(x)$.
        Further, instead of answering $\adv_Q$'s quantum random oracle queries with a truly random oracle, the challenger simulates for $\adv_Q$ a quantum accessible oracle $\oracle_{q}$ as an oracle that maps a basis element $|x, y \rangle$ into the element $|x, y \oplus \RAND^{\oracle_{q}(\rho', x, z)}\rangle$, where $\oracle_{q}$ is a random quantum-accessible oracle.
        By \cref{def: history-free reduction ring signature},  $\oracle_{quant}$ is computationally indistinguishable from random for quantum adversaries.
        The success probability of $\adv_{Q}$ is thus $ \hfrandtv$ close to the probability in $\game{1}$,

        \[\left|\prob{\adv_Q\,\text{wins in}\,\game{1}} - \prob{\adv_Q\,\text{wins in}\,\game{2}}\right| \leq \hfrandtv.\]

        \describegame Modify the challenger from $\game{2}$ as follows: instead of generating $(\sk_i, \pk_i)$ and computing $x = \INSTANCE(\rho)$, start off by running the challenger for problem $P$.
        When the challenger sends $x$, then start the challenger from the previous game using this $x$.
        The challenger is no longer in possession of $\sk_i$, therefore, when $\adv_Q$ asks for a signature on $m$, answer with $\SIGN^{\oracle_{q}}(\rho', m, z)$.
        First, since $\INSTANCE$ is part of a history-free reduction, this change in how we compute $x$ only affects the distribution of $x$ with $\hfinstancetv$, and hence the behavior of $\adv_Q$.
        Second, as long as all signing algorithms succeed, changing how we answer signing queries only negligibly affects the behavior of $\adv_Q$, i.e., either multiplicatively using the Rényi divergence $\hfsignrenyi$ and \cref{lemma: renyi divergence}, or additively using the statistical distance $\hfsigntv$.

        \[\prob{\adv_Q\,\text{wins in}\,\game{2}} \leq \hfsigntv + \prob{\adv_Q\,\text{wins in}\,\game{3}} + \hfinstancetv\]
        or alternatively using the Rényi divergence
        \[\prob{\adv_Q\,\text{wins in}\,\game{2}} \leq \left(\hfsignrenyi(\prob{\adv_Q\,\text{wins in}\,\game{3}} + \hfinstancetv) \right)^{\frac{\alpha-1}{\alpha}}\]

        For the Rényi divergence, it should be noted that the modification of the instance, i.e., \ removing the use of the secret key, is only possible after the signature queries are simulated.
        Otherwise, the queries cannot be answered.

        \describegame Define $\game{4}$ as in $\game{3}$, except for the following:
        First, generate a key $k$ for the $2(\nrhashqueries + \nrsignqueries + 1)$-wise independent function $F$\footnote{\cite{AC:BDFLSZ11} use quantum-accessible \acp{PRF}, but Zhandry~\cite[Theorem 6.1]{C:Zhandry12} showed that $2(\nrhashqueries + \nrsignqueries + 1)$-wise independent functions suffice in this step.}, and then to answer a query to $\oracle_{q}$, the challenger applies the unitary transformation that takes a basis element $|x, y\rangle$ into $|x, y \oplus F(\rho', k, x)\rangle$.
        This can perfectly simulate the random oracle~\cite[Theorem 6.1]{C:Zhandry12}, i.e.
        \[\prob{\adv_Q\,\text{wins in}\,\game{3}} =  \prob{\adv_Q\,\text{wins in}\,\game{4}} \]

        Given a quantum adversary that has a non-negligible advantage in $\game{4}$, we construct a quantum algorithm $\bdv_Q$ that solves the problem $P$.
        When $\bdv_Q$ receives instance $x$ from the challenger for problem $P$, it computes $(\rho, z) \gets \START(x, N)$ and generates a key $k$ for $F$.
        Then, it simulates $\adv_Q$ on $\rho$.
        $\bdv_Q$ answers $\adv_Q$'s quantum random oracle queries using $\RAND^{F(\rho', k, \cdot)}(\rho', \cdot, z)$.
        It answers signing queries using $\SIGN^{F(\rho', k, \cdot)}(\rho', \cdot, z)$.
        Then, when $\adv_Q$ outputs a forgery candidate $(m, \sigma)$ for $\rho^{*}$, $\bdv_Q$ computes $\FINISH^{F(\rho^*, k, \cdot)}(\rho^*, m, \sigma, z)$, and returns the result to the challenger for problem $P$.

        Observe that the behavior of $\adv_Q$ in $\game{4}$ is identical to that of $\adv_Q$ as a subroutine of $\bdv_Q$.
        If $(m, \sigma)$ is a valid forgery, then since $\FINISH$ is part of a history-free reduction, $\FINISH^{F(\rho^*, k, \cdot)}(\rho^*, m, \sigma, z)$ will cause the challenger for the problem $P$ to accept with non-negligble probability contradicting our assumption that $P$ is hard for quantum computers.
        Now, with probability $(1-\hfabort)$, no abort occurs, and with probability at least $(1-\hffinish)$, the output is a valid solution for the problem $P$.
        Combined, this yields an adversary $\bdv_Q$
        \[\prob{\adv_Q\,\text{wins in}\, \game{4}} \leq \prob{\bdv_Q \, \text{solves}\, P} + \hfabort + \hffinish.\]
    \end{gamedescription}
\end{proof}

\subsection{Proof of \cref{thm: Unf RPSF RS HF}}

\begin{proof}[{\cref{thm: Unf RPSF RS HF}}]
    By adapting the proof of \cite[Theorem 2]{AC:BDFLSZ11} we can construct a history-free reduction for \rsig.
    Using \cref{thm: ring signature history free}, a generalization of the \ac{QROM} lifting theorem for history-free reductions \cite[Theorem 1]{AC:BDFLSZ11}, we then obtain a \ac{QROM} security proof.
    We analyze the explicit security bounds this chain of arguments supports.
    We now describe the history-free reduction to the collision-resistance of the \ac{RPSF}, i.e.,
    the $\FINISH$ algorithm outputs a probable collision.

    \begin{itemize}
        \item On input $\rho$, define $\INSTANCE(\rho) := \rho$ and $\START(\rho) = (\rho, \rho)$.
              The key generation of the ring signature is exactly the key generation of the \ac{RPSF} which refers to the hard problem of collision-resistance of the \ac{RPSF} we reduce to.
        \item When $\adv$ queries the (hash) oracle $\hash(\rho', r)$, $\adv_{\collision}$ responds with
              \[\RAND^{\oracle_{c}}(\rho', r, \rho) := f_{\rho'}(\sampledom_{\rho'}(\secparam: \oracle_{c}(\rho', r))).\]
        \item When $\adv$ queries for a signature on $(\rho', m)$ signed by party $i$ where $\pk_i \in \rho' \cap \rho$, $\adv_{\collision}$ responds with
              \[\SIGN^{\oracle_{c}}(\rho', m, \rho) := (\sampledom_{\rho'}(\secparam: \oracle_{c}(\rho', m))).\]
        \item When $\adv$ outputs $(\rho^{*}, m^{*}, \sigma^{*})$, $\adv_{\collision}$ outputs a potential collision
              \[\FINISH^{\oracle_{c}}(\rho^{*}, m^{*},\sigma^{*}) := (\sampledom_{\rho^{*}}(\secparam: \oracle_{c}(\rho^{*}, m^{*})), \sigma^{*}).\]
    \end{itemize}

    The key generation of the ring keys remains the same as in the original game, which is the same as for the underlying \ac{RPSF}, which means that the instances are distributed the same, i.e., $\hfinstancetv = 0$.

    By definition of a $\rpsf$, the distribution of $f_{\rho^{'}}(\cdot)$ is within TV distance $\rpsfdomtv$ from uniform if at least one of the public keys in $\rho^{'}$ was generated honestly.
    As $\oracle(\rho', r) = \RAND^{\oracle_c}(\rho', r, \rho) = f_{\rho'}(\sampledom_{\rho'}(\secparam: \oracle_{c}(\rho', r)))$ and $\oracle_c$ is a random oracle, the outputs $\oracle(\rho', r)$ are distributed independently according to a distribution that is at most $\rpsfdomtv$ away from uniform.
    Using \cref{thm: Oracle Distribution Switching using Statistical Distance}, for any algorithm $\adv$ making $q$ quantum queries to its oracle, the TV distance between the output of $\adv$ using a random oracle or $\oracle$ is at most $8(\nrsignqueries + \nrhashqueries)\sqrt{\rpsfdomtv} =: \hfrandtv$.

    $\SIGN$ is now following the distribution of conditional preimage sampling, which by definition of the \ac{RPSF} is similar to normal signing up to a statistical distance of $\nrsignqueries\rpsfpretv =: \hfsigntv$ and a Rényi divergence of $(\rpsfprerenyi)^{\nrsignqueries} =: \hfsignrenyi$ for $\nrsignqueries$ sign queries.
    In this process of replacement, there is an additional loss:
    the $\SIGN$ algorithm is now deterministic.
    This could be used for distinguishing.
    However, we are in the $\weufcraone$ game, so the adversary can only query the signing oracle for a fixed ring and message once, so the determinism can not be detected.

    Observe that $\RAND$ and $\SIGN$ correspond to each other in that for a given message $m$, they define the same value in the range, and implicitly in the domain.
    Furthermore, the algorithms abort in the same cases, e.g., if a signature is queried for a secret key $\sk$ that is not contained in the ring, or if the challenger does not have knowledge of it.
    If $\adv$ outputs a valid forgery $(m^*, \sigma^*)$ for $\rho^* \subseteq \rho$, $\FINISH^{\oracle_{c}}(\rho^*, m^*, \sigma^*)$ produces a collision with probability at least $1 -2^{-\minentropy(\secpar)}$ due to the min-entropy of the \ac{RPSF}, i.e., $\hffinish = 2^{-\minentropy(\secpar)}$.
    $(m^*, \salt^*)$ cannot be part of a sign query, so $\oracle_c(\rho^*,\salt^*\|m^*)$ is independent of the forgery.
    The signature produced in $\FINISH$ by running $\SIGN$ is thus using the conditional preimage sampling for a fixed target for which the min-entropy statement is defined.
    Finally, we apply \cref{thm: ring signature history free} to get the claimed security bound.
\end{proof}

\section{AOS Ring Signatures from Merkle Tree Based C\&O Protocols}\label{sec: merkle-tree based co protocols}
\begin{theorem}\label{thm: CO UFNRA Merkle}
    Let $\Sigma$ be a $\mathfrak{S}^{*}$-sound Merkle tree $C\&O$ $\mathsf{\Sigma}$-protocol with challenge space $\mathcal{C}=\mathcal{C}_{\lambda}$, $\ell = \ell(\secparam)$ commitments, $\hash: \{0,1\}^{*} \to \mathcal{Y}$ be a hash function modelled as a \ac{QROM} and $\gamma:\mathcal{Y} \to \mathcal{C}_{\lambda}$
    be a function with $\gammacldef$.
    Set $\omega = \omega(\lambda) := \max_{c \in \mathcal{C}_\lambda} |c|$.
    Furthermore, let $N \in \mathbb{N}$ be fixed with $N(2\ell -1) \leq 2(\nrhashqueries -1)$.
    For any $\eufnra$ adversary $\adv$ making at most $\nrhashqueries$ quantum queries to $\hash$, there exists an efficient adversary $\bdv$ against witness-recovery of the $\Sigma$ with
    \begin{align*}
         & \advantage{\eufnra}{\adv, \AOS(\Sigma), N}  \leq N\cdot \advantage{\wrec}{\bdv, \Sigma}
        + 2\cdot N(\omega(\secpar)\log \ell + 1)\cdot |\mathcal{Y}|^{-1}                                                                                                                                                  \\
         & \hspace{1cm} + \nrhashqueries^{2}10/|\mathcal{Y}|\left(\sqrt{(\nrhashqueries-1) \cdot \gamma_{cl}} + \sqrt{\max\left\{2(\nrhashqueries-1), \gamma_{cl} \cdot (sz_{triv}^{\mathfrak{S}})^N\right\}}\right)^{2}. \\
         & \hspace{1cm}\leq N\cdot \advantage{\wrec}{\bdv, \Sigma}
        + 2\cdot N(\omega(\secpar) \log \ell + 1)\cdot|\mathcal{Y}|^{-1}                                                                                                                                                  \\
         & \hspace{1cm}+ \nrhashqueries^{2}(\nrhashqueries-1)20\gamma_{cl}/|\mathcal{Y}|\left(1 + 2(sz_{triv}^{\mathfrak{S}})^{N/2} + (sz_{triv}^{\mathfrak{S}})^N\right).
    \end{align*}
\end{theorem}

\begin{proof}[Outline]
    We omit the full proof here, instead we briefly describe how our previous proof would change.

    First, we consider how \cref{lemma: transition combined} changes.
    We must redefine the success game to use a different procedure to recompute the openings.
    As in \cite{C:DFMS22}, this is basically just a backwards traversal of the Merkle tree.
    The collision and size definition stay the same.
    Therefore, we can still utilize \cref{lemma: transition collision bound}.
    This leaves us to consider how to adapt \cref{lemma: transition success bound}.
    The cases $1b$ and $2b$ remain as they are, and it remains to consider the cases $1a$ and $2a$.
    In $1a$, the number of possible targets for a collision grows, and there are a total of $N \cdot (2\ell -1)$ targets now, as $(2\ell -1)$ limits the number of elements of the Merkle tree.
    In $2a$, the number each query now fixes at most $2$ other commitments (challenge queries fix 1, and commitment queries fix up to 2).
    Hence, there are at most $2s$ targets to consider.\footnote{In \cite{C:DFMS22}, the bound was given as $s \cdot (2\ell -1)$, but it can be reduced to this tighter bound.}
    Combining them bounds the transition capacity as
    \[\leq \nrhashqueries\sqrt{\frac{10}{|\mathcal{Y}|}}\left(\sqrt{(\nrhashqueries-1) \cdot \gamma_{cl}} + \sqrt{\max\left\{N(2\ell -1), 2(\nrhashqueries-1), \gamma_{cl} \cdot (sz_{triv}^{\mathfrak{S}})^N\right\}}\right).\]

    Then, we would use this transition bound in \cref{thm: CO UFNRA}.
    Due to the switch, we have to make the same modifications that were done in \cite{C:DFMS22}.
    More precisely, this means that we can keep almost everything as it is, except that we must check for consistency of the eventual signatures.
    This signature now builds a Merkle tree in its verification, and the verification of a response from the $\mathsf{\Sigma}$-protocol requires $\log \ell$ many consistency checks (the height of the Merkle tree).
    Combining this with the ring size of $N$ and the consistency between the classical queries in the verification, and the database is given except with probability
    \[2\cdot N(\omega(\secpar)\cdot \log \ell + 1)\cdot |\mathcal{Y}|^{-1}.\]
    The remainder of the proof is analogous to the proof of \cref{thm: CO UFNRA}.
\end{proof}

\section{Applications and Instantiations}
This section serves to demonstrate how our techniques can be used in practice.
We provide some background in the additional preliminaries in Appendix~\ref{subsec: NTRU Lattices}
\subsection{\textsc{Falcon} in the QROM}\label{sec: Falcon in the QROM}
Recently, the security of \textsc{Falcon} has been revisited formally in \cite{EPRINT:GajJanKil24b} at Eurocrypt 2026, but their proof does not look at the \ac{QROM}.
\cite{qrom_falcon} uses the Rényi divergence for a \ac{QROM} proof of \textsc{Falcon}, but only to derive a bound on the statistical distance using Pinsker's inequality, which, as shown in \cite{EPRINT:GajJanKil24b}, is already insufficient for the classical proof.
To the best of our knowledge, there is no \ac{QROM} proof for \textsc{Falcon} that avoids costly statistical distance arguments for the preimage sampler.

\textsc{Falcon} as a signature scheme can be classified as a \ac{HSwA} signature scheme.
Aborts have been considered in the \ac{QROM} both for signatures and sigma protocols, e.g., in \cite{EPRINT:FalFehHua25,C:DFPS23,C:BBDDFG23}.
The work of Fallahpour et al. \cite{EPRINT:FalFehHua25} could be used for
\textsc{Falcon}. We choose to write a direct proof, to avoid a separate analysis of the necessary accepting-transcript \ac{HVZK} property using Rényi divergence arguments.
The work of \cite{C:DFPS23} is written for \ac{FSwA}, but when considering generalised \ac{FSwA} by fixing the first message of the sigma protocol as the randomness used in the signature, this can capture \ac{HSwA} signatures.
However, we were not able to directly map their simulator to \textsc{Falcon}.

The previous general solutions also share another problem that hinders their application: they assume that the challenge/hash value is distributed uniformly at random, which is not the case for \textsc{Falcon}.
In \textsc{Falcon}, the values that are reprogrammed into the random oracle come from a distribution that is close to uniform in terms of Rényi divergence.

As no existing proof technique can be used directly, we translate the classical proof from \cite{EPRINT:GajJanKil24b} into the \ac{QROM}.

To prove the security of \textsc{Falcon} in the \ac{QROM}, we use the generalized abstraction of \textsc{Falcon} to \textsc{CoreFalcon} from \cite{EPRINT:GajJanKil24b} described in \cref{fig: CoreFalcon}.
\begin{figure}
    \begin{pchstack}[boxed, center, space=0.5em]
        \begin{pcvstack}[space=0.5em]
            \procedure[linenumbering]{$\gen(\secparam)$}{
                (\bm{h}, \bm{f},\bm{g}) \gets \ntrutrapgen(q, \alpha)\\
                \bm{B} := \bm{B}_{\bm{f}, \bm{g}}\\
                \pcreturn (\sk := \bm{B}, \pk := \bm{h})
            }
            \procedure[linenumbering]{$\sign(\sk = \bm{B}, \mu)$}{
                r \gets \bin^{\saltbits}\\
                \bm{c} := \hash(\pk, r, \mu) \in \mathcal{R}_q\\
                \pcwhile \|(\bm{s}_1, \bm{s}_2)\|_2 > \beta:\\
                \t (\bm{s}_1, \bm{s}_2) \gets \ntrusamplepre(\bm{B}, s, \bm{c})\\
                \sigma := (r, \bm{s}_2) \in \bin^\saltbits \times \mathcal{R}\\
                \pcreturn \sigma
            }
        \end{pcvstack}

        \begin{pcvstack}[space=0.5em]

            \procedure[linenumbering]{$\verify(\pk = \bm{h}, \mu, \sigma = (r, \bm{s}_2))$}{
                c := \hash(\pk, r, \mu)\\
                \bm{s}_1 := \bm{c} - \bm{s}_2 \bm{h} \mod q\\
                \pcreturn \llbracket \|(\bm{s}_1, \bm{s}_2)\|_2 \leq \beta\rrbracket
            }
            \procedure[linenumbering]{$\sign^{+}(\sk = \bm{B}, \mu)$}{
                \pcwhile \|(\bm{s}_1, \bm{s}_2)\|_2 > \beta:\\
                \t r \gets \bin^{\saltbits}\\
                \t \bm{c} := \hash(\pk, r, \mu) \in \mathcal{R}_q\\
                \t (\bm{s}_1, \bm{s}_2) \gets \ntrusamplepre(\bm{B}, s, \bm{c})\\
                \sigma := (r, \bm{s}_2) \in \bin^\saltbits \times \mathcal{R}\\
                \pcreturn \sigma
            }
        \end{pcvstack}
    \end{pchstack}
    \caption{Variations of \textsc{Falcon}, namely \textsc{CoreFalcon} $= (\gen, \sign, \verify)$ and \textsc{CoreFalcon}$^{+}= (\gen, \sign^{+}, \verify)$ from \cite{EPRINT:GajJanKil24b}.
        We used the notation of \cite{C:GajJanKil24} for NTRU lattices, which differs slightly from \cite{EPRINT:GajJanKil24b}, but they refer to the same underlying construction}\label{fig: CoreFalcon}
\end{figure}

The proof we show here is primarily to demonstrate that \textsc{Falcon} can be proven secure in the \ac{QROM} using Rényi divergence arguments that mirror the classical proof relatively closely.
Compared to the detailed analysis done in \cite{EPRINT:GajJanKil24b}, we have primarily converted the arguments to the \ac{QROM}. One place where our proof departs from the ROM reasoning is when reducing to $\ntrurisis$. In \cite{EPRINT:GajJanKil24b}, a reduction to $\tntrurisis{\nrhashqueries+1}$ is performed by programming a new instance for each query to the random oracle. This clearly does not work in the QROM. We provide two alternative ways out.
i) We use measure-and-reprogram to include a single $\ntrurisis$ target, instead of their variant with $(\nrhashqueries + 1)$ preimages.
ii) We embed a multi-challenge $\tntrurisis{t}$ challenge using small-range distributions from \cref{cor: small-range-distribution}, yielding a reduction to $\tntrurisis{t}$ for $t=\Omega\left(\frac{(\nrhashqueries + C_s)^3}{\epsilon}\right)$ with a success probability loss of $\epsilon$ and a longer reduction runtime.

\SUFCMAFalcon*

\begin{proof}
    The proof follows the structure of \cite[Theorem 1 and 2]{EPRINT:GajJanKil24b}, and we only highlight those parts that are changed when transitioning to the \ac{QROM}.
    \begin{gamedescription}
        \describegame This is the unforgeability game for \textsc{CoreFalcon}$^{+}$ so by definition
        \[\prob{1 \gets \game{1}} = \advantage{\seufcma}{\adv, \textsc{CoreFalcon}^{+}, \nrsignqueries}.\]
        We assume from here on, for the cost of one additional hash query, that the final forgery was queried for its random oracle value.
        \describegame This game is identical to the previous one, except that it aborts if the overall number of sampled preimages in the signing oracle, i.e., including potential repetitions, exceeds the threshold $C_s$.
        Following \cite[Claim 1.]{EPRINT:GajJanKil24b}, we have
        \begin{align*}
             & \left|\prob{1 \gets \game{1}} - \prob{1 \gets \game{2}}\right| \leq       \\
             & \qquad \sum_{i=0}^{\nrsignqueries}{C_s \choose i}(1 - p)^{C_s -i}(p)^{i}.
        \end{align*}
        The signing oracle still runs $\ntrusamplepre$ here, so $p$ is the acceptance probability of the game we are leaving.
        The threshold is never reverted; all later games inherit the abort.
        \describegame This game is identical to $\game{2}$ except that we change how signature queries are answered.
        Instead of querying the random oracle, we program it with $\bm{c}$ from the freshly generated $\tntrurspisis{C_s}$ samples, and we do so for every preimage sampling attempt, including those that the rejection loop later discards.
        Note that we think of this reprogramming as an operation applied to the random oracle, so when we later replace the random oracle by a different function, the reprogramming operation is applied to that oracle instead.
        The analysis of this step deviates from \cite{EPRINT:GajJanKil24b}, as we need to account for the \ac{QROM}.
        The total number of values whose distribution is changed is $C_s$, and there are a total of $C_s + \nrhashqueries + 1$ queries to the random oracle.
        As the challenge carries no rejection loop, each $\bm{c}_i$ is distributed exactly as $\mathcal{Q}_{\bm{h}}$.
        This new distribution replaces uniform values, and there are $C_s$ independent replacements, which allows us to directly apply the maximal Rényi divergence over the choice of $(\bm{h}, \bm{f},\bm{g}) \gets \ntrutrapgen(q, \alpha)$.
        Applying \cref{cor: adaptive reprogramming with other distribution} with the Rényi divergence yields
        \[\prob{1 \gets \game{2}} \leq (r_u^{C_s}\prob{1 \gets \game{3}})^{\frac{\alpha_u -1}{\alpha_u}} + \frac{3C_s}{2}\sqrt{\frac{(C_s + \nrhashqueries + 1)}{2^{\saltbits}}}.\]

        \describegame Conditioned on $\bm{c}_i$, the pair $(\bm{u}_i, \bm{v}_i)$ from the
        $\tntrurspisis{C_s}$ challenge is distributed exactly as
        $\mathcal{D}_{\bm{\Lambda}_{\bm{c}_i}, s}$, since
        $(\bm{0}, \bm{c}_i)$ is a preimage of $\bm{c}_i$ and
        $\mathcal{D}_{\mathcal{R},s} \times \mathcal{D}_{\mathcal{R},s}$ restricted to that
        coset is the coset discrete Gaussian.
        This is precisely the distribution against which $r_p$ is measured in
        \cref{tab: CoreFalcon params}.
        This game is identical to $\game{3}$ except that the output of the preimage sampler is replaced by the pair $(\bm{s}_1, \bm{s}_2) := (\bm{v}_i, \bm{u}_i)$ generated in the previous game for computing $\bm{c}_i$, and the norm check is applied to that pair.
        This change in distribution behaves as in \cite{EPRINT:GajJanKil24b}, giving us
        \[\prob{1 \gets \game{3}} \leq (r_p^{C_s}\prob{1 \gets \game{4}})^{\frac{\alpha_p -1}{\alpha_p}}.\]
        Attempt $i$ is now discarded exactly when $\euclnorm{(\bm{u}_i, \bm{v}_i)} > \beta$, so the loop of $\sign^{+}$ is simulated faithfully, and $\game{4}$ no longer uses the trapdoor.
        Let $(\mu^*, r^*, \bm{u}^*, \bm{v}^*)$ be the final forgery of an adversary, where $(\bm{u}^*, \bm{v}^*) := (\bm{s}_2^*, \bm{c}^* - \bm{s}_2^*\bm{h})$ for $\bm{c}^* := \hash(\bm{h}, r^*, \mu^*)$.
        If $(\mu^*, r^*)$ was part of a programmed attempt $i^*$, then we have a solution for $\tntrurspisis{C_s}$, as the forgery is a different preimage of $\bm{c}_{i^*}$: if the attempt was accepted this follows from $\seufcma$ freshness, and if it was discarded it follows from $\euclnorm{(\bm{u}_{i^*}, \bm{v}_{i^*})} > \beta \geq \euclnorm{(\bm{u}^*, \bm{v}^*)}$.\footnote{Note that the discarded case is only free because the challenge preimages are not conditioned on being short; otherwise it would require a conditional min-entropy argument.}
        Else, if $(\mu^*, r^*)$ has the same $\bm{c}$ as one of the programmed attempts, then $\adv$ has hit one of the targets $T := \{\bm{c}_i\}_{i \in [C_s]}$ with an un-reprogrammed hash value.
        The targets are drawn as $\tntrurspisis{C_s}$ challenges and are thus fixed independently of the oracle, and $|T| \leq C_s$, so their distribution is immaterial here.
        Substituting $\prob{U \in T} \leq C_s/|\mathcal{R}_q|$ for $\prob{U \in L^{x,D}}$ in the proof of \cref{lemma: transition collision bound} gives a transition capacity of $\sqrt{10 C_s/|\mathcal{R}_q|}$ that is independent of the database size and holds at every step; summing over the $C_s + \nrhashqueries + 1$ queries and squaring as in \cref{lemma: probability hash collision qrom} bounds this case by $\frac{10(C_s + \nrhashqueries + 1)^2C_s}{|\mathcal{R}_q|}$.
        Finally, there are two options to bound the last case, yielding the two different statements in our theorem.

        The first approach uses measure-and-reprogram (\cref{thm: measure and reprogram 2.0}) on the initial oracle before any other reprogramming.
        So, there are two levels of reprogramming, measure-and-reprogramming is the \emph{inner} reprogramming and resampling is the \emph{outer} reprogramming.
        In the case we are considering here, verifying the output forgery does not use a hash input with outer reprogramming, so the measure-and-reprogram experiment fails unless inner reprogramming happens for an input that is not subject to outer reprogramming.
        The freshly sampled output of measure-and-reprogram is now replaced by the challenge from a $\ntrurisis$ adversary.
        This captures all possible cases, which means
        \begin{align*}
            \prob{1 \gets \game{4}} & \leq \advantage{\tntrurspisis{C_s}}{\bdv, q, \alpha, \beta}                         \\
                                    & + (2\nrhashqueries + 2C_s + 3)^{2}\advantage{\ntrurisis}{\cdv, 1, q, \alpha, \beta} \\
                                    & + \frac{10(C_s + \nrhashqueries + 1)^2C_s}{|\mathcal{R}_q|}.
        \end{align*}

        The second approach uses small-range distributions (\cref{cor: small-range-distribution}).
        Given $t$ samples from a $\tntrurisis{t}$ adversary, we can replace the random oracle by a small-range function with range size $t$.
        Note that we replace the un-reprogrammed random oracle by a small-range distribution, effectively replacing the reprogrammed oracle with a small-range distribution that has been reprogrammed in the same way.
        This transition is possible as the un-reprogrammed oracle is indistinguishable from the small-range version, and the other game hops that happened after can be simulated by a distinguisher.
        Now, every valid forgery is a valid answer in the $\tntrurisis{t}$ game, as a valid forgery will always relate to an un-reprogrammed hash output.
        This transition does incur an additive loss in the order of $\frac{27(\nrhashqueries + C_s+ 1)^3}{t}$ such that
        \begin{align*}
            \prob{1 \gets \game{4}} & \leq \advantage{\tntrurspisis{C_s}}{\bdv, q, \alpha, \beta}                                       \\
                                    & + \advantage{\tntrurisis{t}}{\ddv, 1, q, \alpha, \beta} + \frac{27(\nrhashqueries + C_s+ 1)^3}{t} \\
                                    & + \frac{10(C_s + \nrhashqueries + 1)^2C_s}{|\mathcal{R}_q|}.
        \end{align*}
        As the reduction needs to write the $t$ $\ntrurisis$ samples into its QRAM to instantiate the small-range distribution, it has an additional additive runtime overhead of $\Theta(t)$ compared to the the measure-and-reprogram reduction.
    \end{gamedescription}
\end{proof}

\begin{remark}\label{remark: spisis}
    The $\tntrurspisis{C_s}$ term is only needed for $\seufcma$.
    Every programmed position refers to a message that was submitted to the signing oracle, so under $\eufcma$ the freshness condition on $\mu^*$ excludes that case and the corresponding summand can be dropped from $\epsilon'$.
\end{remark}
\subsection{Application to \textsc{Gandalf}}\label{subsec: Application to Gandalf}
This section performs a \ac{QROM} security analysis of \textsc{Gandalf} using our techniques.
The original ring signature scheme~\cite[Figure 5]{C:GajJanKil24} can be translated into a \ac{RPSF} (as detailed in \cref{fig:Gandalf RPSF}).

When instantiating \textsc{Gandalf} as a \ac{RPSF}-based ring signature (see \cref{fig: rsig from rpsf}), the resulting signature is syntactically equivalent to \textsc{Gandalf}, except for two small modifications.
First, we need to add $\saltbits$ salt bits such that we can use the Rényi divergence properties (see \cref{subsec: Unforgeability via TAR}).
This does not have any functional impact on the \ac{RPSF} and is also what is done for \textsc{Gandalf}, when turning it into a $\seufcra$ ring signature \cite[Appendix B.2, Theorem 7]{C:GajJanKil24}.
Secondly, the \ac{RPSF} presampling algorithm returns $N+1$ elements for a ring $\rho$ with $|\rho| = N$ instead of $N$ as in \textsc{Gandalf}.
This one element can, however, be saved by changing the evaluation in \cref{fig: rsig from rpsf} such that the \ac{RPSF} recomputes the missing element with the hash target and then all are evaluated.
Thus, this second modification is just an additional optimization in \textsc{Gandalf} and has no loss in security.

\begin{figure}
    \begin{pchstack}[boxed, center, space=0.5em]
        \begin{pcvstack}[space=0.5em]
            \procedure[linenumbering]{$\setup(\secparam)$}{
                (s, \kappa, M) \pccomment{derive from $\secparam$}\\
                \tau := \phi(\kappa)\\
                \beta := \tau \cdot s\cdot \sqrt{(\kappa + 1)M}\\
                \pp := (\kappa, \tau, \beta)\\
                \pcreturn \pp
            }
            \procedure[linenumbering]{$\trapgen(\pp)$}{
                (\bm{h}, \bm{f}, \bm{g}) \gets \ntrutrapgen(q, \alpha)\\
                t := (\bm{f}, \bm{g}) \in \mathcal{R}_q \times \mathcal{R}_q\\
                a := \bm{h} \in \mathcal{R}_q\\
                \pcreturn (a, t)
            }

            \procedure[linenumbering]{$f_{\rho}(d)$}{
                \pcparse \rho \rightarrow (\bm{h}_1, \bm{h}_2, \ldots, \bm{h}_N)\\
                \pcparse d \rightarrow (\bm{u}_1, \bm{u}_2, \ldots, \bm{u}_N, \bm{v})\\
                r := \bm{v} + \sum_{i \in [N]} \bm{h}_i\cdot\bm{u}_i \in \range_{\rho}\\
                \pcreturn r
            }
        \end{pcvstack}

        \begin{pcvstack}[space=0.5em]
            \procedure[linenumbering]{$\samplepre(\rho, t, r)$}{
                \pcparse \rho \rightarrow (\bm{h}_1, \bm{h}_2, \ldots, \bm{h}_N)\\
                \pcparse t \rightarrow (\bm{f}, \bm{g})\\
                \pcparse r \rightarrow \bm{r} \in \range_\rho = \mathcal{R}_q\\
                \pcrequire N \leq \kappa\\
                \pcrequire \exists j \in [N]: \mu(t) = a_j\\
                \pcfor i \in [N] \setminus \{j\}:\\
                \t \bm{u}_i \gets \mathcal{D}_{\mathcal{R}, s}\\
                \t \bm{c}_i := \bm{u}_i \cdot \bm{h}_i \in \mathcal{R}_q\\
                \bm{c}_{j} := \bm{r} - \sum_{i \in {N}\setminus \{j\}}\bm{c}_i\\
                (\bm{u}_j, \bm{v}) \gets \ntrusamplepre(\bm{B}_{\bm{f}, \bm{g}}, s, \bm{c}_j)\\
                d := (\bm{u}_1, \bm{u}_2, \ldots, \bm{u}_N, \bm{v}) \in \mathcal{R}_q^{N+1}\\
                \pcreturn d
            }
        \end{pcvstack}
    \end{pchstack}

    \caption{
        A \ac{RPSF} build from \textsc{Gandalf}.
        The domain $\domain_\rho$ for $|\rho| = N$ is $\{(\bm{d}_1, \bm{d}_2, \ldots \bm{d}_N, \bm{d}) \in \mathcal{R}_q^{N+1}: \euclnorm{(\bm{d}_1, \bm{d}_2, \ldots, \bm{d}_N, \bm{d})} \leq \beta\}$.
        The range $\range_\rho = \mathcal{R}_q$ for every $\rho$.
        The algorithms $\ntrutrapgen$ and $\ntrusamplepre$ are from \cref{lem: NTRU trapgen,cor: presampler KL and Rényi divergence}.
        The function $\mu$ identifies the public value $a$ that belongs to the trapdoor $t$.
        $\sampledom(\rho)$ for $|\rho| = N$ is $\mathcal{D}_{\mathcal{R}, s}^{N + 1}$.}\label{fig:Gandalf RPSF}
\end{figure}

\subsubsection{Properties of the \ac{RPSF}}

The effective \ac{RPSF} built from \textsc{Gandalf} has slightly different correctness properties than those outlined in the definition of \acp{RPSF} (\cref{def: RPSF}).
Both $\sampledom$ and $\samplepre$ for a fixed $\rho$ are defined over $\mathcal{R}^{|\rho|}$ instead of $\domain_\rho$ which is a ball in $\mathcal{R}^{|\rho|}$.
$\samplepre$ will always return valid preimages under evaluation of the map, but the images might not be in the domain.
The correctness property of \acp{RPSF}, for simplicity, does not consider that $\sampledom$ and $\samplepre$ might be outside the domain.
To account for this, the actual function $f_{\rho}$ needs to check if the provided value is in the domain when used in \cref{fig: rsig from rpsf}.
In the $\rsig$ proof (\cref{thm: Unf RPSF RS TAR}), treat $f_{\rho}$ simply as the map from $\mathcal{R}^{|\rho|}$ to $\mathcal{R}_{q}$.
However, when the adversary provides a forgery for a previously signed target, this is only a proper signature, if the simulated signature is also in the domain.
Therefore, we get an additional error term based on whether the sample from $\sampledom$ is in the domain.
This is bounded using \cref{lemma: length bound discrete Gaussian} and describes the additional error term (multiplied by $\nrsignqueries$).

\paragraph{Correctness}

The correctness can be imported directly from \cite{C:GajJanKil24}, as the two ring signatures are semantically equivalent.

\begin{corollary}
    $\rpsf$ (\cref{fig:Gandalf RPSF}) is $\delta(\kappa)$-correct where
    \[\delta(\kappa) = \tau^{(\kappa + 1)M}\cdot e^{\frac{(\kappa + 1)M}{2}(1 - \tau^2)}\]
    with $\tau > 1$.
\end{corollary}

\paragraph{Domain Sampling with Uniform Output.}
We can use the Rényi uniformity arguments for NTRU from \cref{cor: Renyi uniformity of NTRU}.
\begin{corollary}
    Let $q$ be prime, $\alpha \in (1, \infty), \epsilon \in (0,1/2), s \geq \eta_{\epsilon(\bm{\Lambda}_{\bm{h}, q})}$.
    The $\sampledom$ algorithm of \cref{fig:Gandalf RPSF} satisfies domain sampling with uniform outputs (see notation in Property \ref{prop: rpsf domain uniformity}) with
    \[\rpsfdomrenyi \lesssim 1 + \frac{2\alpha \epsilon^2}{(1-\epsilon)^2}.\]
\end{corollary}
\begin{proof}
    Considering the definition of Property \ref{prop: rpsf domain uniformity}, there is at least one honestly generated (non-zero) public value $\bm{h}_j$ in the ring $\rho = (a_1, \ldots, a_N) = (\bm{h}_1, \ldots, \bm{h}_N)$ for $j \in[N]$.
    Let $d = (\bm{u}_1, \ldots, \bm{u}_{N+1})$ be a sample from $\sampledom$.
    Fixing all $i \in[N]\setminus{j}$, and the respective $\bm{u}_i$ from $\sampledom$, they describe a simple shift in $\mathcal{R}_q$:
    \[f_{\rho}(d) = \underbrace{\bm{u}_j\bm{h}_j + \bm{u}_{N+1}}_{:= \bm{x}}+ \sum_{i \in [N]\setminus \{j\}}\bm{u}_i\bm{h}_i.\]
    Applying \cref{cor: Renyi uniformity of NTRU} to $\bm{x}$ gives the Rényi divergence bound, which is preserved by the fixed shift applied after.
\end{proof}

\paragraph{Preimage Sampling is not detectable.}

To show that the preimage sampling is not detectable, we can use \cref{cor: presampler KL and Rényi divergence}.
\begin{corollary}
    Let $M$ be a positive integer, $\alpha > 1$ and $\epsilon \in (0, 1/4)$.
    Then, with parameter choice according to \cref{cor: presampler KL and Rényi divergence} ($s \geq \eta_{\epsilon(\ZZ^{2M})}\cdot \|\bm{B}\|_{GS}$), the $\samplepre$ algorithm of \cref{fig:Gandalf RPSF} has Kullback-Leibler divergence and Rényi divergence bounded as
    \[\rpsfprekl \lesssim 2\epsilon^2\quad \text{and}\quad\rpsfprerenyi \lesssim 1 + 2\alpha\epsilon^2.\]
\end{corollary}

\begin{proof}
    Considering the definition of Property \ref{prop: rpsf preimage sampling undetectable}, there is at least one honestly generated (non-zero) public value $\bm{h}_j$ in the ring $\rho = (a_1, \ldots, a_N) = (\bm{h}_1, \ldots, \bm{h}_N)$ for $j \in[N]$.
    Observe that \cref{cor: presampler KL and Rényi divergence} can be applied for an arbitrary target.
    For a preimage sample $d = (\bm{u}_1, \ldots \bm{u}_{N+1}) \gets \samplepre(\rho, t, \bm{r})$ where $j \in [N]$ is the honest index for which the trapdoor is given.
    Then, treat all other indices $[N+1]\setminus\{j, N+1\}$ as a target shift that is sampled before using the same distribution.
    Now, they have both a fixed arbitrary target with the same distribution so far.
    Applying \cref{cor: presampler KL and Rényi divergence} yields the respective bounds.
\end{proof}

\paragraph{One-Wayness.}
\begin{lemma}
    For any adversary $\adv = (\adv_1, \adv_2)$ against the one-wayness (Property~\ref{prop: rpsf onewayness}) of \cref{fig:Gandalf RPSF}, there exists an adversary $\bdv$ such that
    \[\advantage{\oneway}{\adv, \rpsf} \leq \advantage{\ntrurisis}{\bdv, \kappa, q,\alpha, \beta}.\]
\end{lemma}
\begin{proof}
    Assume an adversary $\adv$ against the one-wayness of \cref{fig:Gandalf RPSF}.
    We can make a straightforward reduction to the $\ntrurisis$ problem.

    Consider $\ntrurisis$ game with $N$ values $\bm{h}_i$, $\beta$ as defined in the setup algorithm of the \ac{RPSF} and a random target $\bm{r} \gets \mathcal{R}_q$.
    All $\bm{h}_i$ will now be given to $\adv$ as the public values and $\bm{r}$.
    $\adv_1$ outputs $(\state, \rho')$ where $\rho' \subseteq \rho$.
    We define $L$ as the set of indices such that $\{\rho_{i}\}_{i \in L} = \rho'$.
    Now, we give the target $\bm{r}$ to $\adv_2$ with $\state$ and we will receive a preimage $d$ of $r$.
    The preimage $d = (\bm{u}_i)_{i \in L \cup\{N+1\}}$ must be in the domain to be properly evaluated by $f_{\rho'}$,
    so $\euclnorm{(\bm{u}_i)_{i \in L \cup\{N+1\}}} \leq \beta$ and $f_{\rho'}(d) = \bm{r}$.
    This preimage can be extended to a $\ntrurisis$ solution by setting the remaining positions to $\bm{0}$, i.e., output $(\bm{u}_i)_{i \in [N+1]}$
    where $\bm{u}_j = \bm{0}$ for $j \notin L \cup\{N+1\}$.
    The evaluation remains the same, and the norm bounds have not changed, making it a valid $\ntrurisis$ solution.
\end{proof}

\paragraph{Conditional Preimage Min-Entropy.}
\begin{lemma}
    If $s \geq 2\eta_{\epsilon}(\bm{\Lambda}_{\bm{h},q})$ and $\epsilon \in (0,1/3)$, then the conditional min-entropy described in Property \ref{prop: rpsf preimage min entropy} of \sampledom~as defined in \cref{fig:Gandalf RPSF} is at least $2M -1$.
\end{lemma}
\begin{proof}
    The domain is defined by the ring size $N$ of $\rho$, where all are honestly generated.
    We consider a subring $\rho' \subseteq \rho$.
    To be meaningful $\rho'$ is of size $N \geq 1$.
    $\sampledom$ samples a discrete Gaussian over $\mathcal{R}^{N+1}$.
    Consider the first $N-1$ to be fixed.
    This allows us to use the conditional min-entropy of the normal NTRU \acp{PSF}, so we can use the existing literature.
    So, for any value of the first $N-1$ elements, we can already get high min-entropy.

    We use the arguments from \cite[Proof of Property 4 in Theorem 4.2]{EC:SteSte11}.
    For any fixed target $\bm{t} \in \mathcal{R}_q$, the conditional distribution of $\bm{z} = (\bm{z}_1, \bm{z}_2) \gets \mathcal{D}_{\mathcal{R}, s}^2$, given $\bm{h}\bm{z}_1 + \bm{z}_2 = \bm{t}$ is exactly $\bm{c} + \mathcal{D}_{\bm{\Lambda}_{\bm{h}, q}, s, - \bm{c}}$ where $\bm{c} = (\bm{1}, \bm{t} - \bm{h})$\footnote{$\bm{c}$ is a preimage of $\bm{t}$, i.e., $\bm{h}\bm{1} - \bm{h} + \bm{t} = \bm{t}$.}.
    Now, apply \cref{lemma: min-entropy shifted discrete Gaussians} to get the conditional min-entropy of at least $2M-1$.
\end{proof}

\paragraph{Collision-Resistance.}
\begin{lemma}
    For any adversary $\adv$ against the collision-resistance (Property~\ref{prop: rpsf collision resistance}) of \cref{fig:Gandalf RPSF}, there exists an adversary $\bdv$ such that
    \[\advantage{\collision}{\adv, \rpsf} \leq \advantage{\ntrursis}{\bdv, \kappa,q,\alpha, 2\beta}.\]
\end{lemma}
\begin{proof}
    The collision resistance, similarly to the one-wayness, can be converted directly from a subring $\rho'$ to the highest ring $\rho$, by adding $\bm{0}$ values.
    Two different short values $d_1$ and $d_2$, mapping to the same target, means that $d_1 - d_2$ (coordinate-wise) maps to $\bm{0}$ due to the linearity.
    This is directly a solution to $\mathcal{R}$-SIS and has length at most $2\beta$.
\end{proof}

\subsubsection{QROM bounds on \textsc{Gandalf}}\label{subsubsec: QROM bounds on Gandalf}

The general proof for \ac{RPSF} considers perfect correctness, but for \textsc{Gandalf} we need to consider that signatures and domain samples might be too long.
The transitioning from presampling to the simulation using $\sampledom$ does not use correctness properties and is independent of those.
The security proof could only fail in the last step when reducing to $\ntrursis$, because the simulated values might be too large.
Other positions do not fail, because they do not explicitly use the correctness properties, and the transitions themselves are independent of correctness.
With a probability of at least $1- \nrsignqueries \cdot \delta(\kappa)$, all are correct, and the reduction works.
Now we combine all previous results with \cref{thm: Unf RPSF RS TAR} to get the following.

\begin{theorem}
    Let $\mathcal{R} = \ZZ[X]/(X^M + 1)$ where $M$ is a power-of-two, $q$ prime, $\alpha_{1}, \alpha_{2} \in (1, \infty), \epsilon \in (0,1/4), s \geq 2 \eta_{\epsilon}(\bm{\Lambda}_{\bm{h}, q})$, $\kappa \in \NN$, $\saltbits \in \NN$ and $\tau > 1$.
    For any adversary $\adv$ against the $\seufcra$ security of $\rsig$ (\cref{fig: rsig from rpsf}) instantiated with $\rpsf$ (\cref{fig:Gandalf RPSF}) with $\ntrusamplepre$ and $\ntrutrapgen$ making at most $\nrsignqueries$ signature and $\nrhashqueries$ \ac{QROM} queries to $\hash$, there exists adversaries $\bdv$ and $\cdv$ such that
    \begin{align*}
        \advantage{\seufcra}{\adv, \rsig, N, \nrsignqueries} & \leq\left(\left(\rpsfdomrenyi \right)^{\nrsignqueries}\left(\left(\rpsfprerenyi\right)^\nrsignqueries\cdot \epsilon'\right)^{\frac{\alpha_{2}-1}{\alpha_{2}}}\right)^{\frac{\alpha_{1} -1}{\alpha_{1}}}  + \frac{3\nrsignqueries}{2}\sqrt{\frac{\nrsignqueries + \nrhashqueries + 1}{2^\saltbits}}
    \end{align*}
    where $\epsilon'  = \advantage{\ntrursis}{\bdv, \kappa,q,\alpha, 2\beta} + \frac{20(\nrsignqueries + \nrhashqueries + 1)^2\nrsignqueries}{q^M} + (2\nrsignqueries + 2\nrhashqueries + 3)^2 \advantage{\ntrurisis}{\cdv, \kappa,q,\alpha, \beta} +\nrsignqueries \cdot \delta(\kappa)$    and $\delta(\kappa) \leq \tau^{(\kappa + 1)M}\cdot e^{\frac{(\kappa + 1)M}{2}(1 - \tau^2)}$.
\end{theorem}

When instantiating, explicit Rényi divergence analysis can be done using for example \cite[Lemma 9]{C:GajJanKil24} or optimized analysis of Rényi divergence as in \cite{C:GajJanKil24}.

\paragraph{Comparison to Classical Proof}
The \ac{QROM} result is very close to what can be expected when transitioning the classical results to the \ac{QROM}.
Guessing the index for inclusion comes at the additional multiplicative cost of all queries, as we need to rely on \cref{thm: measure and reprogram 2.0}.
The term for hash collisions suffers the same factor as we transition to the \ac{QROM} and is standard.
Finally, the tight adaptive reprogramming \cref{cor: adaptive reprogramming with other distribution} only introduces an additive term that can be controlled entirely by the saltbits $\saltbits$.
The remaining terms remain untouched, showing that the classical proof can be transferred to the \ac{QROM} without significant overhead. }


\begin{thebibliography}{10}
    \providecommand{\url}[1]{\texttt{#1}}
    \providecommand{\urlprefix}{URL }
    \providecommand{\doi}[1]{https://doi.org/#1}

    \bibitem{AC:AbeOhkSuz02}
    Abe, M., Ohkubo, M., Suzuki, K.: 1-out-of-n signatures from a variety of keys.
    In: Zheng, Y. (ed.) ASIACRYPT~2002. {LNCS}, vol.~2501, pp. 415--432.
    Springer, Berlin, Heidelberg (Dec 2002). \doi{10.1007/3-540-36178-2_26}

    \bibitem{EC:ABKM22}
    Alagic, G., Bai, C., Katz, J., Majenz, C.: Post-quantum security of the
        {Even}-{Mansour} cipher. In: Dunkelman, O., Dziembowski, S. (eds.)
    EUROCRYPT~2022, Part~III. {LNCS}, vol. 13277, pp. 458--487. Springer, Cham
    (May~/~Jun 2022). \doi{10.1007/978-3-031-07082-2_17}

    \bibitem{AC:BLLSS15}
    Bai, S., Langlois, A., Lepoint, T., Stehl{\'e}, D., Steinfeld, R.: Improved
    security proofs in lattice-based cryptography: Using the {R}{\'e}nyi
    divergence rather than the statistical distance. In: Iwata, T., Cheon, J.H.
    (eds.) ASIACRYPT~2015, Part~I. {LNCS}, vol.~9452, pp. 3--24. Springer,
    Berlin, Heidelberg (Nov~/~Dec 2015). \doi{10.1007/978-3-662-48797-6_1}

    \bibitem{banaszczyk1993new}
    Banaszczyk, W.: New bounds in some transference theorems in the geometry of
    numbers. Mathematische Annalen  \textbf{296}(1),  625--635 (1993)

    \bibitem{C:BBDDFG23}
    Barbosa, M., Barthe, G., Doczkal, C., Don, J., Fehr, S., Gr{\'e}goire, B.,
    Huang, Y.H., H{\"u}lsing, A., Lee, Y., Wu, X.: Fixing and mechanizing the
    security proof of {Fiat}-{Shamir} with aborts and {Dilithium}. In: Handschuh,
    H., Lysyanskaya, A. (eds.) CRYPTO~2023, Part~V. {LNCS}, vol. 14085, pp.
    358--389. Springer, Cham (Aug 2023). \doi{10.1007/978-3-031-38554-4_12}

    \bibitem{rfc9180}
    Barnes, R., Bhargavan, K., Lipp, B., Wood, C.A.: {Hybrid Public Key
    Encryption}. RFC 9180 (Feb 2022). \doi{10.17487/RFC9180},
    \url{https://www.rfc-editor.org/info/rfc9180}

    \bibitem{belovs2019quantum}
    Belovs, A.: Quantum algorithms for classical probability distributions. arXiv
    preprint arXiv:1904.02192  (2019)

    \bibitem{TCC:BenKatMor06}
    Bender, A., Katz, J., Morselli, R.: Ring signatures: Stronger definitions, and
    constructions without random oracles. In: Halevi, S., Rabin, T. (eds.)
    TCC~2006. {LNCS}, vol.~3876, pp. 60--79. Springer, Berlin, Heidelberg (Mar
    2006). \doi{10.1007/11681878_4}

    \bibitem{JC:BenKatMor09}
    Bender, A., Katz, J., Morselli, R.: Ring signatures: Stronger definitions, and
    constructions without random oracles. Journal of Cryptology  \textbf{22}(1),
    114--138 (Jan 2009). \doi{10.1007/s00145-007-9011-9}

    \bibitem{EPRINT:Bernstein22d}
    Bernstein, D.J.: Multi-ciphertext security degradation for lattices. Cryptology
    ePrint Archive, Report 2022/1580 (2022),
    \url{https://eprint.iacr.org/2022/1580}

    \bibitem{AC:BeuKatPin20}
    Beullens, W., Katsumata, S., Pintore, F.: {Calamari} and {Falafl}: Logarithmic
    (linkable) ring signatures from isogenies and lattices. In: Moriai, S., Wang,
    H. (eds.) ASIACRYPT~2020, Part~II. {LNCS}, vol. 12492, pp. 464--492.
    Springer, Cham (Dec 2020). \doi{10.1007/978-3-030-64834-3_16}

    \bibitem{AC:BDFLSZ11}
    Boneh, D., Dagdelen, {\"O}., Fischlin, M., Lehmann, A., Schaffner, C., Zhandry,
    M.: Random oracles in a quantum world. In: Lee, D.H., Wang, X. (eds.)
    ASIACRYPT~2011. {LNCS}, vol.~7073, pp. 41--69. Springer, Berlin, Heidelberg
    (Dec 2011). \doi{10.1007/978-3-642-25385-0_3}

    \bibitem{CiC:BorLaiLer24}
    Borin, G., Lai, Y.F., Leroux, A.: Erebor and durian: Full anonymous ring
    signatures from quaternions and isogenies. {CiC}  \textbf{1}(4), ~4 (2024).
    \doi{10.62056/ava3zivrzn}

    \bibitem{EPRINT:BraKal10}
    Brakerski, Z., Kalai, Y.T.: A framework for efficient signatures, ring
    signatures and identity based encryption in the standard model. Cryptology
    ePrint Archive, Report 2010/086 (2010),
    \url{https://eprint.iacr.org/2010/086}

    \bibitem{PKC:BFGJS22}
    Brendel, J., Fiedler, R., G{\"u}nther, F., Janson, C., Stebila, D.:
    Post-quantum asynchronous deniable key exchange and the {Signal} handshake.
    In: Hanaoka, G., Shikata, J., Watanabe, Y. (eds.) PKC~2022, Part~II. {LNCS},
    vol. 13178, pp. 3--34. Springer, Cham (Mar 2022).
    \doi{10.1007/978-3-030-97131-1_1}

    \bibitem{CCS:CDGORR17}
    Chase, M., Derler, D., Goldfeder, S., Orlandi, C., Ramacher, S., Rechberger,
    C., Slamanig, D., Zaverucha, G.: Post-quantum zero-knowledge and signatures
    from symmetric-key primitives. In: Thuraisingham, B.M., Evans, D., Malkin,
    T., Xu, D. (eds.) ACM CCS 2017. pp. 1825--1842. {ACM} Press (Oct~/~Nov 2017).
    \doi{10.1145/3133956.3133997}

    \bibitem{PKC:CCLM22}
    Chatterjee, R., Chung, K.M., Liang, X., Malavolta, G.: A note on the
    post-quantum security of (ring) signatures. In: Hanaoka, G., Shikata, J.,
    Watanabe, Y. (eds.) PKC~2022, Part~II. {LNCS}, vol. 13178, pp. 407--436.
    Springer, Cham (Mar 2022). \doi{10.1007/978-3-030-97131-1_14}

    \bibitem{C:CGHKLM21}
    Chatterjee, R., Garg, S., Hajiabadi, M., Khurana, D., Liang, X., Malavolta, G.,
    Pandey, O., Shiehian, S.: Compact ring signatures from learning with errors.
    In: Malkin, T., Peikert, C. (eds.) CRYPTO~2021, Part~I. {LNCS}, vol. 12825,
    pp. 282--312. Springer, Cham, Virtual Event (Aug 2021).
    \doi{10.1007/978-3-030-84242-0_11}

    \bibitem{TCC:ChiManSpo19}
    Chiesa, A., Manohar, P., Spooner, N.: Succinct arguments in the quantum random
    oracle model. In: Hofheinz, D., Rosen, A. (eds.) TCC~2019, Part~II. {LNCS},
    vol. 11892, pp. 1--29. Springer, Cham (Dec 2019).
    \doi{10.1007/978-3-030-36033-7_1}

    \bibitem{EC:CFHL21}
    Chung, K.M., Fehr, S., Huang, Y.H., Liao, T.N.: On the compressed-oracle
    technique, and post-quantum security of proofs of sequential work. In:
    Canteaut, A., Standaert, F.X. (eds.) EUROCRYPT~2021, Part~II. {LNCS}, vol.
    12697, pp. 598--629. Springer, Cham (Oct 2021).
    \doi{10.1007/978-3-030-77886-6_21}

    \bibitem{EPRINT:CHHHLY21}
    Chung, K.M., Hsieh, Y.C., Huang, M.Y., Huang, Y.H., Lange, T., Yang, B.Y.:
    Group signatures and accountable ring signatures from isogeny-based
    assumptions. Cryptology ePrint Archive, Report 2021/1368 (2021),
    \url{https://eprint.iacr.org/2021/1368}

    \bibitem{USENIX:CHNRV24}
    Collins, D., {Huguenin-Dumittan}, L., Nguyen, N.K., Rolin, N., Vaudenay, S.:
    K-waay: Fast and deniable post-quantum {X3DH} without ring signatures. In:
    Balzarotti, D., Xu, W. (eds.) USENIX Security 2024. {USENIX} Association (Aug
    2024),
    \url{https://www.usenix.org/conference/usenixsecurity24/presentation/collins}

    \bibitem{PQCRYPTO:DerRamSla18}
    Derler, D., Ramacher, S., Slamanig, D.: Post-quantum zero-knowledge proofs for
    accumulators with applications to ring signatures from symmetric-key
    primitives. In: Lange, T., Steinwandt, R. (eds.) Post-Quantum Cryptography -
    9th International Conference, PQCrypto 2018. pp. 419--440. Springer, Cham
    (2018). \doi{10.1007/978-3-319-79063-3_20}

    \bibitem{deutsch1992rapid}
    Deutsch, D., Jozsa, R.: Rapid solution of problems by quantum computation.
    Proceedings of the Royal Society of London. Series A: Mathematical and
    Physical Sciences  \textbf{439}(1907),  553--558 (1992)

    \bibitem{C:DFPS23}
    Devevey, J., Fallahpour, P., Passel{\`e}gue, A., Stehl{\'e}, D.: A detailed
    analysis of {Fiat}-{Shamir} with aborts. In: Handschuh, H., Lysyanskaya, A.
    (eds.) CRYPTO~2023, Part~V. {LNCS}, vol. 14085, pp. 327--357. Springer, Cham
    (Aug 2023). \doi{10.1007/978-3-031-38554-4_11}

    \bibitem{C:DonFehMaj20}
    Don, J., Fehr, S., Majenz, C.: The measure-and-reprogram technique 2.0:
    Multi-round {Fiat}-{Shamir} and more. In: Micciancio, D., Ristenpart, T.
    (eds.) CRYPTO~2020, Part~III. {LNCS}, vol. 12172, pp. 602--631. Springer,
    Cham (Aug 2020). \doi{10.1007/978-3-030-56877-1_21}

    \bibitem{C:DFMS19}
    Don, J., Fehr, S., Majenz, C., Schaffner, C.: Security of the {Fiat}-{Shamir}
    transformation in the quantum random-oracle model. In: Boldyreva, A.,
    Micciancio, D. (eds.) CRYPTO~2019, Part~II. {LNCS}, vol. 11693, pp. 356--383.
    Springer, Cham (Aug 2019). \doi{10.1007/978-3-030-26951-7_13}

    \bibitem{C:DFMS22}
    Don, J., Fehr, S., Majenz, C., Schaffner, C.: Efficient {NIZKs} and signatures
    from commit-and-open protocols in the {QROM}. In: Dodis, Y., Shrimpton, T.
    (eds.) CRYPTO~2022, Part~II. {LNCS}, vol. 13508, pp. 729--757. Springer, Cham
    (Aug 2022). \doi{10.1007/978-3-031-15979-4_25}

    \bibitem{EC:DFMS22}
    Don, J., Fehr, S., Majenz, C., Schaffner, C.: Online-extractability in the
    quantum random-oracle model. In: Dunkelman, O., Dziembowski, S. (eds.)
    EUROCRYPT~2022, Part~III. {LNCS}, vol. 13277, pp. 677--706. Springer, Cham
    (May~/~Jun 2022). \doi{10.1007/978-3-031-07082-2_24}

    \bibitem{van_Erven_2014}
    van Erven, T., Harremoes, P.: Rényi divergence and kullback-leibler
    divergence. IEEE Transactions on Information Theory  \textbf{60}(7),
    3797–3820 (Jul 2014). \doi{10.1109/tit.2014.2320500},
    \url{http://dx.doi.org/10.1109/TIT.2014.2320500}

    \bibitem{EPRINT:FalFehHua25}
    Fallahpour, P., Fehr, S., Huang, Y.H.: Tighter quantum security for
        {Fiat}-{Shamir}-with-aborts and hash-and-sign-with-retry signatures.
    Cryptology ePrint Archive, Report 2025/985 (2025),
    \url{https://eprint.iacr.org/2025/985}

    \bibitem{C:FiaSha86}
    Fiat, A., Shamir, A.: How to prove yourself: {Practical} solutions to
    identification and signature problems. In: Odlyzko, A.M. (ed.) CRYPTO'86.
    {LNCS}, vol.~263, pp. 186--194. Springer, Berlin, Heidelberg (Aug 1987).
    \doi{10.1007/3-540-47721-7_12}

    \bibitem{fouque2018falcon}
    Fouque, P.A., Hoffstein, J., Kirchner, P., Lyubashevsky, V., Pornin, T., Prest,
    T., Ricosset, T., Seiler, G., Whyte, W., Zhang, Z., et~al.: Falcon:
    Fast-fourier lattice-based compact signatures over ntru. Submission to the
    NIST’s post-quantum cryptography standardization process  \textbf{36}(5),
    1--75 (2018)

    \bibitem{EPRINT:GajJanKil24b}
    Gajland, P., Janneck, J., Kiltz, E.: A closer look at falcon. Cryptology ePrint
    Archive, Report 2024/1769 (2024), \url{https://eprint.iacr.org/2024/1769}

    \bibitem{C:GajJanKil24}
    Gajland, P., Janneck, J., Kiltz, E.: Ring signatures for deniable {AKEM}:
    Gandalf's fellowship. In: Reyzin, L., Stebila, D. (eds.) CRYPTO~2024, Part~I.
    {LNCS}, vol. 14920, pp. 305--338. Springer, Cham (Aug 2024).
    \doi{10.1007/978-3-031-68376-3_10}

    \bibitem{STOC:GenPeiVai08}
    Gentry, C., Peikert, C., Vaikuntanathan, V.: Trapdoors for hard lattices and
    new cryptographic constructions. In: Ladner, R.E., Dwork, C. (eds.) 40th ACM
    STOC. pp. 197--206. {ACM} Press (May 2008). \doi{10.1145/1374376.1374407}

    \bibitem{USENIX:GiaMadOrl16}
    Giacomelli, I., Madsen, J., Orlandi, C.: {ZKBoo}: Faster zero-knowledge for
        {Boolean} circuits. In: Holz, T., Savage, S. (eds.) USENIX Security 2016. pp.
    1069--1083. {USENIX} Association (Aug 2016),
    \url{https://www.usenix.org/conference/usenixsecurity16/technical-sessions/presentation/giacomelli}

    \bibitem{AC:GHHM21}
    Grilo, A.B., H{\"o}velmanns, K., H{\"u}lsing, A., Majenz, C.: Tight adaptive
    reprogramming in the {QROM}. In: Tibouchi, M., Wang, H. (eds.)
    ASIACRYPT~2021, Part~I. {LNCS}, vol. 13090, pp. 637--667. Springer, Cham (Dec
    2021). \doi{10.1007/978-3-030-92062-3_22}

    \bibitem{STOC:Grover96}
    Grover, L.K.: A fast quantum mechanical algorithm for database search. In: 28th
    ACM STOC. pp. 212--219. {ACM} Press (May 1996). \doi{10.1145/237814.237866}

    \bibitem{PKC:HaqSca20}
    Haque, A., Scafuro, A.: Threshold ring signatures: New definitions and
    post-quantum security. In: Kiayias, A., Kohlweiss, M., Wallden, P., Zikas, V.
    (eds.) PKC~2020, Part~II. {LNCS}, vol. 12111, pp. 423--452. Springer, Cham
    (May 2020). \doi{10.1007/978-3-030-45388-6_15}

    \bibitem{PKC:HKKP21}
    Hashimoto, K., Katsumata, S., Kwiatkowski, K., Prest, T.: An efficient and
    generic construction for {Signal}'s handshake ({X3DH}): Post-quantum, state
    leakage secure, and deniable. In: Garay, J. (ed.) PKC~2021, Part~II. {LNCS},
    vol. 12711, pp. 410--440. Springer, Cham (May 2021).
    \doi{10.1007/978-3-030-75248-4_15}

    \bibitem{USENIX2025:HKW}
    Hashimoto, K., Katsumata, S., Wiggers, T.: Bundled authenticated key exchange:
    A concrete treatment of (post-quantum) signal's handshake protocol.
    Cryptology {ePrint} Archive, Paper 2025/040 (2025),
    \url{https://eprint.iacr.org/2025/040}

    \bibitem{HofPipSil98}
    Hoffstein, J., Pipher, J., Silverman, J.H.: {NTRU:} {A} ring-based public key
    cryptosystem. In: Third Algorithmic Number Theory Symposium (ANTS). {LNCS},
    vol.~1423, pp. 267--288. Springer (Jun 1998)

    \bibitem{STOC:IKOS07}
    Ishai, Y., Kushilevitz, E., Ostrovsky, R., Sahai, A.: Zero-knowledge from
    secure multiparty computation. In: Johnson, D.S., Feige, U. (eds.) 39th ACM
    STOC. pp. 21--30. {ACM} Press (Jun 2007). \doi{10.1145/1250790.1250794}

    \bibitem{cryptoeprint:2025/1090}
    Katsumata, S., Niot, G., Tucker, I., Wiggers, T.: Comprehensive deniability
    analysis of signal handshake protocols: {X3DH}, {PQXDH} to fully post-quantum
    with deniable ring signatures. Cryptology {ePrint} Archive, Paper 2025/1090
    (2025), \url{https://eprint.iacr.org/2025/1090}

    \bibitem{Signal:PQXDH}
    Kret, E., Schmidt, R.: The pqxdh key agreement protocol. Online (2023),
    \url{https://signal.org/docs/specifications/pqxdh/}

    \bibitem{EC:LanSteSte14}
    Langlois, A., Stehl{\'e}, D., Steinfeld, R.: {GGHLite}: More efficient
    multilinear maps from ideal lattices. In: Nguyen, P.Q., Oswald, E. (eds.)
    EUROCRYPT~2014. {LNCS}, vol.~8441, pp. 239--256. Springer, Berlin, Heidelberg
    (May 2014). \doi{10.1007/978-3-642-55220-5_14}

    \bibitem{qrom_falcon}
    Li, J., Jiang, H., Wang, H., Duan, Q.: Tighter security proof of falcon in the
    quantum random oracle model. In: Chen, R., Deng, R.H., Yung, M. (eds.)
    Information Security and Cryptology. pp. 142--165. Springer Nature Singapore,
    Singapore (2026)

    \bibitem{ACNS:LuAuZha19}
    Lu, X., Au, M.H., Zhang, Z.: Raptor: {A} practical lattice-based (linkable)
    ring signature. In: Deng, R.H., {Gauthier-Uma{\~n}a}, V., Ochoa, M., Yung, M.
    (eds.) ACNS 2019. {LNCS}, vol. 11464, pp. 110--130. Springer, Cham (Jun
    2019). \doi{10.1007/978-3-030-21568-2_6}

    \bibitem{EC:Lyubashevsky12}
    Lyubashevsky, V.: Lattice signatures without trapdoors. In: Pointcheval, D.,
    Johansson, T. (eds.) EUROCRYPT~2012. {LNCS}, vol.~7237, pp. 738--755.
    Springer, Berlin, Heidelberg (Apr 2012). \doi{10.1007/978-3-642-29011-4_43}

    \bibitem{C:LyuNguSei21}
    Lyubashevsky, V., Nguyen, N.K., Seiler, G.: {SMILE}: Set membership from ideal
    lattices with applications to ring signatures and confidential transactions.
    In: Malkin, T., Peikert, C. (eds.) CRYPTO~2021, Part~II. {LNCS}, vol. 12826,
    pp. 611--640. Springer, Cham, Virtual Event (Aug 2021).
    \doi{10.1007/978-3-030-84245-1_21}

    \bibitem{Signal:X3DH}
    Marlinspike, M., Perrin, T.: The x3dh key agreement protocol. Online (2016),
    \url{https://signal.org/docs/specifications/x3dh/}

    \bibitem{FOCS:MicReg04}
    Micciancio, D., Regev, O.: Worst-case to average-case reductions based on
        {Gaussian} measures. In: 45th FOCS. pp. 372--381. {IEEE} Computer Society
    Press (Oct 2004). \doi{10.1109/FOCS.2004.72}

    \bibitem{mittalring}
    Mittal, B.: Ring trapdoor functions: A lattice-based framework for secure ring
    signatures,
    \url{https://bhumikamittal.in/assets/pdf/undergrad_thesis_bhumika.pdf}

    \bibitem{nielsen2010quantum}
    Nielsen, M.A., Chuang, I.L.: Quantum computation and quantum information.
    Cambridge university press (2010)

    \bibitem{cryptoeprint:2025/853}
    Niot, G.: Practical deniable post-quantum {X3DH}: A lightweight split-{KEM} for
    k-waay. Cryptology {ePrint} Archive, Paper 2025/853 (2025),
    \url{https://eprint.iacr.org/2025/853}

    \bibitem{TCC:PeiRos06}
    Peikert, C., Rosen, A.: Efficient collision-resistant hashing from worst-case
    assumptions on cyclic lattices. In: Halevi, S., Rabin, T. (eds.) TCC~2006.
    {LNCS}, vol.~3876, pp. 145--166. Springer, Berlin, Heidelberg (Mar 2006).
    \doi{10.1007/11681878_8}

    \bibitem{CHES:PopDucGun14}
    P{\"o}ppelmann, T., Ducas, L., G{\"u}neysu, T.: Enhanced lattice-based
    signatures on reconfigurable hardware. In: Batina, L., Robshaw, M. (eds.)
    CHES~2014. {LNCS}, vol.~8731, pp. 353--370. Springer, Berlin, Heidelberg (Sep
    2014). \doi{10.1007/978-3-662-44709-3_20}

    \bibitem{prestthesis}
    Prest, T.: {Gaussian sampling in lattice-based cryptography}. Theses, {Ecole
            normale sup{\'e}rieure - ENS PARIS} (Dec 2015),
    \url{https://theses.hal.science/tel-01245066}

    \bibitem{AC:RivShaTau01}
    Rivest, R.L., Shamir, A., Tauman, Y.: How to leak a secret. In: Boyd, C. (ed.)
    ASIACRYPT~2001. {LNCS}, vol.~2248, pp. 552--565. Springer, Berlin, Heidelberg
    (Dec 2001). \doi{10.1007/3-540-45682-1_32}

    \bibitem{EC:SteSte11}
    Stehl{\'e}, D., Steinfeld, R.: Making {NTRU} as secure as worst-case problems
    over ideal lattices. In: Paterson, K.G. (ed.) EUROCRYPT~2011. {LNCS},
    vol.~6632, pp. 27--47. Springer, Berlin, Heidelberg (May 2011).
    \doi{10.1007/978-3-642-20465-4_4}

    \bibitem{ICICS:WanSun11}
    Wang, J., Sun, B.: Ring signature schemes from lattice basis delegation. In:
    Qing, S., Susilo, W., Wang, G., Liu, D. (eds.) ICICS 11. {LNCS}, vol.~7043,
    pp. 15--28. Springer, Berlin, Heidelberg (Nov 2011).
    \doi{10.1007/978-3-642-25243-3_2}

    \bibitem{C:YELAD21}
    Yuen, T.H., Esgin, M.F., Liu, J.K., Au, M.H., Ding, Z.: {DualRing}: Generic
    construction of ring signatures with efficient instantiations. In: Malkin,
    T., Peikert, C. (eds.) CRYPTO~2021, Part~I. {LNCS}, vol. 12825, pp. 251--281.
    Springer, Cham, Virtual Event (Aug 2021). \doi{10.1007/978-3-030-84242-0_10}

    \bibitem{FOCS:Zhandry12}
    Zhandry, M.: How to construct quantum random functions. In: 53rd FOCS. pp.
    679--687. {IEEE} Computer Society Press (Oct 2012).
    \doi{10.1109/FOCS.2012.37}

    \bibitem{C:Zhandry12}
    Zhandry, M.: Secure identity-based encryption in the quantum random oracle
    model. In: Safavi-Naini, R., Canetti, R. (eds.) CRYPTO~2012. {LNCS},
    vol.~7417, pp. 758--775. Springer, Berlin, Heidelberg (Aug 2012).
    \doi{10.1007/978-3-642-32009-5_44}

    \bibitem{C:Zhandry19}
    Zhandry, M.: How to record quantum queries, and applications to quantum
    indifferentiability. In: Boldyreva, A., Micciancio, D. (eds.) CRYPTO~2019,
    Part~II. {LNCS}, vol. 11693, pp. 239--268. Springer, Cham (Aug 2019).
    \doi{10.1007/978-3-030-26951-7_9}

    \bibitem{zhang2025pegasus}
    Zhang, X., Li, Z., Steinfeld, R., Zhao, R.K., Liu, J.K., Yuen, T.H.: Pegasus
    and pegaring: Efficient (ring) signatures from sigma-protocols for power
    residue prfs with (q) rom security. Cryptology ePrint Archive  (2025)

\end{thebibliography}
\end{document}